\documentclass{article}

\usepackage{pdfpages}
\usepackage{fullpage}
\usepackage{textcomp}
\usepackage{algorithm}
\usepackage{algorithmicx}
\usepackage{algpseudocode}
\usepackage{xcolor}
\usepackage{multirow}
\usepackage{amsmath}
\usepackage{amsfonts}
\usepackage{graphicx}
\usepackage{amssymb}
\usepackage{amsthm}

\newtheorem{lemma}{Lemma}
\newtheorem{theorem}[lemma]{Theorem}
\newtheorem{corollary}[lemma]{Corollary}
\newtheorem{definition}[lemma]{Definition}

\def\eps{\varepsilon}

\title{Dynamic geometric set cover and hitting set}

\author{Pankaj K. Agarwal \footnote{Department of Computer Science, Duke University, USA} \\ \texttt{pankaj@cs.duke.edu}
\and Hsien-Chih Chang \footnotemark[1] \\ \texttt{hsienchih.chang@duke.edu}
\and Subhash Suri \footnote{Department of Computer Science, University of California at Santa Barbara, USA} \\ \texttt{suri@cs.ucsb.edu}
\and Allen Xiao \footnotemark[1] \\ \texttt{axiao@cs.duke.edu}
\and Jie Xue \footnotemark[2] \\ \texttt{jiexue@ucsb.edu}}

\date{}

\begin{document}

\maketitle

\begin{abstract}
We investigate dynamic versions of geometric set cover and hitting set where points and ranges may be inserted or deleted, and we want to efficiently maintain an (approximately) optimal solution for the current problem instance.
While their static versions have been extensively studied in the past, surprisingly little is known about dynamic geometric set cover and hitting set.
For instance, even for the most basic case of one-dimensional interval set cover and hitting set, no nontrivial results were known.
The main contribution of our paper are two frameworks that lead to efficient data structures for dynamically maintaining set covers and hitting sets in $\mathbb{R}^1$ and $\mathbb{R}^2$.
The first framework uses bootstrapping and gives a $(1+\varepsilon)$-approximate data structure for dynamic interval set cover in $\mathbb{R}^1$ with $O(n^\alpha/\varepsilon)$ amortized update time for any constant $\alpha > 0$; in $\mathbb{R}^2$, this method gives $O(1)$-approximate data structures for unit-square (and  quadrant) set cover and hitting set with $O(n^{1/2+\alpha})$ amortized update time.
The second framework uses local modification, and leads to a $(1+\varepsilon)$-approximate data structure for dynamic interval hitting set in $\mathbb{R}^1$ with $\widetilde{O}(1/\varepsilon)$ amortized update time; in $\mathbb{R}^2$, it gives $O(1)$-approximate data structures for unit-square (and quadrant) set cover and hitting set in the \textit{partially} dynamic settings with $\widetilde{O}(1)$ amortized update time.
\end{abstract}

\section{Introduction}
Given a pair $(S,\mathcal{R})$ where $S$ is a set of points and $\mathcal{R}$ is a collection of geometric ranges in a Euclidean space, the \textit{geometric set cover} (resp., \textit{hitting set}) problem is to find the smallest number of ranges in $\mathcal{R}$ (resp., points in $S$) that cover all points in $S$ (resp., hit all ranges in $\mathcal{R}$).
Geometric set cover and hitting set are classical geometric optimization problems, with numerous applications in databases, sensor networks, VLSI design, etc.

In many applications, the problem instance can change over time and re-computing a new solution after each change is too costly.
In these situations, a dynamic algorithm that can update the solution after a change more efficiently than constructing the entire new solution from scratch is highly desirable.
This motivates the main problem studied in our paper: dynamically maintaining geometric set covers and hitting sets under insertion and deletion of points and ranges.

Although (static) geometric set cover and hitting set have been extensively studied over the years, their dynamic variants are surprisingly open.
For example, even for the most fundamental case, dynamic \emph{interval} set cover and hitting set in one dimension, no nontrivial results were previously known.
In this paper, we propose two algorithmic frameworks for the problems, which lead to efficient data structures for dynamic set cover and hitting set for intervals in $\mathbb{R}^1$ and unit squares and quadrants in $\mathbb{R}^2$.
We believe that our approaches can be extended to solve dynamic set cover and hitting set in other geometric settings, or more generally, other dynamic problems in computational geometry.

\subsection{Related work} \label{sec-related}
The set cover and hitting set problems in general setting are well-known to be NP-complete~\cite{hartmanis1982computers}.
A simple greedy algorithm achieves an $O(\log n)$-approximation~\cite{chvatal1979greedy,johnson1974approximation,lovasz1975ratio}, which is tight under appropriate complexity-theoretic assumptions~\cite{dinur2014analytical,khot2008vertex}.
In many geometric settings, the problems remain NP-hard or even hard to approximate~\cite{berman1997complexities,megiddo1984complexity,megiddo1982complexity}.
However, by exploiting the geometric nature of the problems, efficient algorithms with better approximation factors can
be obtained.
For example, Mustafa and Ray~\cite{mustafa2010improved} showed the existence of polynomial-time approximation schemes (PTAS) for halfspace hitting set in $\mathbb{R}^3$ and disk hitting set.
There is also a PTAS for unit-square set cover given by Erlebach and van Leeuwen~\cite{erlebach2010ptas}.
Agarwal and Pan~\cite{agarwal2014near} proposed approximation algorithms with near-linear running time for the set cover and hitting set problems for halfspaces in $\mathbb{R}^3$, disks in $\mathbb{R}^2$, and orthogonal rectangles.

Dynamic problems have received a considerable attention in recent years~\cite{baswana2015fully,bernstein2016faster,bhattacharya2017deterministic,bhattacharya2018deterministic,bhattacharya2016new,gupta2013fully,neiman2016simple,solomon2016fully}.
In particular, dynamic set cover in general setting has been studied in~\cite{abboud2019dynamic,bhattacharya2015design,gupta2017online}.
All the results were achieved in the partially dynamic setting where ranges are fixed and only points are dynamic.
Gupta et al.~\cite{gupta2017online} showed that an $O(\log n)$-approximation can be maintained using $O(f \log n)$ amortized update time and an $O(f^3)$-approximation can be maintained using $O(f^2)$ amortized update time, where $f$ is the maximum number of ranges that a point belongs to.
Bhattacharya et al.~\cite{bhattacharya2015design} gave an $O(f^2)$-approximation data structure for dynamic set cover with $O(f \log n)$ amortized update time.
Abboud et al.~\cite{abboud2019dynamic} proved that one can maintain a $(1+\varepsilon)f$-approximation using $O(f^2 \log n/\varepsilon^5)$ amortized update time.

In geometric settings, only the dynamic hitting set problem has been considered~\cite{ganjugunte2011geometric}.
Ganjugunte~\cite{ganjugunte2011geometric} studied two different dynamic settings: \textbf{(i)} only the range set $\mathcal{R}$ is dynamic and \textbf{(ii)} $\mathcal{R}$ is dynamic and $S$ is semi-dynamic (i.e., insertion-only).
Ganjugunte~\cite{ganjugunte2011geometric} showed that, for pseudo-disks in $\mathbb{R}^2$, dynamic hitting set in setting \textbf{(i)} can be solved using $O(\gamma(n) \log^4 n)$ amortized update time with approximation factor $O(\log^2 n)$, and that in setting \textbf{(ii)} can be solved using $O(\gamma(n) \sqrt{n} \log^4 n)$ amortized update time with approximation factor $O(\log^6 n/\log\log n)$,
where $\gamma(n)$ is the time for finding a point in $X$ contained in a query pseudo-trapezoid (see~\cite{ganjugunte2011geometric} for details).
Dynamic geometric hitting set in the fully dynamic setting (where both points and ranges can be inserted or deleted) as well as dynamic geometric set cover has not yet been studied before, to the best of our knowledge.


\subsection{Our results}
Let $(S,\mathcal{R})$ be a dynamic geometric set cover (resp., hitting set) instance.
We are interested in proposing efficient data structures to \textit{maintain} an (approximately) optimal solution for $(S,\mathcal{R})$.
This may have various definitions, resulting in different variants of the problem.
A natural variant is to maintain the number $\mathsf{opt}$, which is the size of an optimal set cover (resp., hitting set) for $(S,\mathcal{R})$, or an approximation of $\mathsf{opt}$.
However, in many applications, only maintaining the optimum is not satisfactory and one may hope to maintain a ``real'' set cover (resp., hitting set) for the dynamic instance.
Therefore, in this paper, we formulate the problem as follows.
We require the data structure to, after each update, store (implicitly) a solution for the current problem instance satisfying some quality requirement such that certain information about the solution can be queried efficiently.
For example, one can ask how large the solution is, whether a specific element in $\mathcal{R}$ (resp., $S$) is used in the solution, what the entire solution is, etc.
We will make this more precise shortly.

In dynamic settings, it is usually more natural and convenient to consider a \textit{multiset} solution for the set cover (resp., hitting set) instance.
That is, we allow the solution to be a \textit{multiset} of elements in $\mathcal{R}$ (resp., $S$) that cover all points in $S$ (resp., hit all ranges in $\mathcal{R}$), and the quality of the solution is also evaluated in terms of the multiset cardinality.
In the static problems, one can always efficiently remove the duplicates in a multiset solution to obtain a (ordinary) set cover or hitting set with even better quality (i.e., smaller cardinality), hence computing a multiset solution is essentially equivalent to computing an ordinary solution.
However, in dynamic settings, the update time has to be sublinear, and in general this is not sufficient for detecting and removing duplicates.
Therefore, in this paper, we mainly focus on multiset solutions (though some of our data structures can maintain an ordinary set cover or hitting set).
Unless explicitly mentioned otherwise, solutions for set cover and hitting set always refer to multiset solutions hereafter.

Precisely, we require a dynamic set cover (resp., hitting set) data structure to store implicitly, after each update, a set cover $\mathcal{R}'$ (resp., a hitting set $S'$) for the current instance $(S,\mathcal{R})$ such that the following queries are supported.
\begin{itemize}\itemsep=0pt
\item
\textbf{Size query}: reporting the (multiset) size of $\mathcal{R}'$ (resp., $S'$).
\item \textbf{Membership query}: reporting, for a given range $R \in \mathcal{R}$ (resp., a given point $a \in S$), the number of copies of $R$ (resp., $a$) contained in $\mathcal{R}'$ (resp., $S'$).
\item \textbf{Reporting query}: reporting all the elements in $\mathcal{R}'$ (resp., $S'$).
\end{itemize}
We require the size query to be answered in $O(1)$ time, a membership query to be answered in $O(\log |\mathcal{R}'|)$ time (resp., $O(\log |S'|)$ time), and the reporting query to be answered in $O(|\mathcal{R}'|)$ time (resp., $O(|S'|)$ time); this is the best one can expect in the pointer machine model.

We say that a set cover (resp., hitting set) instance is \textit{fully dynamic} if insertions and deletions on both points and ranges are allowed, and \textit{partially dynamic} if only the points (resp., ranges) can be inserted and deleted.
This paper mainly focuses on the fully dynamic setting, while some results are achieved in the partially dynamic setting.
Thus, unless explicitly mentioned otherwise, problems are always considered in the fully dynamic setting.

The main contribution of this paper are two frameworks for designing dynamic geometric set cover and hitting set data structures, leading to efficient data structures in $\mathbb{R}^1$ and $\mathbb{R}^2$ (see Table~\ref{tab-results}).
The first framework is based on bootstrapping, which results in efficient (approximate) dynamic data structures for interval set cover, quadrant/unit-square set cover and hitting set (see the first three rows of Table~\ref{tab-results} for detailed bounds).
The second framework is based on local modification, which results in efficient (approximate) dynamic data structures for interval hitting set, quadrant/unit-square set cover and hitting set in the \textit{partially} dynamic setting (see the last three rows of Table~\ref{tab-results} for detailed bounds).

\renewcommand\arraystretch{1.3}

\begin{table}[h]
    \centering
    \begin{tabular}{|c|c|c|c|c|c|}
        \hline
        \textbf{Framework} & \textbf{Problem} & \textbf{Range} & \textbf{Approx. factor} & \textbf{Update time} & \textbf{Setting} \\
        \hline
        \multirow{3}{*}{Bootstrapping} & SC & Interval & $1+\varepsilon$ & $\widetilde{O}(n^\alpha/\varepsilon)$ & Fully dynamic \\
        \cline{2-6}
        & SC \& HS & Quadrant & $O(1)$ & $\widetilde{O}(n^{1/2+\alpha})$ & Fully dynamic \\
        \cline{2-6}
        & SC \& HS & Unit square & $O(1)$ & $\widetilde{O}(n^{1/2+\alpha})$ & Fully dynamic \\
        \hline
        \multirow{3}{*}{Local modification} & HS & Interval & $1+\varepsilon$ & $\widetilde{O}(1/\varepsilon)$ & Fully dynamic \\
        \cline{2-6}
        & SC \& HS & Quadrant & $O(1)$ & $\widetilde{O}(1)$ & Part. dynamic \\
        \cline{2-6}
        & SC \& HS & Unit square & $O(1)$ & $\widetilde{O}(1)$ & Part. dynamic \\
        \hline
    \end{tabular}
    \vspace{2mm}
    \caption{Summary of our results for dynamic geometric set cover and hitting set (SC = set cover and HS = hitting set). All update times are amortized. The notation $\widetilde{O}(\cdot)$ hides logarithmic factors, $n$ is the size of the current instance, and $\alpha > 0$ is any small constant. All data structures can be constructed in $\widetilde{O}(n_0)$ time where $n_0$ is the size of the initial instance.}
    \label{tab-results}
\end{table}

\paragraph{Organization.}
The rest of the paper is organized as follows.
Section~\ref{sec-notation} gives the preliminaries required for the paper and Section~\ref{sec-overview} gives an overview of our two frameworks.
Our first framework (bootstrapping) and second framework (local modification) for designing efficient dynamic geometric set cover and hitting set data structures are presented in Section~\ref{sec-bootstrapping} and Section~\ref{sec-localmod}, respectively.
To make the paper more readable, the proofs of the technical lemmas and some details are deferred to the appendix.

\section{Preliminaries} \label{sec-notation}
In this section, we introduce the basic notions used throughout the paper.

\paragraph{Multi-sets and disjoint union.}
A \textit{multi-set} is a set in which elements can have multiple copies.
The \textit{multiplicity} of an element $a$ in a multi-set $A$ is the number of the copies of $a$ in $A$.
For two multi-sets $A$ and $B$, we use $A \sqcup B$ to denote the \textit{disjoint union} of $A$ and $B$, in which the multiplicity of an element $a$ is the sum of the multiplicities of $a$ in $A$ and $B$.

\paragraph{Basic data structures.}
A data structure built on a dataset (e.g., point set, range set, set cover or hitting set instances) of size $n$ is \textit{basic} if it can be constructed in $\widetilde{O}(n)$ time and can be dynamized with $\widetilde{O}(1)$ update time (with a bit abuse of terminology, sometimes we also use ``basic data structures'' to denote the dynamized version of such data structures).

\paragraph{Output-sensitive algorithms.}
In some set cover and hitting set problems, if the problem instance is properly stored in some data structure, it is possible to compute an (approximate) optimal solution in sub-linear time.
An \textit{output-sensitive} algorithm for a set cover or hitting set problem refers to an algorithm that can compute an (approximate) optimal solution in $\widetilde{O}(\mathsf{out})$ time (where $\mathsf{out}$ is the size of the output solution), by using some \textit{basic} data structure built on the problem instance.

\section{An overview of our two frameworks} \label{sec-overview}
The basic idea of our first framework is \textit{bootstrapping}.
Namely, we begin from a simple inefficient dynamic set cover or hitting set data structure (e.g., a data structure that re-computes a solution after each update), and repeatedly use the current data structure to obtain an improved one.
The main challenge here is to design the bootstrapping procedure: how to use a given data structure to construct a new data structure with improved update time.
We achieve this by using output-sensitive algorithms and carefully partitioning the problem instances to sub-instances.

Our second framework is much simpler, which is based on \textit{local modification}.
Namely, we construct a new solution by slightly modifying the previous one after each update, and re-compute a new solution periodically using an output-sensitive algorithm.
This framework applies to the problems which are \textit{stable}, in the sense that the optimum of a dynamic instance does not change significantly.

\section{First framework: Bootstrapping} \label{sec-bootstrapping}
In this section, we present our first frame work for dynamic geometric set cover and hitting set, which is based on bootstrapping and results in sub-linear data structures for dynamic interval set cover and dynamic quadrant and unit-square set cover (resp., hitting set).

\subsection{Warm-up: 1D set cover for intervals} \label{sec-DISC}
As a warm up, we first study the 1D problem: dynamic interval set cover.
First, we observe that interval set cover admits a simple \textit{exact} output-sensitive algorithm.
Indeed, interval set cover can be solved using the greedy algorithm that repeatedly picks the leftmost uncovered point and covers it using the interval with the rightmost right endpoint, and the algorithm can be easily made output-sensitive if we store the points and intervals in binary search trees.
\begin{lemma} \label{lem-osint}
Interval set cover admits an exact output-sensitive algorithm.
\end{lemma}
\noindent
This algorithm will serve an important role in the design of our data structure.

\subsubsection{Bootstrapping}
As mentioned before, our data structure is designed using bootstrapping.
Specifically, we prove the following bootstrapping theorem, which is the technical heart of our result.
The theorem roughly states that given a dynamic interval set cover data structure, one can obtain another dynamic interval set cover data structure with improved update time.
\begin{theorem} \label{thm-boot}
Let $\alpha \in [0,1]$ be a number.
If there exists a $(1+\eps)$-approximate dynamic interval set cover data structure $\mathcal{D}_\textnormal{old}$ with $\widetilde{O}(n^\alpha/\eps^{1-\alpha})$ amortized update time and $\widetilde{O}(n_0)$ construction time for any $\eps>0$,
then there exists a $(1+\eps)$-approximate dynamic interval set cover data structure $\mathcal{D}_\textnormal{new}$ with $\widetilde{O}(n^{\alpha'}/\eps^{1-\alpha'})$ amortized update time and $\widetilde{O}(n_0)$ construction time for any $\eps>0$, where $\alpha' = \alpha/(1+\alpha)$.
Here $n$ (resp., $n_0$) denotes the size of the current (resp., initial) problem instance.
\end{theorem}


Assuming the existence of $\mathcal{D}_\text{old}$ as in the theorem, we are going to design the improved data structure $\mathcal{D}_\text{new}$.
Let $(S,\mathcal{I})$ be a dynamic interval set cover instance, and $\varepsilon>0$ be the approximation factor.
We denote by $n$ (resp., $n_0$) the size of the current (resp., initial) $(S,\mathcal{I})$.

\paragraph{The construction of $\mathcal{D}_\text{new}$.}
Initially, $|S|+|\mathcal{I}| = n_0$.
Essentially, our data structure $\mathcal{D}_\text{new}$ consists of two parts\footnote{In implementation level, we may need some additional support data structures (which are very simple). For simplicity of exposition, we shall mention them when discussing the implementation details.}.
The first part is the basic data structure $\mathcal{A}$ required for the output-sensitive algorithm of Lemma~\ref{lem-osint}.
The second part is a family of $\mathcal{D}_\text{old}$ data structures defined as follows.
Let $f$ be a function to be determined shortly.
We partition the real line $\mathbb{R}$ into $r = \lceil n_0/f(n_0,\eps) \rceil$ connected portions (i.e., intervals) $J_1,\dots,J_r$ such that each portion $J_i$ contains $O(f(n_0,\eps))$ points in $S$ and $O(f(n_0,\eps))$ endpoints of the intervals in $\mathcal{I}$.
Define $S_i = S \cap J_i$ and define $\mathcal{I}_i \subseteq \mathcal{I}$ as the sub-collections consisting of the intervals that ``partially intersect'' $J_i$, i.e., $\mathcal{I}_i = \{I \in \mathcal{I}: J_i \cap I \neq \emptyset \text{ and } J_i \nsubseteq I\}$.
When the instance $(S,\mathcal{I})$ is updated, the partition $J_1,\dots,J_r$ will remain unchanged, but the $S_i$'s and $\mathcal{I}_i$'s will change along with $S$ and $\mathcal{I}$.
We view each $(S_i,\mathcal{I}_i)$ as a dynamic interval set cover instance, and let $\mathcal{D}_\text{old}^{(i)}$ be the data structure $\mathcal{D}_\text{old}$ built on $(S_i,\mathcal{I}_i)$ for the approximation parameter $\tilde{\eps} = \eps/2$.
Thus, $\mathcal{D}_\text{old}^{(i)}$ maintains a $(1+\tilde{\eps})$-approximate optimal set cover for $(S_i,\mathcal{I}_i)$.
The second part of $\mathcal{D}_\text{new}$ consists of the data structures $\mathcal{D}_\text{old}^{(1)},\dots,\mathcal{D}_\text{old}^{(r)}$.

\paragraph{Update and reconstruction.}
After an operation on $(S,\mathcal{I})$, we update the basic data structure $\mathcal{A}$.
Also, we update the data structure $\mathcal{D}_\text{old}^{(i)}$ if the instance $(S_i,\mathcal{I}_i)$ changes due to the operation.
Note that an operation on $S$ changes exactly one $S_i$ and an operation on $\mathcal{I}$ changes at most two $\mathcal{I}_i$'s (because an interval can belong to at most two $\mathcal{I}_i$'s).
Thus, we in fact only need to update at most two $\mathcal{D}_\text{old}^{(i)}$'s.
Besides the update, we also reconstruct the entire data structure $\mathcal{D}_\text{new}$ periodically (as is the case for many dynamic data structures).
Specifically, the first reconstruction of $\mathcal{D}_\text{new}$ happens after processing $f(n_0,\eps)$ operations.
The reconstruction is the same as the initial construction of $\mathcal{D}_\text{new}$, except that $n_0$ is replaced with $n_1$, the size of $(S,\mathcal{I})$ at the time of reconstruction.
Then the second reconstruction happens after processing $f(n_1,\eps)$ operations since the first reconstruction, and so forth.

\paragraph{Maintaining a solution.}
We now discuss how to maintain a $(1+\varepsilon)$-approximate optimal set cover $\mathcal{I}_\text{appx}$ for $(S,\mathcal{I})$.
Let $\mathsf{opt}$ denote the optimum (i.e., the size of an optimal set cover) of the current $(S,\mathcal{I})$.
Set $\delta = \min\{(6+2\eps) \cdot r/\eps, n\}$.
If $\mathsf{opt} \leq \delta$, then the output-sensitive algorithm can compute an optimal set cover for $(S,\mathcal{I})$ in $\widetilde{O}(\delta)$ time.
Thus, we simulate the output-sensitive algorithm within that amount of time.
If the algorithm successfully computes a solution, we use it as our $\mathcal{I}_\text{appx}$.
Otherwise, we construct $\mathcal{I}_\text{appx}$ as follows.
For $i \in \{1,\dots,r\}$, we say $J_i$ is \textit{coverable} if there exists $I \in \mathcal{I}$ such that $J_i \subseteq I$ and \textit{uncoverable} otherwise.
Let $P = \{i : J_i \text{ is coverable}\}$ and $P' = \{i : J_i \text{ is uncoverable}\}$.
We try to use the intervals in $\mathcal{I}$ to ``cover'' all coverable portions.
That is, for each $i \in P$, we find an interval in $\mathcal{I}$ that contains $J_i$, and denote by $\mathcal{I}^*$ the collection of these intervals.
Then we consider the uncoverable portions.
If for some $i \in P'$, the data structure $\mathcal{D}_\text{old}^{(i)}$ tells us that the current $(S_i,\mathcal{I}_i)$ does not have a set cover, then we immediately make a no-solution decision, i.e., decide that the current $(S,\mathcal{I})$ has no feasible set cover, and continue to the next operation.
Otherwise, for every $i \in P'$, the data structure $\mathcal{D}_\text{old}^{(i)}$ maintains a $(1+\tilde{\eps})$-approximate optimal set cover $\mathcal{I}_i^*$ for $(S_i,\mathcal{I}_i)$.
We then define $\mathcal{I}_\text{appx} = \mathcal{I}^* \sqcup \left( \bigsqcup_{J_i \in \mathcal{P}'} \mathcal{I}_i^* \right)$.

Later we will prove that $\mathcal{I}_\text{appx}$ is always a $(1+\varepsilon)$-approximate optimal set cover for $(S,\mathcal{I})$.
Before this, let us consider how to store $\mathcal{I}_\text{appx}$ properly to support the size, membership, and reporting queries in the required query times.
If $\mathcal{I}_\text{appx}$ is computed by the output-sensitive algorithm, then the size of $\mathcal{I}_\text{appx}$ is at most $\delta$, and we have all the elements of $\mathcal{I}_\text{appx}$ in hand.
In this case, it is not difficult to build a data structure on $\mathcal{I}_\text{appx}$ to support the desired queries.
On the other hand, if $\mathcal{I}_\text{appx}$ is defined as the disjoint union of $\mathcal{I}^*$ and $\mathcal{I}_i^*$'s, the size of $\mathcal{I}_\text{appx}$ might be very large and thus we are not able to explicitly extract all elements of $\mathcal{I}_\text{appx}$.
Fortunately, in this case, each $\mathcal{I}_i^*$ is already maintained in the data structure $\mathcal{D}_\text{old}^{(i)}$.
Therefore, we actually only need to compute $P$, $P'$, and $\mathcal{I}^*$; with these in hand, one can already build a data structure to support the desired queries for $\mathcal{I}_\text{appx}$.
We defer the detailed discussion to Appendix~\ref{appx-anaint}.

\paragraph{Correctness.}
We now prove the correctness of our data structure $\mathcal{D}_\text{new}$.
We first show the correctness of the no-solution decision.
\begin{lemma} \label{lem-intnoans}
$\mathcal{D}_\textnormal{new}$ makes a no-solution decision iff the current $(S,\mathcal{I})$ has no set cover.
\end{lemma}
Next, we show that the solution $\mathcal{I}_\text{appx}$ maintained by $\mathcal{D}_\text{new}$ is truly a $(1+\varepsilon)$-approximate optimal set cover for $(S,\mathcal{I})$.
If $\mathcal{I}_\text{appx}$ is computed by the output-sensitive algorithm, then it is an optimal set cover for $(S,\mathcal{I})$.
Otherwise, $\mathsf{opt} > \delta = \min\{(6+2\eps) \cdot r/\eps,n\}$, i.e., either $\mathsf{opt} > (6+2\eps) \cdot r/\eps$ or $\mathsf{opt} > n$.
If $\mathsf{opt} > n$, then the current $(S,\mathcal{I})$ has no set cover (i.e., $\mathsf{opt} = \infty$) and thus $\mathcal{D}_\text{new}$ makes a no-solution decision by Lemma~\ref{lem-intnoans}.
So assume $\mathsf{opt} > (6+2\eps) \cdot r/\eps$.
In this case, $\mathcal{I}_\text{appx} = \mathcal{I}^* \sqcup (\bigsqcup_{i \in P'} \mathcal{I}_i^*)$.
For each $i \in P'$, let $\mathsf{opt}_i$ be the optimum of the instance $(S_i,\mathcal{I}_i)$.
Then we have $|\mathcal{I}_i^*| \leq (1+\tilde{\eps}) \cdot \mathsf{opt}_i$ for all $i \in P'$ where $\tilde{\eps} = \eps/2$.
Since $|\mathcal{I}^*| \leq r$, we have
\begin{equation} \label{eq-Iappx}
    |\mathcal{I}_\text{appx}| = |\mathcal{I}^*| + \sum_{i \in P'} |\mathcal{I}_i^*| \leq r+ \left(1+\frac{\eps}{2}\right) \sum_{i \in P'} \mathsf{opt}_i.
\end{equation}
Let $\mathcal{I}_\text{opt}$ be an optimal set cover for $(S,\mathcal{I})$.
We observe that for $i \in P'$, $\mathcal{I}_\text{opt} \cap \mathcal{I}_i$ is a set cover for $(S_i,\mathcal{I}_i)$, because $J_i$ is uncoverable (so the points in $S_i$ cannot be covered by any interval in $\mathcal{I} \backslash \mathcal{I}_i$).
It immediately follows that $\mathsf{opt}_i \leq |\mathcal{I}_\text{opt} \cap \mathcal{I}_i|$ for all $i \in P'$.
Therefore, we have
\begin{equation} \label{eq-opti}
    \sum_{i \in P'} \mathsf{opt}_i \leq \sum_{i \in P'} |\mathcal{I}_\text{opt} \cap \mathcal{I}_i|.
\end{equation}
The right-hand side of the above inequality can be larger than $|\mathcal{I}_\text{opt}|$ as some intervals in $\mathcal{I}_\text{opt}$ can belong to two $\mathcal{I}_i$'s.
The following lemma bounds the number of such intervals.
\begin{lemma} \label{lem-atmost2}
There are at most $2r$ intervals in $\mathcal{I}_\text{opt}$ that belong to exactly two $\mathcal{I}_i$'s.
\end{lemma}
\noindent
The above lemma immediately implies
\begin{equation} \label{eq-opt+2r}
    \sum_{i \in P'} |\mathcal{I}_\text{opt} \cap \mathcal{I}_i| \leq |\mathcal{I}_\text{opt}| + 2r = \mathsf{opt} + 2r.
\end{equation}
Combining Inequalities~\ref{eq-Iappx}, \ref{eq-opti}, and \ref{eq-opt+2r}, we deduce that
\begin{equation*}
    \begin{aligned}
        |\mathcal{I}_\text{appx}| & \leq r+\left(1+\frac{\eps}{2}\right) \sum_{i \in P'} \mathsf{opt}_i \\
        & \leq r+\left(1+\frac{\eps}{2}\right) \sum_{i \in P'} |\mathcal{I}_\text{opt} \cap \mathcal{I}_i| \\
        & \leq r+\left(1+\frac{\eps}{2}\right) \cdot (\mathsf{opt} + 2r) = (3+\eps) \cdot r+\left(1+\frac{\eps}{2}\right) \cdot \mathsf{opt} \\
        & < \frac{\eps}{2} \cdot \mathsf{opt}+\left(1+\frac{\eps}{2}\right) \cdot \mathsf{opt} = (1+\eps) \cdot \mathsf{opt},
    \end{aligned}
\end{equation*}
where the last inequality follows from the assumption $\mathsf{opt} > (6+2\eps) \cdot r/\eps$.

\paragraph{Time analysis.}
We briefly discuss the update and construction time of $\mathcal{D}_\text{new}$; a detailed analysis can be found in Appendix~\ref{appx-anaint}.
Since $\mathcal{D}_\text{new}$ is reconstructed periodically, it suffices to consider the first period (i.e., the period before the first reconstruction).
The construction of $\mathcal{D}_\text{new}$ can be easily done in $\widetilde{O}(n_0)$ time.
The update time of $\mathcal{D}_\text{new}$ consists of the time for updating the data structures $\mathcal{A}$ and $\mathcal{D}_\text{old}^{(1)},\dots,\mathcal{D}_\text{old}^{(r)}$, the time for maintaining the solution, and the time for reconstruction.
Since the period consists of $f(n_0,\varepsilon)$ operations, the size of each $(S_i,\mathcal{I}_i)$ is always bounded by $O(f(n_0,\varepsilon))$ during the period.
As argued before, we only need to update at most two $\mathcal{D}_\text{old}^{(i)}$'s after each operation.
Thus, updating the $\mathcal{D}_\text{old}$ data structures takes $\widetilde{O}(f(n_0,\varepsilon)^\alpha/\varepsilon^{1-\alpha})$ amortized time.
Maintaining the solution can be done in $\widetilde{O}(\delta + r)$ time, with a careful implementation.
The time for reconstruction is bounded by $\widetilde{O}(n_0+f(n_0,\varepsilon))$; we amortize it over the $f(n_0,\varepsilon)$ operations in the period and the amortized time cost is then $\widetilde{O}(n_0/f(n_0,\varepsilon))$, i.e., $\widetilde{O}(r)$.
In total, the amortized update time of $\mathcal{D}_\text{new}$ (during the first period) is $\widetilde{O}(f(n_0,\varepsilon)^\alpha/\varepsilon^{1-\alpha} + \delta + r)$.
If we set $f(n,\eps) = \min\{n^{1-\alpha'}/\eps^{\alpha'},n/2\}$ where $\alpha'$ is as defined in Theorem~\ref{thm-boot}, a careful calculation (see Appendix~\ref{appx-anaint}) shows that the amortized update time becomes $\widetilde{O}(n^{\alpha'}/\varepsilon^{1-\alpha'})$.

\subsubsection{Putting everything together}
With the bootstrapping theorem in hand, we are now able to design our dynamic interval set cover data structure.
The starting point is a ``trivial'' data structure, which simply uses the output-sensitive algorithm of Lemma~\ref{lem-osint} to re-compute an optimal interval set cover after each update.
Clearly, the update time of this data structure is $\widetilde{O}(n)$ and the construction time is $\widetilde{O}(n_0)$.
Thus, there exists a $(1+\varepsilon)$-approximate dynamic interval set cover data structure with $\widetilde{O}(n^{\alpha_0}/\varepsilon^{1-\alpha_0})$ amortized update time for $\alpha_0 = 1$ and $\widetilde{O}(n_0)$ construction time.
Define $\alpha_i = \alpha_{i-1}/(1+\alpha_{i-1})$ for $i \geq 1$.
By applying Theorem~\ref{thm-boot} $i$ times for a constant $i \geq 1$, we see the existence of a $(1+\varepsilon)$-approximate dynamic interval set cover data structure with $\widetilde{O}(n^{\alpha_i}/\varepsilon^{1-\alpha_i})$ amortized update time and $\widetilde{O}(n_0)$ construction time.
One can easily verify that $\alpha_i = 1/(i+1)$ for all $i \geq 0$.
Therefore, for any constant $\alpha>0$, we have an index $i \geq 0$ satisfying $\alpha_i < \alpha$ and hence $\widetilde{O}(n^{\alpha_i}/\varepsilon^{1-\alpha_i}) = O(n^\alpha/\varepsilon)$.
We finally conclude the following.
\begin{theorem} \label{thm-ISC}
For a given approximation factor $\eps>0$ and any constant $\alpha>0$, there exists a $(1+\varepsilon)$-approximate dynamic interval set cover data structure $\mathcal{D}$ with $O(n^\alpha/\eps)$ amortized update time and $\widetilde{O}(n_0)$ construction time.
\end{theorem}

\subsection{2D set cover and hitting set for quadrants and unit squares} \label{sec-DQSC}
In this section, we present our bootstrapping framework for 2D dynamic set cover and hitting set.
Our framework works for quadrants and unit squares.

We first show that dynamic unit-square set cover, dynamic unit-square hitting set, and dynamic quadrant hitting set can all be reduced to dynamic quadrant set cover.
\begin{lemma} \label{lem-sqreduction}
Suppose there exists a $c$-approximate dynamic quadrant set cover data structure with $f(n)$ amortized update time and $\widetilde{O}(n_0)$ construction time, where $f$ is an increasing function.
Then there exist $O(c)$-approximate dynamic unit-square set cover, dynamic unit-square hitting set, and dynamic quadrant hitting set data structures with $\widetilde{O}(f(n))$ amortized update time and $\widetilde{O}(n_0)$ construction time.
\end{lemma}

Now it suffices to consider dynamic quadrant set cover.
In order to do bootstrapping, we need an output-sensitive algorithm for quadrant set cover, analog to the one in Lemma~\ref{lem-osint} for intervals.
To design such an algorithm is considerably more difficult compared to the 1D case, and we defer it to Section~\ref{sec-osquad}.
Before this, let us first discuss the bootstrapping procedure, assuming the existence of a $\mu$-approximate output-sensitive algorithm for quadrant set cover.

\subsubsection{Bootstrapping}
We prove the following bootstrapping theorem, which is the technical heart of our result.
\begin{theorem} \label{thm-bootquad}
Assume quadrant set cover admits a $\mu$-approximate output-sensitive algorithm for some constant $\mu \geq 1$.
Then we have the following result. \\
\textnormal{\bf (*)} Let $\alpha \in [0,1]$ be a number.
If there exists a $(\mu+\eps)$-approximate dynamic quadrant set cover data structure $\mathcal{D}_\textnormal{old}$ with $\widetilde{O}(n^\alpha/\eps^{1-\alpha})$ amortized update time and $\widetilde{O}(n_0)$ construction time for any $\eps>0$,
then there exists a $(\mu+\eps)$-approximate dynamic quadrant set cover data structure $\mathcal{D}_\textnormal{new}$ with $\widetilde{O}(n^{\alpha'}/\eps^{1-\alpha'})$ amortized update time and $\widetilde{O}(n_0)$ construction time for any $\eps>0$, where $\alpha' = 2\alpha/(1+2\alpha)$.
Here $n$ (resp., $n_0$) denotes the size of the current (resp., initial) problem instance.
\end{theorem}

Assuming the existence of $\mathcal{D}_\text{old}$ as in the theorem, we are going to design the improved data structure $\mathcal{D}_\text{new}$.
Let $(S,\mathcal{Q})$ be a dynamic quadrant set cover instance.
As before, we denote by $n$ (resp., $n_0$) the size of the current (resp., initial) $(S,\mathcal{Q})$.

\begin{figure}[h]
    \centering
    \includegraphics[height=6cm]{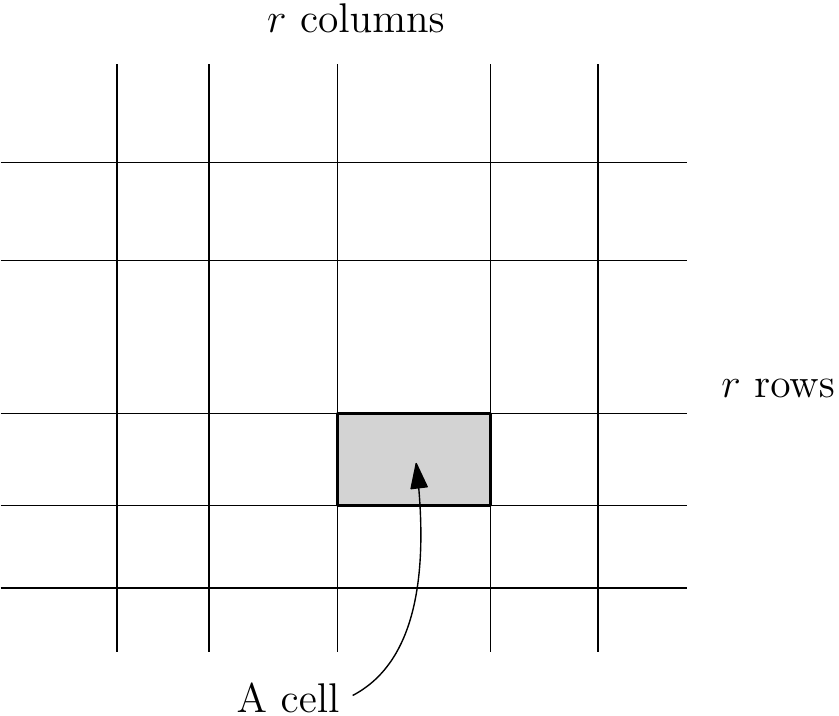}
    \caption{The $r \times r$ grid. Note that the cells may have different sizes.}
    \label{fig-partition}
\end{figure}

\begin{figure}[h]
    \centering
    \includegraphics[height=3.5cm]{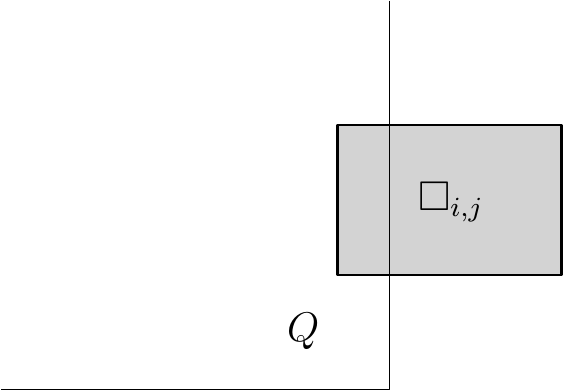}
    \caption{A quadrant $Q$ that left intersects $\Box_{i,j}$.}
    \label{fig-leftint}
\end{figure}

\paragraph{The construction of $\mathcal{D}_\text{new}$.}
Initially, $|S|+|\mathcal{Q}| = n_0$.
Essentially, our data structure $\mathcal{D}_\text{new}$ consists of two parts.
The first part is the data structure $\mathcal{A}$ required for the $\mu$-approximate output-sensitive algorithm.
The second part is a family of $\mathcal{D}_\text{old}$ data structures defined as follows.
Let $f$ be a function to be determined shortly.
We use an orthogonal grid to partition the plane $\mathbb{R}^2$ into $r \times r$ cells for $r = \lceil n_0 / f(n_0,\eps) \rceil$ such that each row (resp., column) of the grid contains $O(f(n_0,\eps))$ points in $S$ and $O(f(n_0,\eps))$ vertices of the quadrants in $\mathcal{Q}$ (see Figure~\ref{fig-partition} for an illustration).
Denote by $\Box_{i,j}$ the cell in the $i$-th row and $j$-th column.
Define $S_{i,j} = S \cap \Box_{i,j}$.
Also, we need to define a sub-collection $\mathcal{Q}_{i,j} \subseteq \mathcal{Q}$.
Recall that in the 1D case, we define $\mathcal{I}_i$ as the sub-collection of intervals in $\mathcal{I}$ that partially intersect the portion $J_i$.
However, for technical reasons, here we cannot simply define $\mathcal{Q}_{i,j}$ as the sub-collection of quadrants in $\mathcal{Q}$ that partially intersects $\Box_{i,j}$.
Instead, we define $\mathcal{Q}_{i,j}$ as follows.
We include in $\mathcal{Q}_{i,j}$ all the quadrants in $\mathcal{Q}$ whose vertices lie in $\Box_{i,j}$.
Besides, we also include in $\mathcal{Q}_{i,j}$ the following (at most) four \textit{special} quadrants.
We say a quadrant $Q$ \textit{left intersects} $\Box_{i,j}$ if $Q$ partially intersects $\Box_{i,j}$ and contains the left edge of $\Box_{i,j}$ (see Figure~\ref{fig-leftint} for an illustration); similarly, we define ``right intersects'', ``top intersects'', and ``bottom intersects''.
Among a collection of quadrants, the \textit{leftmost/rightmost/topmost/bottommost} quadrant refers to the quadrant whose vertex is the leftmost/rightmost/topmost/bottommost.
We include in $\mathcal{Q}_{i,j}$ the rightmost quadrant in $\mathcal{Q}$ that left intersects $\Box_{i,j}$, the leftmost quadrant in $\mathcal{Q}$ that right intersects $\Box_{i,j}$, the bottommost quadrant in $\mathcal{Q}$ that top intersects $\Box_{i,j}$, and the topmost quadrant in $\mathcal{Q}$ that bottom intersects $\Box_{i,j}$ (if these quadrants exist).
When the instance $(S,\mathcal{Q})$ is updated, the grid keeps unchanged, but the $S_{i,j}$'s and $\mathcal{Q}_{i,j}$'s change along with $S$ and $\mathcal{Q}$.
We view each $(S_{i,j},\mathcal{Q}_{i,j})$ as a dynamic quadrant set cover instance, and let $\mathcal{D}_\text{old}^{(i,j)}$ be the data structure $\mathcal{D}_\text{old}$ built on $(S_{i,j},\mathcal{Q}_{i,j})$ for the approximation factor $\tilde{\varepsilon} = \varepsilon/2$.
The second part of $\mathcal{D}_\text{new}$ consists of the data structures $\mathcal{D}_\text{old}^{(i,j)}$ for $i,j \in \{1,\dots,r\}$.

\paragraph{Update and reconstruction.}
After each operation on $(S,\mathcal{Q})$, we update the data structure $\mathcal{A}$.
Also, if some $(S_{i,j},\mathcal{Q}_{i,j})$ changes, we update the data structure $\mathcal{D}_\text{old}^{(i,j)}$.
Note that an operation on $S$ changes exactly one $S_{i,j}$, and an operation on $\mathcal{Q}$ may only change the $\mathcal{Q}_{i,j}$'s in one row and one column (specifically, if the vertex of the inserted/deleted quadrant lies in $\Box_{i,j}$, then only $\mathcal{Q}_{i,1},\dots,\mathcal{Q}_{i,r},\mathcal{Q}_{1,j},\dots,\mathcal{Q}_{r,j}$ may change).
Thus, we in fact only need to update the $\mathcal{D}_\text{old}^{(i,j)}$'s in one row and one column.
Besides the update, we also reconstruct the entire data structure $\mathcal{D}_\text{new}$ periodically, where the (first) reconstruction happens after processing $f(n_0,\eps)$ operations.
This part is totally the same as in our 1D data structure.

\paragraph{Maintaining a solution.}
We now discuss how to maintain a $(\mu+\varepsilon)$-approximate optimal set cover $\mathcal{Q}_\text{appx}$ for $(S,\mathcal{Q})$.
Let $\mathsf{opt}$ denote the optimum of the current $(S,\mathcal{Q})$.
Set $\delta = \min\{(8\mu + 4\eps + 2) \cdot r^2 / \eps,n\}$.
If $\mathsf{opt} \leq \delta$, then the output-sensitive algorithm can compute a $\mu$-approximate optimal set cover for $(S,\mathcal{Q})$ in $\widetilde{O}(\mu \delta)$ time.
Thus, we simulate the output-sensitive algorithm within that amount of time.
If the algorithm successfully computes a solution, we use it as our $\mathcal{Q}_\text{appx}$.
Otherwise, we construct $\mathcal{Q}_\text{appx}$ as follows.
We say the cell $\Box_{i,j}$ is \textit{coverable} if there exists $Q \in \mathcal{Q}$ that contains $\Box_{i,j}$ and \textit{uncoverable} otherwise.
Let $P = \{(i,j): \Box_{i,j} \text{ is coverable}\}$ and $P' = \{(i,j): \Box_{i,j} \text{ is uncoverable}\}$.
We try to use the quadrants in $\mathcal{I}$ to ``cover'' all coverable cells.
That is, for each $(i,j) \in P$, we find a quadrant in $\mathcal{Q}$ that contains $\Box_{i,j}$, and denote by $\mathcal{Q}^*$ the set of all these quadrants.
Then we consider the uncoverable cells.
If for some $(i,j) \in P'$, the data structure $\mathcal{D}_\text{old}^{(i,j)}$ tells us that the instance $(S_{i,j},\mathcal{Q}_{i,j})$ has no set cover, then we immediately make a no-solution decision, i.e., decide that the current $(S,\mathcal{Q})$ has no feasible set cover, and continue to the next operation.
Otherwise, for each $(i,j) \in P'$, the data structure $\mathcal{D}_\text{old}^{(i,j)}$ maintains a $(\mu+\tilde{\eps})$-approximate optimal set cover $\mathcal{Q}_{i,j}^*$ for $(S_{i,j},\mathcal{Q}_{i,j})$.
We then define $\mathcal{Q}_\text{appx} = \mathcal{Q}^* \sqcup \left( \bigsqcup_{(i,j) \in P'} \mathcal{Q}_{i,j}^* \right)$.

We will see later that $\mathcal{Q}_\text{appx}$ is always a $(\mu+\varepsilon)$-approximate optimal set cover for $(S,\mathcal{Q})$.
Before this, let us briefly discuss how to store $\mathcal{Q}_\text{appx}$ to support the desired queries.
If $\mathcal{Q}_\text{appx}$ is computed by the output-sensitive algorithm, then we have all the elements of $\mathcal{Q}_\text{appx}$ in hand and can easily store them in a data structure to support the queries.
Otherwise, $\mathcal{Q}_\text{appx}$ is defined as the disjoint union of $\mathcal{Q}^*$ and $\mathcal{Q}_{i,j}^*$'s.
In this case, the size and reporting queries can be handled in the same way as that in the 1D problem, by taking advantage of the fact that $\mathcal{Q}_{i,j}^*$ is maintained in $\mathcal{D}_\text{old}^{(i,j)}$.
However, the situation for the membership query is more complicated, because now a quadrant in $\mathcal{Q}$ may belong to many $\mathcal{Q}_{i,j}^*$'s.
This issue can be handled by collecting all special quadrants in $\mathcal{Q}_\text{appx}$ and building on them a data structure that supports the membership query.
We defer the detailed discussion to Appendix~\ref{appx-anaquad}.

\paragraph{Correctness.}
We now prove the correctness of our data structure $\mathcal{D}_\text{new}$.
First, we show that the no-solution decision made by our data structure is correct.
\begin{lemma} \label{lem-noansquad}
    $\mathcal{D}_\textnormal{new}$ makes a no-solution decision iff the current $(S,\mathcal{Q})$ has no set cover.
\end{lemma}
\noindent
Next, we show that the solution $\mathcal{Q}_\text{appx}$ maintained by $\mathcal{D}_\text{new}$ is truly a $(\mu+\varepsilon)$-approximate optimal set cover for $(S,\mathcal{Q})$.
If $\mathcal{Q}_\text{appx}$ is computed by the output-sensitive algorithm, then it is a $\mu$-approximate optimal set cover for $(S,\mathcal{Q})$.
Otherwise, $\mathsf{opt} > \delta = \min\{(8\mu + 4\eps + 2) \cdot r^2/\eps,n\}$, i.e., either $\mathsf{opt} > (8\mu + 4\eps + 2) \cdot r^2/\eps$ or $\mathsf{opt} > n$.
If $\mathsf{opt} > n$, then $(S,\mathcal{Q})$ has no set cover (i.e., $\mathsf{opt} = \infty$) and $\mathcal{D}_\text{new}$ makes a no-solution decision by Lemma~\ref{lem-noansquad}.
So assume $\mathsf{opt} > (8\mu + 4\eps + 2)r^2/\eps$.
In this case, $\mathcal{Q}_\text{appx} = \mathcal{Q}^* \sqcup (\bigsqcup_{(i,j) \in P'} \mathcal{Q}_{i,j}^*)$.
For each $(i,j) \in P'$, let $\mathsf{opt}_{i,j}$ be the optimum of $(S_{i,j},\mathcal{Q}_{i,j})$.
Then we have $|\mathcal{Q}_{i,j}^*| \leq (\mu+\tilde{\eps}) \cdot \mathsf{opt}_{i,j}$ for all $(i,j) \in P'$ where $\tilde{\eps} = \eps/2$.
Since $|\mathcal{Q}^*| \leq r^2$, we have
\begin{equation} \label{eq-quadappx}
    \mathcal{Q}_\text{appx} = |\mathcal{Q}^*| + \sum_{(i,j) \in P'} |\mathcal{Q}_{i,j}^*| \leq
    r^2 + \left( \mu+\frac{\eps}{2} \right) \sum_{(i,j) \in P'} \mathsf{opt}_{i,j}.
\end{equation}
Let $\mathcal{Q}_{i,j}' \subseteq \mathcal{Q}_{i,j}$ consist of the non-special quadrants, i.e., those whose vertices are in $\Box_{i,j}$.
\begin{lemma} \label{lem-quadopt}
    We have $|\mathcal{Q}_\textnormal{opt} \cap \mathcal{Q}_{i,j}'| + 4 \geq \mathsf{opt}_{i,j}$ for all $(i,j) \in P'$, and in particular,
    \begin{equation} \label{eq-quadij}
        \mathsf{opt} + 4r^2 = |\mathcal{Q}_\textnormal{opt}| + 4r^2 \geq \sum_{(i,j) \in P'} \mathsf{opt}_{i,j}.
    \end{equation}
\end{lemma}

\noindent
Using Equations~\ref{eq-quadappx} and~\ref{eq-quadij}, we deduce that
\begin{equation*}
    \begin{aligned}
        |\mathcal{Q}_\text{appx}| & \leq r^2+\left(\mu+\frac{\eps}{2}\right) \sum_{(i,j) \in P'} \mathsf{opt}_{i,j} \\
        & \leq r^2+\left(\mu+\frac{\eps}{2}\right) (\mathsf{opt} + 4r^2) \\
        & \leq (4\mu + 2\eps +1) \cdot r^2+\left(\mu+\frac{\eps}{2}\right) \cdot \mathsf{opt} \\
        & < \frac{\eps}{2} \cdot \mathsf{opt}+\left(\mu+\frac{\eps}{2}\right) \cdot \mathsf{opt} = (\mu+\eps) \cdot \mathsf{opt},
    \end{aligned}
\end{equation*}
where the last inequality follows from the fact that $\mathsf{opt} > (8\mu + 4\eps +2) \cdot r^2/\eps$.

\paragraph{Time analysis.}
We briefly discuss the update and construction time of $\mathcal{D}_\text{new}$; a detailed analysis can be found in Appendix~\ref{appx-anaquad}.
It suffices to consider the first period (i.e., the period before the first reconstruction).
We first observe the following fact.
\begin{lemma} \label{lem-bound}
    At any time in the first period, we have $\sum_{k=1}^r (|S_{i,k}| + |\mathcal{Q}_{i,k}|) = O(f(n_0,\varepsilon)+r)$ for all $i \in \{1,\dots,r\}$ and $\sum_{k=1}^r (|S_{k,j}| + |\mathcal{Q}_{k,j}|) = O(f(n_0,\varepsilon)+r)$ for all $j \in \{1,\dots,r\}$.
\end{lemma}
The above lemma implies that the sum of the sizes of all $(S_{i,j},\mathcal{Q}_{i,j})$ is $O(n_0+r^2)$ at any time in the first period.
Therefore, constructing $\mathcal{D}_\text{new}$ can be easily done in $\widetilde{O}(n_0 + r^2)$ time.
The update time of $\mathcal{D}_\text{new}$ consists of the time for reconstruction, the time for updating $\mathcal{A}$ and $\mathcal{D}_\text{old}^{(i,j)}$'s, and the time for maintaining the solution.
Using almost the same analysis as in the 1D problem, we can show that the reconstruction takes $\widetilde{O}(r + r^2/f(n_0,\varepsilon))$ amortized time and maintaining the solution can be done in $\widetilde{O}(\delta + r^2)$ time, with a careful implementation.
The time for updating the $\mathcal{D}_\text{old}$ data structures requires a different analysis.
Let $m_{i,j}$ denote the current size of $(S_{i,j},\mathcal{Q}_{i,j})$.
As argued before, we in fact only need to update the $\mathcal{D}_\text{old}$ data structures in one row and one column (say the $i$-th row and $j$-th column).
Hence, updating the $\mathcal{D}_\text{old}$ data structures takes $\widetilde{O}(\sum_{k=1}^r m_{i,k}^\alpha/\varepsilon^{1-\alpha} + \sum_{k=1}^r m_{k,j}^\alpha/\varepsilon^{1-\alpha})$ amortized time.
Lemma~\ref{lem-bound} implies that $\sum_{k=1}^r m_{i,k} = O(f(n_0,\varepsilon)+r)$ and $\sum_{k=1}^r m_{k,j} = O(f(n_0,\varepsilon)+r)$.
Since $\alpha \leq 1$, by H\"older's inequality and Lemma~\ref{lem-bound},
\begin{equation*}
    \sum_{k=1}^r m_{i,k}^\alpha \leq \left(\frac{\sum_{k=1}^r m_{i,k}}{r}\right)^\alpha \cdot r = O(r^{1-\alpha} \cdot (f(n_0,\eps)+r)^\alpha)
\end{equation*}
and similarly $\sum_{k=1}^r m_{k,j}^\alpha = O(r^{1-\alpha} \cdot (f(n_0,\eps)+r)^\alpha)$.
It follows that updating the $\mathcal{D}_\text{old}$ data structures takes $\widetilde{O}(r^{1-\alpha} \cdot (f(n_0,\eps) + r)^\alpha/ \eps^{1-\alpha})$ amortized time. 
In total, the amortized update time of $\mathcal{D}_\text{new}$ (during the first period) is $\widetilde{O}(r^{1-\alpha} \cdot (f(n_0,\eps)+r)^\alpha + \delta + r^2)$.
If we set $f(n,\eps) = \min\{n^{1-\alpha'/2}/(\sqrt{\eps})^{\alpha'},n/2\}$ where $\alpha'$ is as defined in Theorem~\ref{thm-bootquad}, a careful calculation (see Appendix~\ref{appx-anaquad}) shows that the amortized update time becomes $\widetilde{O}(n^{\alpha'}/\varepsilon^{1-\alpha'})$ .

\subsubsection{An output-sensitive quadrant set cover algorithm} \label{sec-osquad}
We propose an $O(1)$-approximate output-sensitive algorithm for quadrant set cover, which is needed for applying Theorem~\ref{thm-bootquad}.
Let $(S,\mathcal{Q})$ be a quadrant set cover instance of size $n$, and $\mathsf{opt}$ be its optimum.
Our goal is to compute an $O(1)$-approximate optimal set cover for $(S,\mathcal{Q})$ in $\widetilde{O}(\mathsf{opt})$ time, using some basic data structure built on $(S,\mathcal{Q})$.

For simplicity, let us assume that $(S,\mathcal{Q})$ has a set cover; how to handle the no-solution case is discussed in Appendix~\ref{appx-nosol}.
There are four types of quadrants in $\mathcal{Q}$, southeast, southwest, northeast, northwest; we denote by $\mathcal{Q}^\text{SE},\mathcal{Q}^\text{SW},\mathcal{Q}^\text{NE},\mathcal{Q}^\text{NW} \subseteq \mathcal{Q}$ the sub-collections of these types of quadrants, respectively.
Let $U^\text{SE}$ denote the union of the quadrants in $\mathcal{Q}^\text{SE}$, and define $U^\text{SW},U^\text{NE},U^\text{NW}$ similarly.
Since $(S,\mathcal{Q})$ has a set cover, we have $S = (S \cap U^\text{SE}) \cup (S \cap U^\text{SW}) \cup (S \cap U^\text{NE}) \cup (S \cap U^\text{NW})$.
Therefore, if we can compute $O(1)$-approximate optimal set covers for $(S \cap U^\text{SE},\mathcal{Q})$, $(S \cap U^\text{SW},\mathcal{Q})$, $(S \cap U^\text{NE},\mathcal{Q})$, and $(S \cap U^\text{NW},\mathcal{Q})$, then the union of these four set covers is an $O(1)$-approximate optimal set cover for $(S,\mathcal{Q})$.

\begin{figure}[h]
    \centering
    \includegraphics[height=4.5cm]{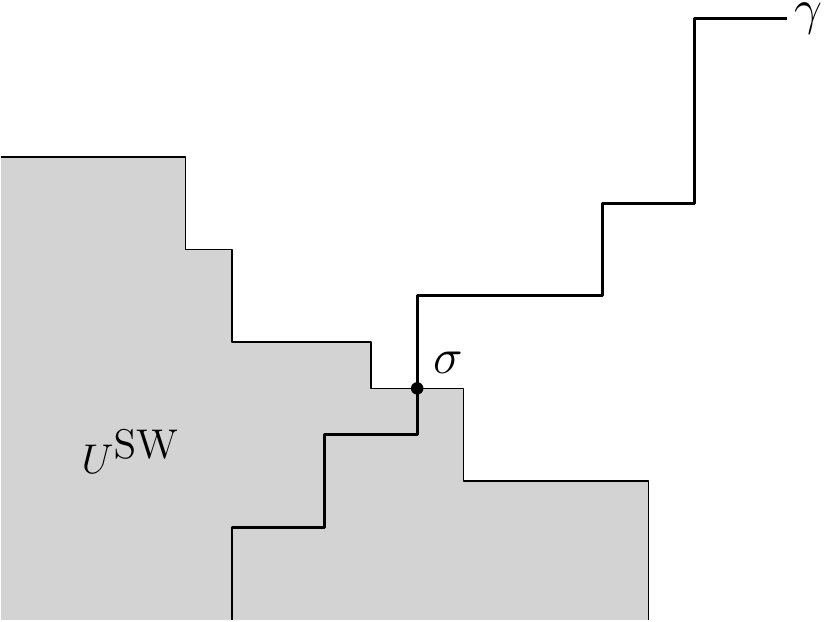}
    \caption{Illustrating the curve $\gamma$ and the point $\sigma$.}
    \label{fig-sigma}
\end{figure}

With this observation, it now suffices to show how to compute an $O(1)$-approximate optimal set cover for $(S \cap U^\text{SE},\mathcal{Q})$ in $\widetilde{O}(\mathsf{opt}^\text{SE})$ time, where $\mathsf{opt}^\text{SE}$ is the optimum of $(S \cap U^\text{SE},\mathcal{Q})$.
The main challenge is to guarantee the running time and approximation ratio simultaneously.
We begin by introducing some notation.
Let $\gamma$ denote the boundary of $U^\text{SE}$, which is an orthogonal staircase curve from bottom-left to top-right.
If $\gamma \cap U^\text{SW} \neq \emptyset$, then $\gamma \cap U^\text{SW}$ is a connected portion of $\gamma$ that contains the bottom-left end of $\gamma$.
Define $\sigma$ as the ``endpoint'' of $\gamma \cap U^\text{SW}$, i.e., the point on $\gamma \cap U^\text{SW}$ that is closest the top-right end of $\gamma$.
See Figure~\ref{fig-sigma} for an illustration.
If $\gamma \cap U^\text{SW} = \emptyset$, we define $\sigma$ as the bottom-left end of $\gamma$ (which is a point whose $y$-coordinate equals to $-\infty$).
For a number $\tilde{y} \in \mathbb{R}$, we define $\phi(\tilde{y})$ as the \textit{leftmost} point in $S \cap U^\text{SE}$ whose $y$-coordinate is greater than $\tilde{y}$; we say $\phi(\tilde{y})$ does not exist if no point in $S \cap U^\text{SE}$ has $y$-coordinate greater than $\tilde{y}$.
For a point $a \in \mathbb{R}^2$ and a collection $\mathcal{P}$ of quadrants, we define $\varPhi_\rightarrow(a,\mathcal{P})$ and $\varPhi_\uparrow(a,\mathcal{P})$ as the rightmost and topmost quadrants in $\mathcal{P}$ that contains $a$, respectively.
For a quadrant $Q$, we denote by $x(Q)$ and $y(Q)$ the $x$- and $y$-coordinates of the vertex of $Q$, respectively.

To get some intuition, let us consider a very simple case, where $\mathcal{Q}$ only consists of southeast quadrants.
In this case, one can compute an optimal set cover for $(S \cap U^\text{SE},\mathcal{Q})$ using a greedy algorithm similar to the 1D interval set cover algorithm: repeatedly pick the leftmost uncovered point in $S \cap U^\text{SE}$ and cover it using the topmost (southeast) quadrant in $\mathcal{Q}$.
Using the notations defined above, we can describe this algorithm as follows.
Set $\mathcal{Q}_\text{ans} \leftarrow \emptyset$ and $\tilde{y} \leftarrow -\infty$ initially, and repeatedly do $a \leftarrow \phi(\tilde{y})$, $Q \leftarrow \varPhi_\uparrow(\sigma,\mathcal{Q}^\text{SE})$, $\mathcal{Q}_\text{ans} \leftarrow \mathcal{Q}_\text{ans} \cup \{Q\}$, $\tilde{y} \leftarrow y(Q)$ until $\phi(\tilde{y})$ does not exist.
Eventually, $\mathcal{Q}_\text{ans}$ is the set cover we want.

Now we try to extend this algorithm to the general case.
However, the situation here becomes much more complicated, since we may have three other types of quadrants in $\mathcal{Q}$, which have to be carefully dealt with in order to guarantee the correctness.
But the intuition remains the same: we still construct the solution in a greedy manner.
The following procedure describes our algorithm.
\begin{enumerate}
    \item $\mathcal{Q}_\text{ans} \leftarrow \emptyset$. $\tilde{y} \leftarrow -\infty$.
    If $\phi(\tilde{y})$ does not exist, then go to Step~6.
    \item $\mathcal{Q}_\text{ans} \leftarrow \{\varPhi_\rightarrow(\sigma,\mathcal{Q}^\text{SW}),\varPhi_\uparrow(\sigma,\mathcal{Q}^\text{SE})\}$. $\tilde{y} \leftarrow y(\varPhi_\uparrow(\sigma,\mathcal{Q}^\text{SE}))$. If $\phi(\tilde{y})$ exists, then $a \leftarrow \phi(\tilde{y})$, else go to Step~6.
    \item If $a \in U^\text{NE}$, then $\mathcal{Q}_\text{ans} \leftarrow \mathcal{Q}_\text{ans} \cup \{\varPhi_\uparrow(a,\mathcal{Q}^\text{NE}),\varPhi_\uparrow(a,\mathcal{Q}^\text{SE})\}$ and go to Step~6.
    \item If $a \in U^\text{NW}$, then $\mathcal{Q}_\text{ans} \leftarrow \mathcal{Q}_\text{ans} \cup \{\varPhi_\rightarrow(a,\mathcal{Q}^\text{NW}),\varPhi_\uparrow(a,\mathcal{Q}^\text{SE})\}$ and $Q \leftarrow \varPhi_\uparrow(v,\mathcal{Q}^\text{SE})$ where $v$ is the vertex of $\varPhi_\rightarrow(a,\mathcal{Q}^\text{NW})$, otherwise $Q \leftarrow \varPhi_\uparrow(a,\mathcal{Q}^\text{SE})$.
    \item $\mathcal{Q}_\text{ans} \leftarrow \mathcal{Q}_\text{ans} \cup \{Q\}$. $\tilde{y} \leftarrow y(Q)$.
    If $\phi(\tilde{y})$ exists, then $a \leftarrow \phi(\tilde{y})$ and go to Step~3.
    \item Output $\mathcal{Q}_\text{ans}$.
\end{enumerate}

\noindent
The following lemma proves the correctness of our algorithm.
\begin{lemma} \label{lem-qans}
$\mathcal{Q}_\textnormal{ans}$ covers all points in $S \cap U^\textnormal{SE}$, and $|\mathcal{Q}_\textnormal{ans}| = O(\mathsf{opt}^\textnormal{SE})$.
\end{lemma}

The remaining task is to show how to perform our algorithm in $\widetilde{O}(\mathsf{opt}^\text{SE})$ time using basic data structures.
It is clear that our algorithm terminates in $O(\mathsf{opt}^\text{SE})$ steps, since we include at least one quadrant to $\mathcal{Q}_\text{ans}$ in each iteration of the loop Step 3--5 and eventually $|\mathcal{Q}_\text{ans}| = O(\mathsf{opt}^\text{SE})$ by Lemma~\ref{lem-qans}.
Thus, it suffices to show that each step can be done in $\widetilde{O}(1)$ time.
In every step of our algorithm, all work can be done in constant time except the tasks of computing the point $\sigma$, testing whether $a \in U^\text{NE}$ and $a \in U^\text{NW}$ for a given point $a$,
computing the quadrants $\varPhi_\rightarrow(a,\mathcal{Q}^\text{SW}),\varPhi_\rightarrow(a,\mathcal{Q}^\text{NW}),\varPhi_\uparrow(a,\mathcal{Q}^\text{SE}),\varPhi_\uparrow(a,\mathcal{Q}^\text{NE})$ for a given point $a$, and computing $\phi(\tilde{y})$ for a given number $\tilde{y}$.
All these tasks except the computation of $\phi(\cdot)$ can be easily done in $\widetilde{O}(1)$ time by storing the quadrants in binary search trees.
To compute $\phi(\cdot)$ in $\widetilde{O}(1)$ time is more difficult, and we achieve this by properly using range trees built on both $S$ and $\mathcal{Q}^\text{SE}$.
The details are presented in Appendix~\ref{appx-quados}.

Using the above algorithm, we can compute $O(1)$-approximate optimal set covers for $(S \cap U^\text{SE},\mathcal{Q})$, $(S \cap U^\text{SW},\mathcal{Q})$, $(S \cap U^\text{NE},\mathcal{Q})$, and $(S \cap U^\text{NW},\mathcal{Q})$.
As argued before, the union of these four set covers, denoted by $\mathcal{Q}^*$, is an $O(1)$-approximate optimal set covers for $(S,\mathcal{Q})$.
\begin{theorem} \label{thm-osquad}
Quadrant set cover admits an $O(1)$-approximate output-sensitive algorithm.
\end{theorem}

\subsubsection{Putting everything together}
With the bootstrapping theorem in hand, we are now able to design our dynamic quadrant set cover data structure.
Again, the starting point is a ``trivial'' data structure which uses the output-sensitive algorithm of Theorem~\ref{thm-osquad} to re-compute an optimal quadrant set cover after each update.
Clearly, the update time of this data structure is $\widetilde{O}(n)$ and the construction time is $\widetilde{O}(n_0)$.
Let $\mu = O(1)$ be the approximation ratio of the output-sensitive algorithm.
The trivial data structure implies the existence of a $(\mu+\varepsilon)$-approximate dynamic quadrant set cover data structure with $\widetilde{O}(n^{\alpha_0}/\varepsilon^{1-\alpha_0})$ amortized update time for $\alpha_0 = 1$ and $\widetilde{O}(n_0)$ construction time.
Define $\alpha_i = 2\alpha_{i-1}/(1+2\alpha_{i-1})$ for $i \geq 1$.
By applying Theorem~\ref{thm-bootquad} $i$ times for a constant $i \geq 1$, we see the existence of a $(\mu+\varepsilon)$-approximate dynamic quadrant set cover data structure with $\widetilde{O}(n^{\alpha_i}/\varepsilon^{1-\alpha_i})$ amortized update time and $\widetilde{O}(n_0)$ construction time.
One can easily verify that $\alpha_i = 2^i/(2^{i+1}-1)$ for all $i \geq 0$.
Therefore, for any constant $\alpha>0$, we have an index $i \geq 0$ satisfying $\alpha_i < 1/2 + \alpha$ and hence $\widetilde{O}(n^{\alpha_i}/\varepsilon^{1-\alpha_i}) = O(n^{1/2 + \alpha}/\varepsilon)$.
Setting $\varepsilon$ to be any constant, we finally conclude the following.
\begin{theorem} \label{thm-QSC}
For any constant $\alpha>0$, there exists an $O(1)$-approximate dynamic quadrant set cover data structure with $O(n^{1/2+\alpha})$ amortized update time and $\widetilde{O}(n_0)$ construction time.
\end{theorem}

\noindent
By the reduction of Lemma~\ref{lem-sqreduction}, we have the following corollary.
\begin{corollary}
For any constant $\alpha>0$, there exist $O(1)$-approximate dynamic unit-square set cover, dynamic unit-square hitting set, and dynamic quadrant hitting set data structures with $O(n^{1/2+\alpha})$ amortized update time and $\widetilde{O}(n_0)$ construction time.
\end{corollary}


\section{Second framework: Local modification} \label{sec-localmod}

In this section, we give our second framework for dynamic geometric set cover and hitting set, which results in poly-logarithmic update time data structures for dynamic interval hitting set, as well as dynamic quadrant and unit-square set cover (resp., hitting set) in the partially dynamic setting.
This framework applies to \textit{stable} dynamic geometric set cover and hitting set problems, in which each operation cannot change the optimum of the problem instance significantly.
One can easily see that any dynamic set cover (resp., hitting set) problem in the partially dynamic setting is stable.
Interestingly, we will see that the dynamic interval hitting set problem, even in the fully dynamic setting, is stable.

The basic idea behind this framework is quite simple, and we call it \textit{local modification}.
Namely, at the beginning, we compute a good solution $S^*$ for the problem instance using an output-sensitive algorithm, and after each operation, we slightly modify $S^*$, by adding/removing a constant number of elements, to guarantee that $S^*$ is a feasible solution for the current instance.
After processing many operations, we re-compute a new $S^*$ and repeat this procedure.
The stability of the problem can guarantee that $S^*$ is always a good approximation for the optimal solution.

Next, we formally define the notion of \textit{stability}.
The \textit{quasi-optimum} of a set cover instance $(S,\mathcal{R})$ is the size of an optimal set cover for $(S',\mathcal{R})$, where $S' \subseteq S$ consists of the points that can be covered by $\mathcal{R}$.
The quasi-optimum of a hitting set instance is defined in a similar way.
For a dynamic set cover (resp., hitting set) instance $\varPi$, we use $\mathsf{opt}_i(\varPi)$ to denote the quasi-optimum of $\varPi$ at the time $i$ (i.e., after the $i$-th operation).
\begin{definition}
A dynamic set cover (resp., hitting set) instance $\varPi$ is \textbf{stable} if $|\mathsf{opt}_i(\varPi) - \mathsf{opt}_{i+1}(\varPi)| \leq 1$ for all $i \geq 0$.
A dynamic set cover (resp., hitting set) problem is \textbf{stable} if any instance of the problem is stable.
\end{definition}

\begin{lemma} \label{lem-stab}
Any dynamic set cover (resp., hitting set) problem in the partially dynamic setting is stable.
The dynamic interval hitting set problem is stable.
\end{lemma}

\subsection{Dynamic interval hitting set}
We show how to use the idea of local modification to obtain a $(1+\varepsilon)$-approximate dynamic interval hitting set data structure $\mathcal{D}$ with $\widetilde{O}(1/\varepsilon)$ amortized update time.

Similarly to interval set cover, the interval hitting set problem can also be solved using a simple greedy algorithm, which can be performed in $\widetilde{O}(\mathsf{opt})$ time if we store (properly) the points and intervals in some basic data structure (e.g., binary search trees).
We omit the proof of the following lemma, as it is almost the same as the proof of Lemma~\ref{lem-osint}.
\begin{lemma}
    The interval hitting set problem yields an exact output-sensitive algorithm.
\end{lemma}

Let $(S,\mathcal{I})$ be a dynamic interval hitting set instance.
We denote by $n$ (resp., $n_0$) the size of the current (resp., initial) $(S,\mathcal{I})$, and by $\mathsf{opt}$ the optimum of the current $(S,\mathcal{I})$.

\paragraph{The data structure.}
Our data structure $\mathcal{D}$ basically consists of three parts.
The first part is the basic data structure $\mathcal{A}$ built on $(S,\mathcal{I})$ required for the output-sensitive algorithm.
The second part is the following dynamic data structure $\mathcal{B}$ which can tell us whether the current instance has a hitting set.
\begin{lemma} \label{lem-dsexhs}
    There exists a dynamic data structure $\mathcal{B}$ built on $(S,\mathcal{I})$ with $\widetilde{O}(1)$ update time and $\widetilde{O}(n_0)$ construction time that can indicate whether the current $(S,\mathcal{I})$ has a hitting set or not.
\end{lemma}

\noindent
The third part consists of the following two basic data structures $\mathcal{C}_1$ and $\mathcal{C}_2$.
\begin{lemma} \label{lem-supleftright}
One can store $S$ in a basic data structure $\mathcal{C}_1$ such that given a point $a$, the rightmost (resp., leftmost) point in $S$ to the left (resp., right) of $a$ can be reported in $\widetilde{O}(1)$ time (if they exist).
\end{lemma}
\begin{lemma} \label{lem-suphitint}
One can store $S$ in a basic data structure $\mathcal{C}_2$ such that given an interval $I$, a point $a \in S$ that is contained in $I$ can be reported in $\widetilde{O}(1)$ time (if it exists).
\end{lemma}

\paragraph{Maintaining a solution.}
Let $(S_0,\mathcal{I}_0)$ be the initial $(S,\mathcal{I})$, and define $\mathcal{I}_0' \subseteq \mathcal{I}_0$ as the sub-collection consisting of all intervals that are hit by $S_0$.
In the construction of $\mathcal{D}$, after building the data structures $\mathcal{A}$ and $\mathcal{B}$, we compute an optimal hitting set $S^*$ for $(S_0,\mathcal{I}_0')$.
Set $\mathsf{cnt} = 0$ and $\mathsf{opt}^\sim = |S^*|$ initially.
The data structure $\mathcal{D}$ handles the operations as follows.
After each operation, $\mathcal{D}$ first increment $\mathsf{cnt}$ by 1 and query the data structure $\mathcal{B}$ to see whether the current $(S,\mathcal{I})$ has a hitting set.
Whenever $\mathsf{cnt} \geq \varepsilon \cdot \mathsf{opt}^\sim /(2+\varepsilon)$ and the current $(S,\mathcal{I})$ has a hitting set, we re-compute an optimal hitting set $S^*$ for the current $(S,\mathcal{I})$ using the output-sensitive algorithm and reset $\mathsf{cnt} = 0$ and $\mathsf{opt}^\sim = |S^*|$.
Otherwise, we apply local modification on $S^*$ (i.e., ``slightly'' modify $S^*$) as follows.
\begin{itemize}
    \item If the operation inserts a point $a$ to $S$, then include $a$ in $S^*$.
    If the operation deletes a point $a$ from $S$, then remove $a$ from $S^*$ (if $a \in S^*$) and include in $S^*$ the rightmost point in $S$ to the left of $a$ and the leftmost point in $S$ to the right of $a$ (if they exist), which can be found using the data structure $\mathcal{C}_1$ of Lemma~\ref{lem-supleftright}.
    \item If the operation inserts an interval $I$ to $\mathcal{I}$, then include in $S^*$ an arbitrary point in $S$ that hits $I$ (if it exists), which can be found using the data structure $\mathcal{C}_2$ of Lemma~\ref{lem-suphitint}.
    If the operation deletes an interval from $\mathcal{I}$, keep $S^*$ unchanged.
\end{itemize}
\noindent
Note that we do local modification on $S^*$ no matter whether the current $(S,\mathcal{I})$ has a hitting set or not.
After this, our data structure $\mathcal{D}$ uses $S^*$ as the solution for the current $(S,\mathcal{I})$, if the current $(S,\mathcal{I})$ has a hitting set.
In order to support the desired queries for $S^*$, we can simply store $S^*$ in a (dynamic) binary search tree.
When we re-compute $S^*$, we re-build the binary search tree that stores $S^*$.
For each local modification on $S^*$, we do insertion/deletion on the binary search tree according to the change of $S^*$.

\paragraph{Correctness.}
To see the correctness of our data structure, we have to show that $S^*$ is always a $(1+\varepsilon)$-approximate optimal hitting set for the current $(S,\mathcal{I})$, whenever $(S,\mathcal{I})$ has a hitting set.
To this end, we observe some simple facts.
First, the counter $\mathsf{cnt}$ is exactly the number of the operations processed after the last re-computation of $S^*$, and the number $\mathsf{opt}^\sim$ is the optimum at the last re-computation of $S^*$.
Second, after each local modification, the size of $S^*$ either increases by 1 or keeps unchanged.
Based on these two facts and the stability of dynamic interval hitting set (Lemma~\ref{lem-stab}), we can easily bound the size of $S^*$.
\begin{lemma} \label{lem-inthscorrect}
We have $|S^*| \leq (1+\varepsilon) \cdot \mathsf{opt}$ at any time.
\end{lemma}

\noindent
With the above lemma in hand, it now suffices to show that $S^*$ is a feasible hitting set.
\begin{lemma} \label{lem-inthscorrect2}
Whenever the current $(S,\mathcal{I})$ has a hitting set, $S^*$ is a hitting set for $(S,\mathcal{I})$.
\end{lemma}

\paragraph{Time analysis.}
The construction time of $\mathcal{D}$ is clearly $\widetilde{O}(n_0)$.
Specifically, when constructing $\mathcal{D}$, we need to build the data structures $\mathcal{A},\mathcal{B},\mathcal{C}_1,\mathcal{C}_2$, all of which can be built in $\widetilde{O}(n_0)$ time.
Also, we need to compute an optimal hitting set for $(S_0,\mathcal{I}_0')$ using the output-sensitive algorithm, which can also be done in $\widetilde{O}(n_0)$ time.
We show that the amortized update time of $\mathcal{D}$ is $\widetilde{O}(1/\varepsilon)$.
After each operation, we need to update the data structures $\mathcal{A},\mathcal{B},\mathcal{C}_1,\mathcal{C}_2$, which takes $\widetilde{O}(1)$ time.
If we do local modification on $S^*$, then it takes $\widetilde{O}(1)$ time.
Indeed, in each local modification, there are $O(1)$ elements to be added to or removed from $S^*$.
These elements can be found in $\widetilde{O}(1)$ time using the data structures $\mathcal{C}_1$ and $\mathcal{C}_2$.
After these elements are found, we need to modify the binary search tree that stores $S^*$ for supporting the desired queries, which can also be done in $\widetilde{O}(1)$ time.
In the case where we re-compute a new $S^*$, it takes $\widetilde{O}(\mathsf{opt})$ time.
However, we can amortize this time over the operations in between the last re-computation of $S^*$ and the current one.
Let $\mathsf{opt}^\sim$ be the quasi-optimum of $(S,\mathcal{I})$ at the last re-computation.
Suppose there are $t$ operations in between the two re-computations.
Then $t \geq \varepsilon \cdot \mathsf{opt}^\sim/(2+\varepsilon)$.
By the stability of dynamic interval hitting set (Lemma~\ref{lem-stab}), we have $\mathsf{opt} \leq \mathsf{opt}^\sim + t$, because $\mathsf{opt}$ is also the quasi-optimum of the current $(S,\mathcal{I})$.
It follows that
\begin{equation*}
    \frac{\mathsf{opt}}{t} \leq \frac{\mathsf{opt}^\sim + t}{t} \leq \frac{2+\varepsilon}{\varepsilon} +1 = O(1/\varepsilon).
\end{equation*}
Thus, the amortized time for the current re-computation is $\widetilde{O}(1/\varepsilon)$.
We conclude the following.
\begin{theorem}
For a given approximation factor $\eps>0$, there is a $(1+\varepsilon)$-approximate dynamic interval hitting data structure $\mathcal{D}$ with $\widetilde{O}(1/\eps)$ amortized update time and $\widetilde{O}(n_0)$ construction time.
\end{theorem}

\subsection{Dynamic set cover and hitting set in the partially dynamic settings}
The idea of local modification can also be applied to dynamic geometric set cover (resp., hitting set) problems in the partially dynamic setting (which are stable by Lemma~\ref{lem-stab}), as long as we have an output-sensitive algorithm for the problem and an efficient dynamic data structure that can indicate whether the current instance has a feasible solution or not.

For an example, let us consider dynamic quadrant set cover in the partially dynamic setting.
Let $(S,\mathcal{Q})$ be a dynamic quadrant set cover instance with only point updates.
As before, we denote by $n$ (resp., $n_0$) the size of the current (resp., initial) $(S,\mathcal{Q})$, and by $\mathsf{opt}$ the optimum of the current $(S,\mathcal{Q})$.

\paragraph{The data structure.}
Our data structure $\mathcal{D}$ consists of three parts.
The first part is the basic data structure $\mathcal{A}$ required for performing the output-sensitive algorithm of Theorem~\ref{thm-osquad}.
The second part is the following dynamic data structure $\mathcal{B}$ which can tell us whether the current instance has a hitting set.
\begin{lemma} \label{lem-dsexsc}
    There exists a dynamic data structure $\mathcal{B}$ built on $(S,\mathcal{Q})$ with $\widetilde{O}(1)$ update time and $\widetilde{O}(n_0)$ construction time that can indicate whether the current $(S,\mathcal{Q})$ has a set cover or not.
\end{lemma}

\noindent
The third part is a simple (static) data structure $\mathcal{C}$ built on $\mathcal{Q}$ that can report in $\widetilde{O}(1)$ time, for a given point $a \in \mathbb{R}^2$, a quadrant in $\mathcal{Q}$ that contains $a$ (if it exists); this can be done using a standard orthogonal range-search data structure, e.g., range tree, whose construction time is $\widetilde{O}(|\mathcal{Q}|)$.

\paragraph{Maintaining a solution.}
Let $S_0$ be the initial $S$, and define $S_0' \subseteq S_0$ as the subset consisting of the points covered by $\mathcal{Q}$.
At the beginning, we compute a $\mu$-approximate optimal set cover $\mathcal{Q}^*$ for $(S_0',\mathcal{Q})$, where $\mu$ is the approximation ratio of the output-sensitive algorithm.
This can be done as follows.
We first compute $S_0'$ by testing whether each point in $S$ is covered by some range in $\mathcal{Q}$ using the data structure $\mathcal{C}$, and then build on $(S_0',\mathcal{Q})$ the basic data structure required for the output-sensitive algorithm and perform the output-sensitive algorithm on $(S_0',\mathcal{Q})$.
Set $\mathsf{cnt} = 0$ and $\mathsf{opt}^\sim = |\mathcal{Q}^*|$.
After each iteration, $\mathcal{D}$ first increments $\mathsf{cnt}$ by 1 and checks whether the current $(S,\mathcal{Q})$ has a set cover or not using the data structure $\mathcal{B}$.
Whenever $\mathsf{cnt} \geq (\varepsilon/\mu) \cdot \mathsf{opt}^\sim / (2+\varepsilon)$ and the current $(S,\mathcal{Q})$ has a set cover, we re-compute a $\mu$-approximate optimal set cover $\mathcal{Q}^*$ for the current $(S,\mathcal{Q})$ using the output-sensitive algorithm, and reset $\mathsf{cnt} = 0$ and $\mathsf{opt}^\sim = |\mathcal{Q}^*|$.
Otherwise, we apply local modification on $\mathcal{Q}^*$ as follows.
If the operation inserts a point $a$ to $S$, then include in $\mathcal{Q}^*$ an arbitrary range in $\mathcal{Q}$ that contains $a$, which can be found using the data structure $\mathcal{C}$ (if it exists).
If the operation deletes a point $a$ from $S$, then keep $\mathcal{Q}^*$ unchanged.
After this, our data structure $\mathcal{D}$ uses $\mathcal{Q}^*$ as the solution for the current $(S,\mathcal{Q})$.
As in the last section, we can store $\mathcal{Q}^*$ in a (dynamic) binary search tree to support the desired queries (by fixing some order on the family of quadrants).

\paragraph{Correctness.}
Using the same argument as in the last section, we can show that $|\mathcal{Q}^*| \leq (\mu+\varepsilon) \cdot \mathsf{opt}$ at any time.
Furthermore, it is also easy to see that $\mathcal{Q}^*$ is always a feasible set cover.
\begin{lemma} \label{lem-quadsccorrect}
Whenever the current $(S,\mathcal{Q})$ has a set cover, $\mathcal{Q}^*$ is a set cover for $(S,\mathcal{Q})$.
\end{lemma}

\paragraph{Time analysis.}
Using the same analysis as in the last section, we can show that the construction time of the data structure $\mathcal{D}$ is $\widetilde{O}(n)$, and the amortized update time is $\widetilde{O}(1/\varepsilon)$.
Setting $\varepsilon$ to be any constant, we conclude the following.
\begin{theorem}
There exists an $O(1)$-approximate partially dynamic quadrant set cover data structure with $\widetilde{O}(1)$ amortized update time and $\widetilde{O}(n_0)$ construction time.
\end{theorem}
Using the reduction of Lemma~\ref{lem-sqreduction}, we also have the following result.
Note that although the reduction of Lemma~\ref{lem-sqreduction} is described for the fully dynamic setting, it actually applies to reduce dynamic unit-square set cover (resp., hitting set) in the partially dynamic setting and dynamic quadrant hitting set in the partially dynamic setting to dynamic quadrant set cover data structure in the partially dynamic setting.
\begin{corollary}
There exist an $O(1)$-approximate partially dynamic unit-square set cover, partially dynamic unit-square hitting set, and partially dynamic quadrant hitting set data structures with $\widetilde{O}(1)$ amortized update time and $\widetilde{O}(n_0)$ construction time.
\end{corollary}

\bibliography{my_bib}

\appendix

\section{Missing proofs} \label{appx-proof}

In this section, we give the proofs of the technical lemmas in Section~\ref{sec-DISC} and Section~\ref{sec-DQSC}, which are missing in the main body of the paper.

\subsection{Proof of Lemma~\ref{lem-osint}}
Consider an interval set cover instance $(S,\mathcal{I})$ of size $n$.
\begin{lemma} \label{lem-aaa}
One can store $S$ in a basic data structure that can report in $\widetilde{O}(1)$ time, for a query point $q \in \mathbb{R}$, the leftmost point in $S$ that is to the right of $q$ (if it exists).
\end{lemma}
\begin{proof}
We simply store $S$ in a binary search tree.
Given a query point $q \in \mathbb{R}$, one can find the leftmost point in $S$ to the right of $q$ by searching in the tree in $\widetilde{O}(1)$ time.
Clearly, the tree can be built in $\widetilde{O}(|S|)$ time and dynamized with $\widetilde{O}(1)$ update time, and hence it is basic.
\end{proof}
\begin{lemma} \label{lem-bbb}
One can store $\mathcal{I}$ in a basic data structure that can report in $\widetilde{O}(1)$ time, for a query point $q \in \mathbb{R}$, the interval in $\mathcal{I}$ with the rightmost right endpoint that contains $q$ (if it exists).
\end{lemma}
\begin{proof}
We store $\mathcal{I}$ in a binary search tree $\mathcal{T}$ where the key of each interval is its left endpoint.
We augment each node $\mathbf{u} \in \mathcal{T}$ with an additional field which stores the interval in the subtree rooted at $\mathbf{u}$ that has the rightmost right endpoint.
Given a query point $q \in \mathbb{R}$, we first look for the interval in $\mathcal{T}$ with the rightmost right endpoint whose key is smaller than or equal to $q$.
With the augmented fields, this interval can be found in $\widetilde{O}(1)$ time.
If this interval contains $q$, then we report it; otherwise, no interval in $\mathcal{I}$ contains $q$.
Clearly, $\mathcal{T}$ can be built in $\widetilde{O}(|\mathcal{I}|)$ time and dynamized with $\widetilde{O}(1)$ update time, and hence it is basic.
\end{proof}
\noindent
We store $S$ in the data structure of Lemma~\ref{lem-aaa} and $\mathcal{I}$ in the data structure of Lemma~\ref{lem-bbb}.

Recall the greedy algorithm for the static interval set cover problem.
The algorithm computes an optimal set cover $\mathcal{I}_\text{opt}$ for $(S,\mathcal{I})$ by repeatedly picking the leftmost uncovered point and covering it using the right-rightmost interval.
Formally, the procedure is presented below.
\begin{enumerate}
    \item $q \leftarrow -\infty$ and $\mathcal{I}_\text{opt} \leftarrow \emptyset$.
    \item Find the leftmost point $a \in S$ to the right of $q$.
    \item $\mathcal{I}_\text{opt} \leftarrow \mathcal{I}_\text{opt} \cup \{I_a\}$, where $I_a$ is the right-rightmost interval in $\mathcal{I}$ covering $a$.
    \item $q \leftarrow$ the right endpoint of $I_a$, then go to Step~2.
\end{enumerate}
With the data structures of Lemma~\ref{lem-aaa} and Lemma~\ref{lem-bbb} in hand, the greedy algorithm can be performed in $\widetilde{O}(\mathsf{opt})$ time, where $\mathsf{opt}$ is the optimum of $(S,\mathcal{I})$.
Indeed, the data structure of Lemma~\ref{lem-aaa} can be used to implement Step~2 and the data structure of Lemma~\ref{lem-bbb} can be used to implement Step~3, both in $\widetilde{O}(1)$ time.
Since the algorithm terminates after $O(\mathsf{opt})$ iterations, the total time cost is $\widetilde{O}(\mathsf{opt})$.

\subsection{Proof of Lemma~\ref{lem-intnoans}}
Our algorithm makes a no-solution decision iff $(S_i,\mathcal{I}_i)$ has no set cover for some $i \in P'$.
If $(S_i,\mathcal{I}_i)$ has a set cover for all $i \in P'$, then $(S,\mathcal{I})$ clearly has a set cover, because the portions $J_i$ for $i \in P$ are coverable.
If $(S_i,\mathcal{I}_i)$ has no set cover for some $i \in P'$, then $(S,\mathcal{I})$ has no set cover, because $J_i$ is uncoverable and thus the points in $S_i$ can only be covered by the intervals in $\mathcal{I}_i$.
Therefore, the no-solution decision of our data structure is correct.

\subsection{Proof of Lemma~\ref{lem-atmost2}}
Suppose the portions $J_1,\dots,J_r$ are sorted from left to right.
Let $s_i$ be the separation point of $J_i$ and $J_{i+1}$.
Observe that an interval $I \in \mathcal{I}_\text{opt}$ belongs to exactly two $\mathcal{I}_i$'s only if $I$ contains one of the separation points $s_1,\dots,s_{r-1}$.
We claim that for each $s_i$, at most two intervals in $\mathcal{I}_\text{opt}$ contain $s_i$.
Assume there are three intervals $I^-,I,I^+$ that contain $s_i$.
Without loss of generality, assume that $I^-$ (resp., $I^+$) has the leftmost left endpoint (resp., the rightmost right endpoint) among $I^-,I,I^+$.
Then one can easily see that $I \subseteq I^- \cup I^+$.
Therefore, $\mathcal{I}_\text{opt} \backslash \{I\}$ is also a set cover for $(S,\mathcal{I})$, contradicting the optimality of $\mathcal{I}_\text{opt}$.
Thus, at most two intervals in $\mathcal{I}_\text{opt}$ contain $s_i$.
It follows that there are at most $2(r-1)$ intervals in $\mathcal{I}_\text{opt}$ that contain some separation point, and only these intervals can belong to exactly two $\mathcal{I}_i$'s, which proves the lemma.

\subsection{Proof of Lemma~\ref{lem-sqreduction}}
We first observe that unit-square set cover and hitting set are equivalent.
A unit square $Z$ contains a point $a \in \mathbb{R}^2$ iff the unit square centered at $a$ contains the center of $Z$.
Therefore, given a unit-square set cover instance $(S,\mathcal{Z})$, if we replace each point $a \in S$ with a unit square $Z_a$ centered at $a$ and replace each unit square $Z \in \mathcal{Z}$ with the center $c_Z$ of $Z$, then $(S,\mathcal{Z})$ is reduced to an equivalent unit-square hitting set instance $(\{c_Z: Z \in \mathcal{Z}\},\{Z_a: a \in S\})$.
In the same way, one can also reduce a unit-square hitting set instance to an equivalent unit-square set cover instance.
Thus, we only need to consider unit-square set cover.

In order to reduce dynamic unit-square set cover to dynamic quadrant set cover, we apply a grid technique.
To this end, let us first consider a static unit-square set cover instance $(S,\mathcal{Z})$.
We build a grid $\varGamma$ on the plane consisting of square cells of side-length 1.
For each cell $\Box$ of the grid, we define $S_\Box = S \cap \Box$ and $\mathcal{Z}_\Box = \{Z \in \mathcal{Z}: Z \cap \Box \neq \emptyset\}$.
We call a cell \textit{nonempty} if $S_\Box \neq \emptyset$ or $\mathcal{Z}_\Box \neq \emptyset$, and \textit{truly nonempty} if $S_\Box \neq \emptyset$.
Let $\varPsi$ be the set of all truly nonempty cells.
It is clear that $(S,\mathcal{Z})$ has a set cover iff $(S_\Box,\mathcal{Z}_\Box)$ has a set cover for all $\Box \in \varPsi$, since the points in $S_\Box$ can only be covered using the unit squares in $\mathcal{Z}_\Box$.
Let $\mathcal{Z}_\Box^* \subseteq \mathcal{Z}_\Box$ be a $c$-approximate optimal set cover for $(S_\Box,\mathcal{Z}_\Box)$.
We claim that $\bigsqcup_{\Box \in \varPsi} \mathcal{Z}_\Box^*$ is a $4c$-approximate optimal set cover for $(S,\mathcal{Z})$.
Consider an optimal set cover $\mathcal{Z}^*$ for $(S,\mathcal{Z})$.
We have $|\mathcal{Z}_\Box^*| \leq c \cdot |\mathcal{Z}^* \cap \mathcal{Z}_\Box|$, because $\mathcal{Z}^* \cap \mathcal{Z}_\Box$ is a set cover for $(S_\Box,\mathcal{Z}_\Box)$.
On the other hand, we have $\sum_{\Box \in \varPsi} |\mathcal{Z}^* \cap \mathcal{Z}_\Box| \leq 4 |\mathcal{Z}^*|$, since any unit square can intersect at most four cells of $\varGamma$.
It follows that
\begin{equation*}
    \left| \bigsqcup_{\Box \in \varPsi} \mathcal{Z}_\Box^* \right| = \sum_{\Box \in \varPsi} |\mathcal{Z}_\Box^*| \leq \sum_{\Box \in \varPsi} c \cdot |\mathcal{Z}^* \cap \mathcal{Z}_\Box| \leq 4c \cdot |\mathcal{Z}^*|.
\end{equation*}
In this way, we reduce the instance $(S,\mathcal{Z})$ to the instances $(S_\Box,\mathcal{Z}_\Box)$ for $\Box \in \varPsi$.
It seems that solving each $(S_\Box,\mathcal{Z}_\Box)$ is still a unit-square set cover problem.
However, it is in fact a quadrant set cover problem.
Indeed, a cell $\Box$ is a square of side-length 1, hence the intersection of any unit square $Z$ and $\Box$ is equal to the intersection of a quadrant $Q_Z$ and $\Box$, where $Q_Z$ is the quadrant obtained by removing the two edges of $Z$ outside $\Box$.
Since the points in $S_\Box$ are all contained in $\Box$, the quadrant $Q_Z \in \mathcal{Z}_\Box$ covers the same set of points in $S_\Box$ as the unit square $Z$.
Therefore, $(S_\Box,\mathcal{Z}_\Box)$ is equivalent to the quadrant set cover instance $(S_\Box,\{Q_Z: Z \in \mathcal{Z}_\Box\})$.

The above observation gives us a reduction from dynamic unit-square set cover to dynamic quadrant set cover.
Suppose we have a $c$-approximate dynamic quadrant set cover data structure $\mathcal{D}$ with $f(n)$ amortized update time and $\widetilde{O}(n_0)$ construction time, where $f$ is an increasing function.
We are going to design a $4c$-approximate dynamic unit-square set cover data structure $\mathcal{D}'$.
Let $(S,\mathcal{Z})$ be a dynamic unit-square set cover instance.
When constructing $\mathcal{D}'$, we build the grid $\varGamma$ described above and compute $(S_\Box,\mathcal{Z}_\Box)$ for all nonempty cells.
The grid $\varGamma$ is always fixed, while the instances $(S_\Box,\mathcal{Z}_\Box)$ change along with the change of $S$ and $\mathcal{Z}$.
As argued before, we can regard each $(S_\Box,\mathcal{Z}_\Box)$ as a dynamic quadrant set cover instance.
We then construct a data structure $\mathcal{D}_\Box$ for each $(S_\Box,\mathcal{Z}_\Box)$, which is the dynamic quadrant set cover data structure $\mathcal{D}$ built on $(S_\Box,\mathcal{Z}_\Box)$.
If $(S_\Box,\mathcal{Z}_\Box)$ has a set cover, then $\mathcal{D}_\Box$ maintains a $c$-approximate optimal set cover for $(S_\Box,\mathcal{Z}_\Box)$, which we denote by $\mathcal{Z}_\Box^*$.
Furthermore, we also need two data structures to maintain the nonempty and truly nonempty cells of $\varGamma$, respectively.
We store the nonempty cells in a (dynamic) binary search tree $\mathcal{T}_1$ (by fixing some total order of the cells of $\varGamma$).
Each node $\mathbf{u} \in \mathcal{T}_1$ has a pointer pointing to the data structure $\mathcal{D}_\Box$ where $\Box$ is the cell corresponding to $\mathbf{u}$.
Similarly, we store all the truly nonempty cells in a (dynamic) binary search tree $\mathcal{T}_2$, with a pointer at each node pointing to the data structure $\mathcal{D}_\Box$ for the corresponding cell $\Box$.
The data structures $\mathcal{D}_\Box$ and the binary search trees $\mathcal{T}_1,\mathcal{T}_2$ form our dynamic unit-square set cover data structure $\mathcal{D}'$.
Besides this, $\mathcal{D}'$ also maintains two numbers $m$ and $s$, where $m$ is the number of the truly nonempty cells $\Box$ such that $(S_\Box,\mathcal{Z}_\Box)$ has no set cover and $s$ is the sum of $|\mathcal{Z}_\Box^*|$ over all truly nonempty cells $\Box$ such that $(S_\Box,\mathcal{Z}_\Box)$ has a set cover.
The construction time of $\mathcal{D}'$ is $\widetilde{O}(n_0)$.
Indeed, finding the nonempty (and truly nonempty) cells and building the binary search trees $\mathcal{T}_1,\mathcal{T}_2$ can be easily done in $\widetilde{O}(n_0)$ time.
Furthermore, the sum of the sizes of all $(S_\Box,\mathcal{Z}_\Box)$ is $\widetilde{O}(n_0)$, since each point in $S$ lies in exactly one cell and each unit-square in $\mathcal{Z}$ intersects at most four cells.
Therefore, computing the instances $(S_\Box,\mathcal{Z}_\Box)$ can be done in $\widetilde{O}(n_0)$ time.
The time for building each data structure $\mathcal{D}_\Box$ is near-linear in the size of $(S_\Box,\mathcal{Z}_\Box)$, thus building all $\mathcal{D}_\Box$ takes $\widetilde{O}(n_0)$ time.

Next, we consider how to handle updates on $(S,\mathcal{Z})$.
For an update on $(S,\mathcal{Z})$, there are (at most) four cells $\Box$ for which $(S_\Box,\mathcal{Z}_\Box)$ changes, and we call them \textit{involved} cells.
For each involved cell $\Box$, we do the following.
First, if $\Box$ is previously empty (resp., not truly nonempty) and becomes nonempty (resp., truly nonempty) after the update, we insert a new node to the binary search tree $\mathcal{T}_1$ (resp., $\mathcal{T}_2$) corresponding to $\Box$.
Conversely, if $\Box$ is previously nonempty (resp., truly nonempty) and becomes empty (resp., not truly nonempty) after the update, we delete the node corresponding to $\Box$ from $\mathcal{T}_1$ (resp., $\mathcal{T}_2$).
After this, we need to update the data structure $\mathcal{D}_\Box$.
If $\Box$ is previously empty and becomes nonempty after the update, we create a new data structure $\mathcal{D}_\Box$ for $\Box$, otherwise we find the data structure $\mathcal{D}_\Box$ (which can be done by searching the node $\mathbf{u}$ corresponding to $\Box$ in $\mathcal{T}_1$ and using the pointer at $\mathbf{u}$) and update it.
We observe that all the above work can be done in $\widetilde{O}(f(n))$ amortized time.
Indeed, updating the binary search trees $\mathcal{T}_1$ and $\mathcal{T}_2$ takes logarithmic time.
By assumption, the amortized update time of $\mathcal{D}_\Box$ is $f(n_\Box)$, where $n_\Box$ is the size of the current $(S_\Box,\mathcal{Z}_\Box)$, which is smaller than or equal to $f(n)$ as $f$ is an increasing function.
Since there are at most four involved cells, the amortized update time of $\mathcal{D}'$ is $\widetilde{O}(f(n))$.

At any point, if $m > 0$, then $\mathcal{D}'$ directly decides that the current $(S,\mathcal{Z})$ has no set cover.
Otherwise, $\mathcal{D}'$ uses $\mathcal{Z}_\text{appx} = \bigsqcup_{\Box \in \varPsi} \mathcal{Z}_\Box^*$ as the set cover solution for the current $(S,\mathcal{Z})$, where $\varPsi$ is the set of the truly nonempty cells.
By our previous observation, $\mathcal{Z}_\text{appx}$ is a $4c$-approximate optimal set cover for $(S,\mathcal{Z})$.
To support the desired queries for $\mathcal{Z}_\text{appx}$ is quite easy.
First, the size of $\mathcal{Z}_\text{appx}$ is just the number $s$ maintained by $\mathcal{D}'$, hence the size query for $\mathcal{Z}_\text{appx}$ can be answered in $O(1)$ time.
To answer a membership query, let $Z \in \mathcal{Z}$ be the query element.
We want to know the multiplicity of $Z$ in $\mathcal{Z}_\text{appx}$.
There are at most four cells of $\varGamma$ that intersect $Z$.
For each such cell $\Box$, we search in the binary search tree $\mathcal{T}_2$.
If we cannot find a node of $\mathcal{T}_2$ corresponding to $\Box$, then $\Box \notin \varPsi$.
Otherwise, the node $\mathbf{u} \in \mathcal{T}_2$ corresponding to $\Box$ gives us a pointer to the data structure $\mathcal{D}_\Box$, which can report the multiplicity of $Z$ in $\mathcal{Z}_\Box^*$ in $O(\log |\mathcal{Z}_\Box^*|)$ time.
The sum of all these multiplicities is just the multiplicity of $Z$ in $\mathcal{Z}_\text{appx}$.
The time cost for answer the membership query is $O(\log |\varPsi| + \log |\mathcal{Z}_\text{appx}|)$, since the size of $\mathcal{T}_2$ is $|\varPsi|$.
Note that $|\varPsi| \leq |\mathcal{Z}_\text{appx}|$, since $\mathcal{Z}_\Box^* \neq \emptyset$ for all $\Box \in \varPsi$ (for otherwise $\mathcal{Z}_\Box^*$ is not a set cover for $S_\Box$).
Thus, a membership query takes $O(\log |\mathcal{Z}_\text{appx}|)$ time.
To answer the reporting query, we do a traversal in $\mathcal{T}_2$.
For each node $\mathbf{u} \in \mathcal{T}_2$, we use the data structure $\mathcal{D}_\Box$ to report the elements in $\mathcal{Z}_\Box^*$ in $O(|\mathcal{Z}_\Box^*|)$ time, where $\Box \in \varPsi$ is the truly nonempty cell corresponding to $\mathbf{u}$.
The time cost is $O(|\varPsi| + |\mathcal{Z}_\text{appx}|)$.
Again, since $|\varPsi| \leq |\mathcal{Z}_\text{appx}|$, the reporting query can be answered in $O(|\mathcal{Z}_\text{appx}|)$ time.
To summarize, $\mathcal{D}'$ is a $4c$-approximate dynamic unit-square set cover data structure with $\widetilde{O}(f(n))$ amortized update time and $\widetilde{O}(n_0)$ construction time.

Finally, we show how to reduce dynamic quadrant hitting set to dynamic quadrant set cover.
Let us first consider a static quadrant hitting set instance $(S,\mathcal{Q})$.
There are four types of quadrants in $\mathcal{Q}$, southeast, southwest, northeast, and northwest; we denote by $\mathcal{Q}^\text{SE},\mathcal{Q}^\text{SW},\mathcal{Q}^\text{NE},\mathcal{Q}^\text{NW} \subseteq \mathcal{Q}$ the sub-collections consisting of these types of quadrants, respectively.
Let $S_1,S_2,S_3,S_4$ be $c$-approximate optimal hitting sets for the instances $(S,\mathcal{Q}^\text{SE}),(S,\mathcal{Q}^\text{SE}),(S,\mathcal{Q}^\text{SE}),(S,\mathcal{Q}^\text{SE})$, respectively.
We claim that $\bigcup_{i=1}^4 S_i$ is a $4c$-approximate hitting set for $(S,\mathcal{Q})$.
Indeed, we have $|S_1| \leq c \cdot \mathsf{opt}^{SE} \leq c \cdot \mathsf{opt}$, where $\mathsf{opt}^{SE}$ is the optimum of $(S,\mathcal{Q}^\text{SE})$ and $\mathsf{opt}$ is the optimum of $(S,\mathcal{Q})$.
Similarly, we can show that $|S_i| \leq c \cdot \mathsf{opt}$ for all $i \in \{1,\dots,4\}$.
Thus, $|\bigcup_{i=1}^4 S_i| = \sum_{i=1}^4 |S_i| \leq 4c \cdot \mathsf{opt}$.
Therefore, if we can solve dynamic quadrant hitting set for a single type of quadrants with an approximation factor $c$, we are able to solve the general dynamic quadrant hitting set problem with an approximation factor $4c$.

It is easy to see that dynamic quadrant hitting set for a single type of quadrants is equivalent to dynamic quadrant set cover for a single type of quadrants.
To see this, let us consider southeast quadrants.
Note that a point $a$ is contained in a southeast quadrant $Q$ iff the northwest quadrant whose vertex is $a$ (denoted by $Q_a$) contains the vertex $v_Q$ of $Q$.
Therefore, given a quadrant hitting set instance $(S,\mathcal{Q}^\text{SE})$ where $\mathcal{Q}^\text{SE}$ only contains southeast quadrants, if we replace each point $a \in S$ with the northwest quadrant $Q_a$ and replace each southeast quadrant $Q \in \mathcal{Q}^\text{SE}$ with its vertex $v_Q$, then $(S,\mathcal{Q}^\text{SE})$ is reduced to an equivalent quadrant set cover instance $(\{v_Q: Q \in \mathcal{Q}^\text{SE}\},\{Q_a: a \in S\})$.
To summarize, if we can solve dynamic quadrant set cover with an approximation factor $c$, then we can also solve dynamic quadrant hitting set for a single type of quadrants with an approximation factor $c$ and in turn solve the general dynamic quadrant hitting set problem with an approximation factor $4c$.
This completes the proof.

\subsection{Proof of Lemma~\ref{lem-noansquad}}
Our algorithm makes a no-solution decision iff $(S_{i,j},\mathcal{Q}_{i,j})$ has no set cover for some $(i,j) \in P'$.
If $(S_{i,j},\mathcal{Q}_{i,j})$ has a set cover for all $(i,j) \in P'$, then $(S,\mathcal{Q})$ clearly has a set cover, because the cells $\Box_{i,j}$ for $(i,j) \in P$ are coverable.
Suppose $(S_{i,j},\mathcal{Q}_{i,j})$ has no set cover for some $(i,j) \in P'$.
Then there exists a point $a \in S_{i,j}$ that is not covered by any quadrant in $\mathcal{Q}_{i,j}$.
We claim that $a$ is not covered by any quadrant in $\mathcal{Q}$.
Consider a quadrant $Q \in \mathcal{Q}$.
If $Q$ does not intersect $\Box_{i,j}$, then $a \notin Q$.
Otherwise, $Q$ must partially intersect $\Box_{i,j}$, because $\Box_{i,j}$ is uncoverable.
If the vertex of $Q$ lies in $\Box_{i,j}$, then $Q \in \mathcal{Q}_{i,j}$ and thus $a \notin Q$.
The remaining case is that $Q$ partially intersects $\Box_{i,j}$ and contains one edge of $\Box_{i,j}$.
Without loss of generality, assume $Q$ left intersects $\Box_{i,j}$.
The rightmost quadrant $Q'$ that left intersects $\Box_{i,j}$ is contained in $\mathcal{Q}_{i,j}$.
Since $Q \cap \Box_{i,j} \subseteq Q' \cap \Box_{i,j}$ and $a \notin Q'$, we have $a \notin Q$.
It follows that $a$ is not covered by $\mathcal{Q}$ and hence $(S,\mathcal{Q})$ has no set cover.
Therefore, the no-solution decision of our data structure is correct.

\subsection{Proof of Lemma~\ref{lem-quadopt}}
Fix $(i,j) \in P'$.
Let $\mathcal{Q}^\sim \subseteq \mathcal{Q}_\textnormal{opt}$ consist of all quadrants in $\mathcal{Q}_\textnormal{opt}$ that intersect $\Box_{i,j}$.
Then $\mathcal{Q}^\sim$ covers all points in $S_{i,j}$, because the points in $S_{i,j}$ can only be covered by quadrants that intersect $\Box_{i,j}$.
Since $\Box_{i,j}$ is uncoverable by $\mathcal{Q}$, all quadrants in $\mathcal{Q}^\sim$ must partially intersect $\Box_{i,j}$.
It follows that a quadrant in $\mathcal{Q}^\sim$ either has its vertex in $\Box_{i,j}$ or contains one edge of $\Box_{i,j}$.
We claim that if a point $a \in \Box_{i,j}$ is covered by $\mathcal{Q}^\sim$, then $a$ is covered by either $\mathcal{Q}^\sim \cap \mathcal{Q}_{i,j}'$ or a special quadrant in $\mathcal{Q}_{i,j}$.
Let $Q \in \mathcal{Q}^\sim$ be a quadrant such that $a \in Q$.
If $Q \in \mathcal{Q}_{i,j}'$, we are done.
Otherwise, $Q$ must contain one edge of $\Box_{i,j}$, say the left edge.
In this case, $\mathcal{Q}_{i,j}$ should contain a special quadrant $Q_0$, which is the rightmost quadrant that left intersects $\Box_{i,j}$.
It is clear that $Q \cap \Box_{i,j} \subseteq Q_0 \cap \Box_{i,j}$, which implies that $a \in Q_0$.
From this claim, it directly follows that all points in $S_{i,j}$ are covered by $(\mathcal{Q}^\sim \cap \mathcal{Q}_{i,j}') \cup \mathcal{Q}_{i,j}''$, where $\mathcal{Q}_{i,j}'' \subseteq \mathcal{Q}_{i,j}$ consists of the special quadrants.
Since $(\mathcal{Q}^\sim \cap \mathcal{Q}_{i,j}') \cup \mathcal{Q}_{i,j}'' \subseteq \mathcal{Q}_{i,j}$ and $(\mathcal{Q}^\sim \cap \mathcal{Q}_{i,j}') \cup \mathcal{Q}_{i,j}''$ covers $S_{i,j}$, we have $|(\mathcal{Q}^\sim \cap \mathcal{Q}_{i,j}') \cup \mathcal{Q}_{i,j}''| \leq \mathsf{opt}_{i,j}$.
Note that $\mathcal{Q}^\sim \cap \mathcal{Q}_{i,j}' = \mathcal{Q}_\mathsf{opt} \cap \mathcal{Q}_{i,j}'$ and $|\mathcal{Q}_{i,j}''| \leq 4$.
So we have
\begin{equation*}
    |\mathcal{Q}_\textnormal{opt} \cap \mathcal{Q}_{i,j}'| + 4 = |\mathcal{Q}^\sim \cap \mathcal{Q}_{i,j}'| + 4 \geq |(\mathcal{Q}^\sim \cap \mathcal{Q}_{i,j}') \cup \mathcal{Q}_{i,j}''| \geq \mathsf{opt}_{i,j}.
\end{equation*}
Equation~\ref{eq-quadij} then follows from the equation
\begin{equation*}
    \mathsf{opt} = |\mathcal{Q}_\textnormal{opt}| = \sum_{i=1}^r \sum_{j=1}^r |\mathcal{Q}_\textnormal{opt} \cap \mathcal{Q}_{i,j}'| \geq \sum_{(i,j) \in P'} |\mathcal{Q}_\textnormal{opt} \cap \mathcal{Q}_{i,j}'|.
\end{equation*}

\subsection{Proof of Lemma~\ref{lem-bound}}
It suffices to show $\sum_{k=1}^r (|S_{i,k}| + |\mathcal{Q}_{i,k}|) = O(f(n_0,\varepsilon)+r)$ for all $i \in \{1,\dots,r\}$ at any time in the first period.
Fix $i \in \{1,\dots,r\}$.
Initially, the $i$-th row of the grid contains $O(f(n_0,\varepsilon))$ points in $S$ and $O(f(n_0,\varepsilon))$ vertices of the quadrants in $\mathcal{Q}$.
Since the period only consists of $f(n_0,\varepsilon)$ operations, the $i$-th row always contains $O(f(n_0,\varepsilon))$ points in $S$ and $O(f(n_0,\varepsilon))$ vertices of the quadrants in $\mathcal{Q}$ during the entire period.
Thus, at any time in the period, $\sum_{k=1}^r |S_{i,k}| = O(f(n_0,\varepsilon))$ and the total number of the non-special quadrants in $\mathcal{Q}_{i,1},\dots,\mathcal{Q}_{i,r}$ is bounded by $O(f(n_0,\varepsilon))$.
The number of the special quadrants in $\mathcal{Q}_{i,1},\dots,\mathcal{Q}_{i,r}$ is clearly $O(r)$, since each $\mathcal{Q}_{i,k}$ contains at most four special quadrants.
Thus, the equation $\sum_{k=1}^r (|S_{i,k}| + |\mathcal{Q}_{i,k}|) = O(f(n_0,\varepsilon)+r)$ holds at any time in the period.

\subsection{Proof of Lemma~\ref{lem-qans}}
We first prove an invariant of our algorithm: whenever $a$ is updated in the algorithm, the following two properties hold.
\begin{enumerate}
    \item $a$ is set to be the leftmost point in $S \cap U^\text{SE}$ that is not covered by $\mathcal{Q}_\text{ans}$.
    \item Any point in $S \cap U^\text{SE}$ above $a$ is not covered by $\mathcal{Q}_\text{ans}$.
\end{enumerate}
The first update of $a$ happens in Step~2.
At this time, $\mathcal{Q}_\text{ans} = \{\varPhi_\rightarrow(\sigma,\mathcal{Q}^\text{SW}),\varPhi_\uparrow(\sigma,\mathcal{Q}^\text{SE})\}$, $\tilde{y} = y(\varPhi_\uparrow(\sigma,\mathcal{Q}^\text{SE}))$, and $a$ is set to be $\phi(\tilde{y})$, i.e., the leftmost point in $S \cap U^\text{SE}$ whose $y$-coordinate is larger than $\tilde{y}$.
In order to prove the invariant, we only need to show that \textbf{(1)} any point in $S \cap U^\text{SE}$ strictly to the left of $a$ is covered by $\mathcal{Q}_\text{ans}$ and \textbf{(2)} any point in $S \cap U^\text{SE}$ above $a$ (including $a$ itself) is not covered by $\mathcal{Q}_\text{ans}$.
Let $b \in S \cap U^\text{SE}$ be a point strictly to the left of $a$.
First, the $y$-coordinate of $b$ is at most $\tilde{y}$, since $a$ is the leftmost point in $S \cap U^\text{SE}$ whose $y$-coordinate is larger than $\tilde{y}$.
Therefore, if $b$ is to the right of $\sigma$, then $b \in \varPhi_\uparrow(\sigma,\mathcal{Q}^\text{SE})$.
Now assume $b$ is strictly to the left of $\sigma$.
Because $\sigma$ is on the boundary $\gamma$ of $U^\text{SE}$, $b$ must be below $\sigma$.
Thus, $b \in \varPhi_\rightarrow(\sigma,\mathcal{Q}^\text{SW})$.
We see that $\mathcal{Q}_\text{ans}$ covers $b$.
Let $a' \in S \cap U^\text{SE}$ be a point above $a$.
Then the $y$-coordinate of $a'$ is greater than $\tilde{y}$, which implies $a' \notin \varPhi_\uparrow(\sigma,\mathcal{Q}^\text{SE})$.
Furthermore, $a'$ must be strictly to the right of $\sigma$.
Indeed, if $a'$ is to the left of $\sigma$, then $a'$ is below $\sigma$ (as argued above) and hence the $y$-coordinate of $a'$ is at most $\tilde{y}$, resulting in a contradiction.
Now we see $a'$ is strictly above $\sigma$ and strictly to the right of $\sigma$.
Note that $\sigma$ is on the boundary of $U^\text{SW}$, which implies $a' \notin U^\text{SW}$ and in particular $a' \notin \varPhi_\rightarrow(\sigma,\mathcal{Q}^\text{SW})$.
So $a'$ is not covered by $\mathcal{Q}_\text{ans}$.
The invariant holds when we update $a$ in Step~2.

In the rest of the procedure, $a$ is only updated in Step~5.
Note that Step~3--5 form a loop in which we include a constant number of quadrants to cover $a$, see Figure~\ref{fig:quad_os_alg}.
It suffices to show that if the two properties hold at the beginning of an iteration of the loop (before Step~3), then they also hold after we update $a$ in Step~5 of that iteration.
Suppose we are at the beginning of an iteration, and the two properties hold.
If $a \in U^\text{NE}$, then we go directly from Step~3 to Step~6 and the algorithm terminates.
Otherwise, we go to Step~4.
Here we need to distinguish two cases: $a \in U^\text{NW}$ and $a \notin U^\text{NW}$.
\smallskip

\begin{figure}
\centering
\includegraphics[width=0.45\textwidth,page=1]{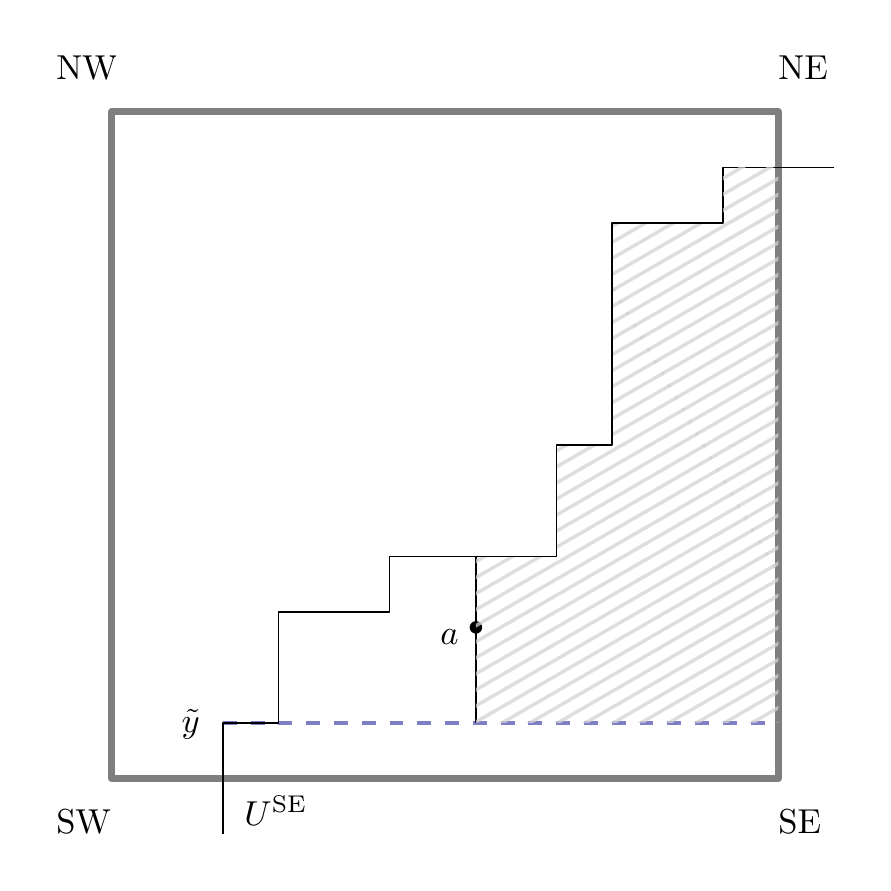}
\includegraphics[width=0.45\textwidth,page=2]{figures/quad_output_sensitive.pdf}\\
\includegraphics[width=0.45\textwidth,page=4]{figures/quad_output_sensitive.pdf}
\includegraphics[width=0.45\textwidth,page=3]{figures/quad_output_sensitive.pdf}
\caption{
\label{fig:quad_os_alg}
Different options for covering $a$, depending on what other quadrant types cover $a$.
The shaded region contains all remaining uncovered points (top-left figure).
If $a \in U^\text{NE}$, then the entire shaded region can be covered by two quadrants (top-right figure).
Otherwise (bottom figures), some progress can be made using up to three quadrants from $\mathcal{Q}^\text{SE}$ and $\mathcal{Q}^\text{NW}$.
After choosing these quadrants, the remaining uncovered points of $S \cap U^\text{SE}$ are not covered by any quadrant of $\mathcal{Q}^\text{SE}$, $\mathcal{Q}^\text{NE}$, and $\mathcal{Q}^\text{NW}$ that also covers $a$.
Note that $a \not\in U^\text{SW}$.
}
\end{figure}

\noindent
\textbf{[Case 1]}
We first consider the case $a \notin U^\text{NW}$.
In this case, we only include in $\mathcal{Q}_\text{ans}$ a new quadrant $Q = \varPhi_\uparrow(a,\mathcal{Q}^\text{SE})$.
Note that after including $Q$, $\mathcal{Q}_\text{ans}$ covers all the points in $S \cap U^\text{SE}$ whose $y$-coordinates are at most $y(Q)$.
To see this, consider a point $b \in S \cap U^\text{SE}$ whose $y$-coordinate is at most $y(Q)$.
If $b$ is strictly to the left of $a$, then $b$ is covered by $\mathcal{Q}_\text{ans}$ even before we include $Q$, since $a$ is the leftmost point in $S \cap U^\text{SE}$ that is not covered by $\mathcal{Q}_\text{ans}$ at the beginning of this iteration (by our assumption).
Otherwise, $b$ is to the right of $a$ and is hence covered by $Q$.
On the other hand, we also notice that after including $Q$, $\mathcal{Q}_\text{ans}$ does not cover any the points in $S \cap U^\text{SE}$ whose $y$-coordinates are greater than $y(Q)$.
To see this, consider a point $b \in S \cap U^\text{SE}$ whose $y$-coordinate is greater than $y(Q)$.
Then $b \notin Q$ and $b$ is strictly above $a$ (since $a \in Q$).
By our assumption, at the beginning of this iteration (when $Q$ is not included in $\mathcal{Q}_\text{ans}$), $\mathcal{Q}_\text{ans}$ does not cover any point above $a$, and in particular does not cover $b$.
Since $b \notin Q$, we see that, after we include $Q$ in $\mathcal{Q}_\text{ans}$, the new $\mathcal{Q}_\text{ans}$ still does not cover $b$.
To summarize, the new $\mathcal{Q}_\text{ans}$ contains all points in $S \cap U^\text{SE}$ whose $y$-coordinates are at most $y(Q)$ and no points in $S \cap U^\text{SE}$ whose $y$-coordinates are greater than $y(Q)$.
Therefore, when we set $a$ to be $\phi(y(Q))$ in Step~5, $a$ is the leftmost point in $S \cap U^\text{SE}$ that is not covered by $\mathcal{Q}_\text{ans}$ and any point in $S \cap U^\text{SE}$ above $a$ is not covered by $\mathcal{Q}_\text{ans}$.
\smallskip

\noindent
\textbf{[Case 2]}
Next, we consider the case $a \in U^\text{NW}$.
In this case, we include in $\mathcal{Q}_\text{ans}$ three new quadrants: $\varPhi_\rightarrow(a,\mathcal{Q}^\text{NW})$, $\varPhi_\uparrow(a,\mathcal{Q}^\text{SE})$, and $Q = \varPhi_\uparrow(v,\mathcal{Q}^\text{SE})$ where $v$ is the vertex of $\varPhi_\rightarrow(a,\mathcal{Q}^\text{NW})$.
We first show that after including these three new quadrants, $\mathcal{Q}_\text{ans}$ covers all the points in $S \cap U^\text{SE}$ whose $y$-coordinates are at most $y(Q)$.
Consider a point $b \in S \cap U^\text{SE}$ whose $y$-coordinate is at most $y(Q)$.
If the $x$-coordinate of $b$ is at least $x(Q)$, then $b \in Q$.
Otherwise, $b$ is strictly to the left of $v$ (because $v \in Q$).
Thus, if $b$ is above $v$, then $b \in \varPhi_\rightarrow(a,\mathcal{Q}^\text{NW})$.
It now suffices to consider the case when $b$ is strictly below $v$.
In this case, $b$ is strictly below $a$, since $v$ is the vertex of $\varPhi_\rightarrow(a,\mathcal{Q}^\text{NW})$.
If $b$ is to the right of $a$, then $b \in \varPhi_\uparrow(a,\mathcal{Q}^\text{SE})$.
If $b$ is strictly to the left of $a$, then $b$ is covered by $\mathcal{Q}_\text{ans}$ even before we include the three new quadrants, since $a$ is the leftmost point in $S \cap U^\text{SE}$ that is not covered by $\mathcal{Q}_\text{ans}$ at the beginning of this iteration (by our assumption).
In any case, $b$ is covered by $\mathcal{Q}_\text{ans}$.
On the other hand, we notice that after including these three new quadrants, $\mathcal{Q}_\text{ans}$ does not cover any points in $S \cap U^\text{SE}$ whose $y$-coordinates are greater than $y(Q)$.
To see this, consider a point $b \in S \cap U^\text{SE}$ whose $y$-coordinate is greater than $y(Q)$.
We first establish some obvious facts.
Since the NW quadrant $\varPhi_\rightarrow(a,\mathcal{Q}^\text{NW})$ and the SE quadrant $\varPhi_\uparrow(a,\mathcal{Q}^\text{SE})$ both contain $a$, we have $v \in \varPhi_\uparrow(a,\mathcal{Q}^\text{SE})$.
It follows that $y(\varPhi_\uparrow(a,\mathcal{Q}^\text{SE})) \leq y(Q)$, by the definition of $Q$.
Therefore, $b$ is strictly above $a$, since the $y$-coordinate of $a$ is at most $y(\varPhi_\uparrow(a,\mathcal{Q}^\text{SE}))$.
By our assumption, at the beginning of this iteration (when the three new quadrants are not included in $\mathcal{Q}_\text{ans}$), $\mathcal{Q}_\text{ans}$ does not cover any point above $a$, and in particular does not cover $b$.
Also, the fact $y(\varPhi_\uparrow(a,\mathcal{Q}^\text{SE})) \leq y(Q)$ implies $b \notin \varPhi_\uparrow(a,\mathcal{Q}^\text{SE})$.
To see $b \notin \varPhi_\rightarrow(a,\mathcal{Q}^\text{NW})$, let $Q' \in \mathcal{Q}^\text{SE}$ be a quadrant contains $b$ (such a quadrant exists as $b \in U^\text{SE}$).
Then $y(Q') > y(Q)$.
By the definition of $Q$, we have $v \notin Q'$ and thus $Q' \cap \varPhi_\rightarrow(a,\mathcal{Q}^\text{NW}) = \emptyset$, which implies $b \notin \varPhi_\rightarrow(a,\mathcal{Q}^\text{NW})$.
Finally, it is clear that $b \notin Q$.
So after we include the three new quadrants in $\mathcal{Q}_\text{ans}$, the new $\mathcal{Q}_\text{ans}$ still does not cover $b$.
To summarize, the new $\mathcal{Q}_\text{ans}$ contains all points in $S \cap U^\text{SE}$ whose $y$-coordinates are at most $y(Q)$ and no points in $S \cap U^\text{SE}$ whose $y$-coordinates are greater than $y(Q)$.
Therefore, when we set $a$ to be $\phi(y(Q))$ in Step~5, $a$ is the leftmost point in $S \cap U^\text{SE}$ that is not covered by $\mathcal{Q}_\text{ans}$ and any point in $S \cap U^\text{SE}$ above $a$ is not covered by $\mathcal{Q}_\text{ans}$.
\smallskip

Combining the discussions for the two cases, the invariant is proved.
With the invariant in hand, we now prove the lemma.
For convenience, we denote by $a_i$ the point $a$ in the $i$-th iteration (before the update of Step~5) of the loop Step~3--5.
The invariant of our algorithm guarantees that $a_i$ is the leftmost point in $S \cap U^\text{SE}$ that is not covered by $\mathcal{Q}_\text{ans}$ at the beginning of the $i$-th iteration.
Furthermore, in the $i$-th iteration, we always add the quadrant $\varPhi_\uparrow(a_i,\mathcal{Q}^\text{SE})$ to $\mathcal{Q}_\text{ans}$, hence $a_i$ is covered by $\mathcal{Q}_\text{ans}$ at the end of the $i$-th iteration.
This implies that our algorithm always terminates, because $\mathcal{Q}_\text{ans}$ covers at least one more point in $S \cap U^\text{SE}$ in each iteration.

We first show (the eventual) $\mathcal{Q}_\text{ans}$ covers all points in $S \cap U^\text{SE}$.
If we go to Step~6 from Step~1, then $S \cap U^\text{SE} = \emptyset$ and nothing needs to be proved.
If we go to Step~6 from Step~2, we know that the $y$-coordinates of all points in $S \cap U^\text{SE}$ are at most $y(\varPhi_\uparrow(\sigma,\mathcal{Q}^\text{SE}))$.
Consider a point $a \in S \cap U^\text{SE}$.
If $a$ is to the right of $\sigma$, then $a \in \varPhi_\uparrow(\sigma,\mathcal{Q}^\text{SE})$, as the $y$-coordinate of $a$ is at most $y(\varPhi_\uparrow(\sigma,\mathcal{Q}^\text{SE}))$.
If $a$ is to the left of $\sigma$, then $a \in \varPhi_\rightarrow(\sigma,\mathcal{Q}^\text{SW})$, because any point in $U^\text{SE}$ to the left of $\sigma$ is covered by $\varPhi_\rightarrow(\sigma,\mathcal{Q}^\text{SW})$.
Thus, $a$ is covered by $\mathcal{Q}_\text{ans}$, implying that all points in $S \cap U^\text{SE}$ are covered by $\mathcal{Q}_\text{ans}$.
The last case is that we go to Step~6 from Step~3.
By the invariant, we know that $a$ is the leftmost point in $S \cap U^\text{SE}$ that is not covered by $\mathcal{Q}_\text{ans}$ just before Step~3.
When we include the two quadrants $\varPhi_\uparrow(a,\mathcal{Q}^\text{NE})$ and $\varPhi_\uparrow(a,\mathcal{Q}^\text{SE})$ in $\mathcal{Q}_\text{ans}$, any point to the right of $a$ is covered by $\mathcal{Q}_\text{ans}$.
Thus, all points in $S \cap U^\text{SE}$ are covered by $\mathcal{Q}_\text{ans}$ when we go from Step~3 to Step~6.
The remaining case is that we go to Step~6 from Step~5.
This case happens only when $\phi(\tilde{y})$ does not exist, or equivalently, the $y$-coordinate of every point in $S \cap U^\text{SE}$ is at most $y(Q)$.
Recall that when proving the invariant, we showed that after we include $Q$ in $\mathcal{Q}_\text{ans}$ in Step~5, $\mathcal{Q}_\text{ans}$ covers all the points in $S \cap U^\text{SE}$ whose $y$-coordinates are at most $y(Q)$.
Hence, all points in $S \cap U^\text{SE}$ are covered by $\mathcal{Q}_\text{ans}$ when we go from Step~5 to Step~6.

We then show that $|\mathcal{Q}_\text{ans}| = O(\mathsf{opt}^\text{SE})$.
To this end, we notice that in each iteration of the loop Step~3--5, we include in $\mathcal{Q}_\text{ans}$ a constant number of quadrants.
Suppose we do $k$ iterations in total during the algorithm.
It suffices to show that $k = O(\mathsf{opt}^\text{SE})$.
We claim that any quadrant $Q \in \mathcal{Q}$ contains at most one of $a_1,\dots,a_k$.
First, we observe that any quadrant in $\mathcal{Q}^\text{NE}$ may only contain $a_k$.
Indeed, if $a_i \in U^\text{SE}$ for some $i<k$, then in the $i$-th iteration, the algorithm goes from Step~3 to Step~6 and terminates.
Next, we show that any quadrant in $\mathcal{Q}^\text{SW}$ does not contain $a_i$ (or equivalently, $a_i \notin U^\text{SW}$) for all $i \in \{1,\dots,k\}$.
Let $i \in \{1,\dots,k\}$.
By the invariant we proved before, $a_i$ is the leftmost point in $S \cap U^\text{SE}$ that is not covered by $\mathcal{Q}_\text{ans}$ at the beginning of the $i$-th iteration.
Since we include $\varPhi_\rightarrow(\sigma,\mathcal{Q}^\text{SW})$ and $\varPhi_\uparrow(\sigma,\mathcal{Q}^\text{SE})$ in $\mathcal{Q}_\text{ans}$ Step~2, we know that $a_i \notin \varPhi_\rightarrow(\sigma,\mathcal{Q}^\text{SW})$ and $a_i \notin \varPhi_\uparrow(\sigma,\mathcal{Q}^\text{SE})$.
Thus, $a_i$ must be strictly above $\sigma$, since any point below $\sigma$ is covered by either $\varPhi_\rightarrow(\sigma,\mathcal{Q}^\text{SW})$ or $\varPhi_\uparrow(\sigma,\mathcal{Q}^\text{SE})$.
It follows that $a_i$ is to the right of $\sigma$, because $a_i \in U^\text{SE}$ and $\sigma$ is on the boundary $\gamma$ of $U^\text{SE}$.
In fact, $a_i$ is strictly to the right of $\sigma$, since any point in $S \cap U^\text{SE}$ that has the same $x$-coordinate as $\sigma$ is covered by $\varPhi_\uparrow(\sigma,\mathcal{Q}^\text{SE})$.
Because $\sigma$ is also on the boundary of $U^\text{SW}$, given the fact that $a_i$ is strictly above $\sigma$ and strictly to the right of $\sigma$, we see that $a_i \notin U^\text{SW}$.

Now it suffices to verify that any quadrant in $\mathcal{Q}^\text{SE}$ or $\mathcal{Q}^\text{NW}$ contains at most one of $a_1,\dots,a_k$.
Let $i,j \in \{1,\dots,k\}$ such that $i<j$.
We want to show that no quadrant in $\mathcal{Q}^\text{SE}$ or $\mathcal{Q}^\text{NW}$ contains both $a_i$ and $a_j$.
Since $a_j$ is not covered by $\mathcal{Q}_\text{ans}$ at the beginning of the $j$-th iteration, it is also not covered by $\mathcal{Q}_\text{ans}$ at any point before the $j$-th iteration and in particular, at the beginning of the $i$-th iteration.
This implies that $a_j$ is to the right of $a_i$.
In the $i$-th iteration, we always add the quadrant $\varPhi_\uparrow(a_i,\mathcal{Q}^\text{SE})$ to $\mathcal{Q}_\text{ans}$.
We have $a_j \notin \varPhi_\uparrow(a_i,\mathcal{Q}^\text{SE})$, since $a_j$ is not covered by $\mathcal{Q}_\text{ans}$ at the end of the $i$-th iteration.
Thus, $a_j$ is strictly above $a_i$.
Now let $Q \in \mathcal{Q}^\text{SE}$ be a quadrant that contains $a_i$.
Note that a point to the right of $a_i$ is contained in $Q$ only if it is contained in $\varPhi_\uparrow(a_i,\mathcal{Q}^\text{SE})$, simply by the definition of $\varPhi_\uparrow$.
Because $a_j \notin \varPhi_\uparrow(a_i,\mathcal{Q}^\text{SE})$, we have $a_j \notin Q$.
Consequently, no quadrant in $\mathcal{Q}^\text{SE}$ contains both $a_i$ and $a_j$.
Next, we show that no quadrant in $\mathcal{Q}^\text{NW}$ contains both $a_i$ and $a_j$.
If $a_i \notin U^\text{NW}$, we are done.
So assume $a_i \in U^\text{NW}$.
In this case, we add the quadrant $\varPhi_\rightarrow(a_i,\mathcal{Q}^\text{NW})$ to $\mathcal{Q}_\text{ans}$ in Step~4 of the $i$-th iteration.
Again, we have $a_j \notin \varPhi_\rightarrow(a_i,\mathcal{Q}^\text{NW})$ for $a_j$ is not covered by $\mathcal{Q}_\text{ans}$ at the end of the $i$-th iteration.
Let $Q \in \mathcal{Q}^\text{NW}$ be a quadrant that contains $a_i$.
Note that a point above $a_i$ is contained in $Q$ only if it is contained in $\varPhi_\rightarrow(a_i,\mathcal{Q}^\text{NW})$, simply by the definition of $\varPhi_\rightarrow$.
Because $a_j \notin \varPhi_\rightarrow(a_i,\mathcal{Q}^\text{NW})$, we have $a_j \notin Q$.
Consequently, no quadrant in $\mathcal{Q}^\text{NW}$ contains both $a_i$ and $a_j$.

To summarize, we showed that any quadrant $Q \in \mathcal{Q}$ contains at most one of $a_1,\dots,a_k$.
Therefore, covering all of $a_1,\dots,a_k$ requires at least $k$ quadrants in $\mathcal{Q}$.
This implies $k \leq \mathsf{opt}^\text{SE}$, because $a_1,\dots,a_k \in S \cap U^\text{SE}$.
As a result, $|\mathcal{Q}_\text{ans}| = O(k) = O(\mathsf{opt}^\text{SE})$.

\subsection{Proof of Lemma~\ref{lem-stab}}
Let $\varPi = (S,\mathcal{R})$ be a dynamic set cover instance with only point updates, and $i \geq 0$ be a number.
We claim that $|\mathsf{opt}_i(\varPi) - \mathsf{opt}_{i+1}(\varPi)| \leq 1$.
Denote by $S_i$ and $S_{i+1}$ be the point set $S$ at the time $i$ and $i+1$.
Also, let $S_i' \subseteq S_i$ and $S_{i+1}' \subseteq S_{i+1}$ consist of the points that are covered by $\mathcal{R}$.
Suppose the $(i+1)$-th operation inserts a point $a$ to $S$, so $S_{i+1} = S_i \cup \{a\}$.
If $a$ is not covered by $\mathcal{R}$, then $S_{i+1}' = S_i'$ and $\mathsf{opt}_{i+1}(\varPi) = \mathsf{opt}_i(\varPi)$.
If $a$ is covered by $\mathcal{R}$, then $S_{i+1}' = S_i' \cup \{a\}$.
Let $R \in \mathcal{R}$ be a range that covers $a$.
An optimal set cover for $(S_i',\mathcal{R})$ together with $R$ is a set cover for $(S_{i+1}',\mathcal{R})$.
It follows that $\mathsf{opt}_i(\varPi) \leq \mathsf{opt}_{i+1}(\varPi) \leq \mathsf{opt}_{i+1}(\varPi)+1$.
The case where the $(i+1)$-th operation deletes a point from $S$ is symmetric.
This shows that dynamic set cover in the partially dynamic setting is stable.
The same argument can also be applied to show that dynamic hitting set in the partially dynamic setting is stable.

Next, we consider the dynamic interval hitting set problem (in the fully dynamic setting).
Let $\varPi = (S,\mathcal{I})$ be a dynamic interval hitting set instance, and $i \geq 0$ be a number.
Denote by $S_i$ and $\mathcal{I}_i$ as the point set $S$ and interval collection $\mathcal{I}$ at the time $i$.
Let $\mathcal{I}_i' \subseteq \mathcal{I}_i$ consist of the intervals that are hit by $S_i$.
We want to show that $|\mathsf{opt}_i(\varPi) - \mathsf{opt}_{i+1}(\varPi)| \leq 1$.
We distinguish two cases.
\smallskip

\noindent
\textbf{[Case 1]} The $(i+1)$-th operation happens on $\mathcal{I}$.
Suppose the $(i+1)$-th operation is an insertion on $\mathcal{I}$, and let $I$ be the interval inserted.
If $I$ is not hit by $S_i$, then $(S_{i+1},\mathcal{I}_{i+1}') = (S_i,\mathcal{I}_i')$ and $\mathsf{opt}_{i+1}(\varPi) = \mathsf{opt}_i(\varPi)$.
If $I$ is hit by $S_i$, then $(S_{i+1},\mathcal{I}_{i+1}') = (S_i,\mathcal{I}_i' \cup \{I\})$.
In this case, we have $\mathsf{opt}_i(\varPi)' \leq \mathsf{opt}_{i+1}(\varPi) \leq \mathsf{opt}_i(\varPi)+1$, since an optimal hitting set for $(S_i,\mathcal{I}_i')$ together with a point hitting $I$ is a hitting set for $(S_{i+1},\mathcal{I}_{i+1}')$.
Thus, $|\mathsf{opt}_i(\varPi) - \mathsf{opt}_{i+1}(\varPi)| \leq 1$.
The case where the $(i+1)$-th operation is a deletion on $\mathcal{I}$ is symmetric.
\smallskip

\noindent
\textbf{[Case 2]} The $(i+1)$-th operation happens on $S$.
Suppose the $(i+1)$-th operation is an insertion on $S$, and let $a$ be the point inserted.
Then $(S_{i+1},\mathcal{I}_{i+1}') = (S_i \cup \{a\},\mathcal{I}_i' \cup \mathcal{J})$ where $\mathcal{J} = \{I \in \mathcal{I}: a \in I\}$.
It is clear that $\mathsf{opt}_{i+1}(\varPi) \leq \mathsf{opt}_i(\varPi) + 1$, because an optimal hitting set for $(S_i,\mathcal{I}_i')$ together with the point $a$ is a hitting set for $(S_{i+1},\mathcal{I}_{i+1}')$.
It suffices to show that $\mathsf{opt}_i(\varPi) \leq \mathsf{opt}_{i+1}(\varPi) + 1$.
Let $S^* \subseteq S_{i+1}$ be an optimal hitting set for $(S_{i+1},\mathcal{I}_{i+1}')$.
If $a \notin S^*$, then $S^*$ is also an optimal hitting set for $(S_i,\mathcal{I}_i')$.
Otherwise, let $a^- \in S_i$ (resp., $a^+ \in S_i$) be the rightmost (resp., leftmost) point to the left (resp., right) of $a$; in the case where $a$ is the leftmost (resp, rightmost) point in $S_i$, then let $a^-$ (resp., $a^+$) be an arbitrary point in $S_i$.
We claim that $(S^* \backslash \{a\}) \cup \{a^-,a^+\}$ is a hitting set for $(S_i,\mathcal{I}_i')$.
Consider an interval $I \in \mathcal{I}_i'$.
If $I \notin \mathcal{J}$, then $I$ is hit by some point in $S^* \backslash \{a\}$.
Otherwise, $I \in \mathcal{J}$, and thus $I$ is hit by $a$.
Since $I \in \mathcal{I}_i'$, $I$ is also hit by some point in $S_i$, say $b$.
If $b$ is to the left of $a$, then $I$ must be hit by $a^-$.
On the other hand, if $b$ is to the right of $a$, $I$ must be hit by $a^+$.
Thus, in any case, $I$ is hit by $(S^* \backslash \{a\}) \cup \{a^-,a^+\}$.
It follows that $|(S^* \backslash \{a\}) \cup \{a^-,a^+\}| = \mathsf{opt}_{i+1}(\varPi) +1 \geq \mathsf{opt}_i(\varPi)$.
The case where the $(i+1)$-th operation is a deletion on $S$ is symmetric.



\subsection{Proof of Lemma~\ref{lem-dsexhs}}
For convenience, we include in $S$ two dummy points $x^- = -\infty$ and $x^+ = +\infty$.
We assume these two dummy points are always in $S$.
Since $x^-$ and $x^+$ do not hit any interval, including them in $S$ is safe.

First, we store $\mathcal{I}$ in a binary search tree $\mathcal{T}_1$ where the key of an interval is its left endpoint.
We augment each node $\mathbf{u} \in \mathcal{T}_1$ with an additional field which stores the interval in the subtree rooted at $\mathbf{u}$ with the leftmost right endpoint.
We notice that using $\mathcal{T}_1$, we can decide in $\widetilde{O}(1)$ time, for two given numbers $a,a^+ \in \mathbb{R}$, whether there is an interval in $\mathcal{I}$ whose both endpoints are in the open interval $(a^-,a^+)$.
Specifically, we first look for the interval in $\mathcal{T}_1$ with the leftmost right endpoint whose key is greater than $a$.
With the augmented fields, this interval can be found in $O(\log n)$ time.
If the two endpoints of this interval are contained $(a,a^+)$, then we return ``yes''; otherwise, no interval in $\mathcal{I}$ has two endpoints in $(a,a^+)$ and we return ``no''.
Clearly, $\mathcal{T}_1$ can be constructed in $\widetilde{O}(\mathcal{I})$ time and dynamized with $\widetilde{O}(1)$ update time, and hence it is basic.

We then store $S$ in a standard range tree $\mathcal{T}_2$.
The points in $S$ are stored at the leaves of $\mathcal{T}_2$.
Each node $\mathbf{u} \in \mathcal{T}_2$ corresponds to a \textit{canonical subset} $S(\mathbf{u})$ of $S$ consisting of the points stored in the subtree rooted at $\mathbf{u}$.
At each internal node $\mathbf{u} \in \mathcal{T}_2$, we store the leftmost and rightmost points in $S(\mathbf{u})$ as well as a separation point $s_\mathbf{u} \in \mathbb{R}$ such that all points in the left (resp., right) subtree of $\mathbf{u}$ are to the left (resp., right) of $s_\mathbf{u}$.
At each leaf $\mathbf{u}$ of $\mathcal{T}$ that does not correspond to the point $x^+$, we maintain an additional field $\sigma(\mathbf{u})$ indicating whether there exists an interval in $\mathcal{I}$ whose both endpoints are in the open interval $(a_\mathbf{u},a_\mathbf{u}^+)$ where $a_\mathbf{u} \in S$ is the point corresponding to $\mathbf{u}$ and $a_\mathbf{u}^+ \in S$ is the leftmost point that is to the right of $a_\mathbf{u}$;
we set $\sigma(\mathbf{u}) = 1$ if such an interval exists and set $\sigma(\mathbf{u}) = 0$ otherwise.
Note that using the binary search tree $\mathcal{T}_1$, $\sigma(\mathbf{u})$ can be computed in $\widetilde{O}(1)$ time for each leaf $\mathbf{u}$ of $\mathcal{T}_2$.
Indeed, we only need to find $a_\mathbf{u}^+$, which can be done in $\widetilde{O}(1)$ time by a walk in $\mathcal{T}_2$, and query $\mathcal{T}_1$ to see whether there is an interval in $\mathcal{I}$ whose both endpoints are in the open interval $(a_\mathbf{u},a_\mathbf{u}^+)$.
Besides, we also maintain a counter $\sigma^*$ that is the total number of the leaves whose $\sigma$-values are equal to 1.
It is easy to see that $(S,\mathcal{I})$ has a hitting set iff $\sigma^* = 0$.

It is clear that $\mathcal{T}_2$ can be constructed in $\widetilde{O}(|S|)$ time.
We show that when $(S,\mathcal{I})$ changes, we can maintain $\mathcal{T}_2$ with the $\sigma$-fields in $\widetilde{O}(1)$ time, by using $\mathcal{T}_1$.
Suppose a point $a \in \mathbb{R}$ is inserted to $S$.
We need to first compute the $\sigma$-value of the leaf corresponding to $a$ using $\mathcal{T}_1$.
Let $b \in S$ be the rightmost point to the left of $a$.
Due to the insertion of $a$, the $\sigma$-value of the leaf $\mathbf{u}$ of $\mathcal{T}_2$ corresponding to $b$ may also change.
So we need to find $\mathbf{u}$, which can be done in $\widetilde{O}(1)$ time by a walk in $\mathcal{T}_2$, and update $\sigma(\mathbf{u})$ using $\mathcal{T}_1$.
The $\sigma$-values of all the other leaves remain unchanged.
A deletion of a point from $S$ is handled similarly.
Now suppose an interval $I$ is inserted to or deleted from $\mathcal{I}$.
We find in $\widetilde{O}(1)$ the leaf $\mathbf{u}$ of $\mathcal{T}_2$ corresponding to the rightmost point in $S$ that is to the left of the left endpoint of $I$.
Note that the insertion/deletion of $I$ does not change the $\sigma$-values of the leaves other than $\mathbf{u}$.
So it suffices to re-compute $\sigma(\mathbf{u})$ using $\mathcal{T}_1$.
The time cost for all the cases above is $\widetilde{O}(1)$.
The rotations of $\mathcal{T}_2$ (for self-balancing) do not change the $\sigma$-fields.
It follows that the counter $\sigma^*$ can also be maintained in $\widetilde{O}(1)$ time.

The dynamic data structure $\mathcal{B}$ in the lemma just consists of (the dynamic versions of) the binary search tree $\mathcal{T}_1$ and the range tree $\mathcal{T}_2$.
The update time of $\mathcal{B}$ is $\widetilde{O}(1)$ time and the construction time is $\widetilde{O}(n_0)$.
To indicate whether the current $(S,\mathcal{I})$ has a hitting set or not, we simply check whether $\sigma^* = 0$ or not.

\subsection{Proof of Lemma~\ref{lem-supleftright}}
We can simply store $S$ in a binary search tree.
Then the rightmost (resp., leftmost) point in $S$ to the left (resp., right) of a given point $a$ can be reported in $\widetilde{O}(1)$ time by searching in the tree.
Clearly, the binary search tree can be constructed in $\widetilde{O}(|S|)$ time and dynamized with $\widetilde{O}(1)$ update time, and thus it is basic.

\subsection{Proof of Lemma~\ref{lem-suphitint}}
We can simply store $S$ in a binary search tree.
Then a point $a \in S$ contained in a given interval $I$ can be reported in $\widetilde{O}(1)$ time by searching in the tree.
Clearly, the binary search tree can be constructed in $\widetilde{O}(|S|)$ time and dynamized with $\widetilde{O}(1)$ update time, and thus it is basic.

\subsection{Proof of Lemma~\ref{lem-inthscorrect}}
We use $\mathsf{opt}'$ to denote the quasi-optimum of the current $(S,\mathcal{I})$.
It suffices to show $|S^*| \leq (1+\varepsilon) \cdot \mathsf{opt}'$ at any time, since we always have $\mathsf{opt}' \leq \mathsf{opt}$.

Initially, $S^*$ is an optimal hitting set for $(S_0,\mathcal{I}_0')$, so we have $|S^*| = \mathsf{opt}'$ at that time.
If the current $S^*$ is obtained by re-computing using the output-sensitive algorithm, then $|S^*| = \mathsf{opt}' = \mathsf{opt}$, as we only do re-computation when the current $(S,\mathcal{I})$ has a hitting set.
Suppose $S^*$ is obtained by local modification.
Consider the last re-computation of $S^*$, and we use $S_1^*$ to denote the $S^*$ at that point and use $\mathsf{opt}_1'$ to denote the quasi-optimum of $(S,\mathcal{I})$ at that point.
Then we have $|S_1^*| = \mathsf{opt}_1'$.
As argued before, $\mathsf{cnt}$ is the number of the operations processed after the last re-computation of $S^*$.
By the stability of dynamic interval hitting set (Lemma~\ref{lem-stab}), we have $|\mathsf{opt}' - \mathsf{opt}_1'| \leq \mathsf{cnt}$, implying $\mathsf{opt}_1' \leq \mathsf{opt}' + \mathsf{cnt}$.
Furthermore, by the fact that the size of $S^*$ either increases by 1 or keeps unchanged after each local modification, we have $|S^*| - |S_1^*| \leq \mathsf{cnt}$.
It follows that
\begin{equation*}
    |S^*| \leq |S_1^*| + \mathsf{cnt} = \mathsf{opt}_1' + \mathsf{cnt} \leq \mathsf{opt}' + 2 \cdot \mathsf{cnt}.
\end{equation*}
Note that $\mathsf{cnt} \leq \varepsilon \cdot \mathsf{opt}_1'/(2+\varepsilon)$, for otherwise we should re-compute $S^*$.
This implies $\mathsf{cnt} \leq (\varepsilon/2) \cdot \mathsf{opt}'$ and hence $|S^*| \leq (1+\varepsilon) \cdot \mathsf{opt}'$.

\subsection{Proof of Lemma~\ref{lem-inthscorrect2}}
Let $(S_i,\mathcal{I}_i)$ denote the instance $(S,\mathcal{I})$ at the time $i$, and $S_i^*$ be the point set $S^*$ at the time $i$.
Define $\mathcal{I}_i' \subseteq \mathcal{I}_i$ as the sub-collection consisting of the intervals that are hit by $S_i$.
Then $(S_i,\mathcal{I}_i')$ always has a hitting set.
Furthermore, if $(S_i,\mathcal{I}_i)$ has a hitting set, then $(S_i,\mathcal{I}_i') = (S_i,\mathcal{I}_i)$.
So it suffices to show that $S_i^*$ is always a hitting set for $(S_i,\mathcal{I}_i')$.
We prove this by induction.
It is clear that $S_0^*$ is a hitting set for $(S_0,\mathcal{I}_0')$.
Assume $S_{i-1}^*$ is a hitting set for $(S_{i-1},\mathcal{I}_{i-1}')$ and we show that $S_i^*$ is a hitting set for $(S_i,\mathcal{I}_i')$.
If $S_i^*$ is obtained by re-computing, then $(S_i,\mathcal{I}_i') = (S_i,\mathcal{I}_i)$, since we only re-compute $S^*$ when the current $(S,\mathcal{I})$ has a hitting set.
In this case, $S_i^*$ is clearly a hitting set for both $(S_i,\mathcal{I}_i')$ and $(S_i,\mathcal{I}_i)$.
So suppose $S_i^*$ is obtained by local modification.

We consider different cases separately according to the $i$-th operation.
If the $i$-th operation inserts a point $a$ to $S$, then $S_i = S_{i-1} \cup \{a\}$ and $\mathcal{I}_i = \mathcal{I}_{i-1}$.
In this case, $S_i^* = S_{i-1}^* \cup \{a\}$ and $\mathcal{I}_i' = \mathcal{I}_{i-1}' \cup \mathcal{J}$, where $\mathcal{J} = \{I \in \mathcal{I}_i: a \in I\}$.
The intervals in $\mathcal{I}_{i-1}'$ are hit by $S_{i-1}^*$ and the intervals in $\mathcal{J}$ are hit by the point $a$.
Hence, $S_i^*$ is a hitting set for $(S_i,\mathcal{I}_i')$.
If the $i$-th operation deletes a point $a$ from $S$, then $S_i = S_{i-1} \backslash \{a\}$ and $\mathcal{I}_{i-1} = \mathcal{I}_i$.
In this case, $S_i^*$ is obtained from $S_{i-1}^* \backslash \{a\}$ by adding the rightmost point $a^-$ in $S_i$ to the left of $a$ and the leftmost point $a^+$ in $S_i$ to the right of $a$ (if they exist).
Consider an interval $I \in \mathcal{I}_i'$.
We want to show that $I$ is hit by $S_i^*$.
Note that $\mathcal{I}_i' \subseteq \mathcal{I}_{i-1}'$.
Thus, $I \in \mathcal{I}_{i-1}'$ and $I$ is hit by $S_{i-1}^*$.
If $a \notin I$, then $I$ is hit by $S_{i-1}^* \backslash \{a\}$ and hence hit by $S_i^*$.
Otherwise, $a \in I$.
Since $I \in \mathcal{I}_i'$ and $a \notin S_i$, $I$ must be hit by some point $b \in S_i$ different from $a$.
If $b$ is to the left of $a$, then $I$ must be hit by $a^-$ (as $a^-$ is in between $b$ and $a$).
On the other hand, if $b$ is to the right of $a$, $I$ must be hit by $a^+$ (as $a^+$ is in between $a$ and $b$).
Therefore, $I$ is hit by $S_i^*$.
If the $i$-th operation inserts an interval $I$ to $\mathcal{I}$, then $S_i = S_{i-1}$ and $\mathcal{I}_i = \mathcal{I}_{i-1} \cup \{I\}$.
In this case, $S_i^*$ is obtained from $S_{i-1}^*$ by adding an arbitrary point $a \in S_i$ that hits $I$ (if $I$ is hit by $S_i$).
If $I$ is not hit by $S_i$, then $(S_i,\mathcal{I}_i') = (S_{i-1},\mathcal{I}_{i-1}')$ and $S_i^* = S_{i-1}^*$, thus $S_i^*$ is a hitting set for $(S_i,\mathcal{I}_i')$.
If $I$ is hit by $S_i$, then $\mathcal{I}_i' = \mathcal{I}_{i-1}' \cup \{I\}$ and $S_i^* = S_{i-1}^* \cup \{a\}$.
The intervals in $\mathcal{I}_{i-1}'$ are hit by $S_{i-1}^*$ by our induction hypothesis and $I$ is hit by the point $a$.
Hence, $S_i^*$ is a hitting set for $(S_i,\mathcal{I}_i')$.
If the $i$-th operation deletes an interval $I$ from $\mathcal{I}$, then $S_i = S_{i-1}$ and $\mathcal{I}_i = \mathcal{I}_{i-1} \backslash \{I\}$.
In this case, $S_i^* = S_{i-1}^*$.
Note that $\mathcal{I}_i' \subseteq \mathcal{I}_{i-1}'$, which implies that $S_i^*$ is a hitting set for $(S_i,\mathcal{I}_i')$.

\subsection{Proof of Lemma~\ref{lem-dsexsc}}
Note that $\mathcal{Q}$ is static.
We simply store $\mathcal{Q}$ in a static data structure $\mathcal{B}_0$ that can decide in $\widetilde{O}(1)$ time, for a given point $a \in \mathbb{R}^2$, whether there exists a quadrant in $\mathcal{Q}$ that covers $a$; this can be done using a standard orthogonal stabbing data structure with $\widetilde{O}(|\mathcal{Q}|)$ construction time.
Then our data structure $\mathcal{B}$ simply maintains the number $\tilde{n}$ of the points in $S$ that are not covered by $\mathcal{Q}$.
Initially, $\tilde{n}$ can be computed in $\widetilde{O}(n_0)$ time by considering every point in $S$ and use $\mathcal{B}_0$ to check if it is covered by $\mathcal{Q}$ in $\widetilde{O}(1)$ time.
After an operation on $S$, we can update $\tilde{n}$ in $\widetilde{O}(1)$ time by checking whether the inserted/deleted point is covered by $\mathcal{Q}$ using $\mathcal{B}_0$.

\subsection{Proof of Lemma~\ref{lem-quadsccorrect}}
Let $S_i$ denote the set $S$ at the time $i$ and $\mathcal{Q}_i^*$ be the $\mathcal{Q}^*$ at the time $i$.
Define $S_i' \subseteq S_i$ as the sub-collection consisting of the points that are covered by $\mathcal{Q}$.
Then $(S_i',\mathcal{Q})$ always has a hitting set.
Furthermore, if $(S_i,\mathcal{Q})$ has a set cover, then $(S_i',\mathcal{Q}) = (S_i,\mathcal{Q})$.
So it suffices to show that $\mathcal{Q}_i^*$ is always a set cover for $(S_i',\mathcal{Q})$.
We prove this by induction.
It is clear that $\mathcal{Q}_0^*$ is a set cover for $(S_0',\mathcal{Q})$.
Assume $\mathcal{Q}_{i-1}^*$ is a set cover for $(S_{i-1}',\mathcal{Q})$ and we show that $\mathcal{Q}_i^*$ is a set cover for $(S_i',\mathcal{Q})$.
If $\mathcal{Q}_i^*$ is obtained by re-computing, then $(S_i',\mathcal{Q}) = (S_i,\mathcal{Q})$, since we only re-compute $\mathcal{Q}^*$ when the current $(S,\mathcal{Q})$ has a set cover.
In this case, $\mathcal{Q}_i^*$ is clearly a set cover for both $(S_i',\mathcal{Q})$ and $(S_i,\mathcal{Q})$.
So suppose $\mathcal{Q}_i^*$ is obtained by local modification.
If the $i$-th iteration inserts a point $a$ to $S$, then $S_i = S_{i-1} \cup \{a\}$.
In this case, $\mathcal{Q}_i^*$ is obtained by including in $\mathcal{Q}_{i-1}^*$ an arbitrary quadrant $Q \in \mathcal{Q}$ that contains $a$ (if such a quadrant exists).
If $a$ is not covered by $\mathcal{R}$, then $S_i' = S_{i-1}'$ and $\mathcal{Q}_i^* = \mathcal{Q}_{i-1}^*$, thus $\mathcal{Q}_i^*$ is a set cover for $(S_i',\mathcal{Q})$.
If $a$ is covered by $\mathcal{R}$, then $S_i' = S_{i-1}' \cup \{a\}$ and $\mathcal{Q}_i^* = \mathcal{Q}_{i-1}^* \cup \{Q\}$.
The points in $S_{i-1}'$ are covered by $\mathcal{Q}_{i-1}^*$ by our induction hypothesis and the point $a$ is covered by $Q$.
Hence, $\mathcal{Q}_i^*$ is a set cover for $(S_i',\mathcal{Q})$.
If the $i$-th iteration delete a point $a$ from $S$, then $S_i = S_{i-1} \backslash \{a\}$.
In this case, $\mathcal{Q}_i^* = \mathcal{Q}_{i-1}^*$.
Note that $S_i' \subseteq S_{i-1}'$, which implies that $\mathcal{Q}_i^*$ is a set cover for $(S_i',\mathcal{Q})$.


\section{Implementation details and detailed time analysis for the dynamic interval set cover data structure} \label{appx-anaint}
We present the implementation details of our dynamic interval set cover data structure $\mathcal{D}_\text{new}$ in Section~\ref{sec-DISC} as well as a detailed time analysis.
Assume the function $f$ we choose satisfies two properties: \textbf{(1)} $f(m,\eps) \leq m/2$ for any $m$, and \textbf{(2)} $f(\Theta(m),\eps) = \Theta(f(m,\eps))$.

First, we discuss how to construct $\mathcal{D}_\text{new}$.
Constructing the data structure $\mathcal{A}$ takes $\widetilde{O}(n_0)$ time, as it is basic.
The portions $J_1,\dots,J_r$ can be computed in $\widetilde{O}(n_0)$ time by sorting the points in $S$ and the endpoints of the intervals in $\mathcal{I}$.
Once $J_1,\dots,J_r$ are computed, we build a (static) point location data structure $\mathcal{B}_1$ which can report in $O(\log r)$ time, for a given point $a \in \mathbb{R}$, the portion $J_i$ that contains $a$.
Clearly, $\mathcal{B}_1$ can be constructed in $\widetilde{O}(r)$ time.
With $\mathcal{B}_1$ in hand, we can determine in $\widetilde{O}(1)$ time, for each point $a \in S$ (resp., each interval $I \in \mathcal{I}$), the portion $J_i$ that contains $a$ (resp., the two $J_i$'s that contain the endpoints of $I$).
By doing this for all points in $S$ and all intervals in $\mathcal{I}$, we obtain all $S_i$ and all $\mathcal{I}_i$ in $\widetilde{O}(n_0 + r)$ time.
After this, we can build the data structures $\mathcal{D}_\text{old}^{(i)}$'s.
Constructing each $\mathcal{D}_\text{old}^{(i)}$ takes $\widetilde{O}(f(n_0,\eps))$ time since $|S_i| + |\mathcal{I}_i| = O(f(n_0,\eps))$.
Hence, the time for constructing all $\mathcal{D}_\text{old}^{(1)},\dots,\mathcal{D}_\text{old}^{(r)}$ is $\widetilde{O}(n_0)$.

The support data structure $\mathcal{B}_1$ will be used later in the implementation of the update procedure of $\mathcal{D}_\text{new}$ (we do not need to update $\mathcal{B}_1$ since it is static).
Besides, we need another support data structure $\mathcal{B}_2$ defined as follows.
\begin{lemma} \label{lem-supintcover}
One can store $\mathcal{I}$ in a basic data structure $\mathcal{B}_2$ such that given an interval $J$, an interval in $\mathcal{I}$ that contains $J$ can be reported in $\widetilde{O}(1)$ time (if it exists).
\end{lemma}
\begin{proof}
We store $\mathcal{I}$ in a binary search tree $\mathcal{T}$ where the key of each interval is its left endpoint.
We augment each node $\mathbf{u} \in \mathcal{T}$ with an additional field which stores the interval in the subtree rooted at $\mathbf{u}$ that has the rightmost right endpoint.
Given a query interval $J$, we first look for the interval in $\mathcal{T}$ with the rightmost right endpoint whose key is smaller than or equal to the left endpoint of $J$.
With the augmented fields, this interval can be found in $\widetilde{O}(1)$ time.
If this interval contains $J$, then we report it; otherwise, no interval in $\mathcal{I}$ contains $J$.
Clearly, $\mathcal{T}$ can be built in $\widetilde{O}(n_0)$ time and dynamized with $\widetilde{O}(1)$ update time, and thus $\mathcal{T}$ is basic.
\end{proof}
\noindent
Since $\mathcal{B}_2$ is basic, it can be constructed in $\widetilde{O}(n_0)$ time and updated in $\widetilde{O}(1)$ time.
We conclude that the construction time of $\mathcal{D}_\text{new}$ is $\widetilde{O}(n_0)$.

Next, we consider how to implement the update procedure of $\mathcal{D}_\text{new}$.
After each operation, we need to update the data structure $\mathcal{A}$ and the support data structure $\mathcal{B}_2$, which can be done in $\widetilde{O}(1)$ time since they are basic.
Also, if some $(S_i,\mathcal{I}_i)$ changes, we need to update the data structure $\mathcal{D}_\text{old}^{(i)}$.
An operation on $S$ changes exactly one $S_i$ and an operation on $\mathcal{I}$ changes at most two $\mathcal{I}_i$'s.
Thus, we only need to update at most two $\mathcal{D}_\text{old}^{(i)}$'s, and by using $\mathcal{B}_1$ we can find these $\mathcal{D}_\text{old}^{(i)}$'s in $\widetilde{O}(1)$ time.
Note that the size of each $(S_i,\mathcal{I}_i)$ is bounded by $O(f(n_0,\eps))$ during the first period (i.e., before the first reconstruction), because the period only consists of $f(n_0,\eps)$ operations.
Thus, updating the $\mathcal{D}_\text{old}$ data structures takes $O(f^\alpha(n_0,\eps)/\eps^{1-\alpha})$ amortized time.

Then we discuss the maintenance of the solution.
The time for simulating the output-sensitive algorithm is $\widetilde{O}(\delta)$, i.e., $\widetilde{O}(\min\{r/\varepsilon,n\})$.
If the algorithm gives the solution $\mathcal{I}_\text{appx}$, we compute $|\mathcal{I}_\text{appx}|$ and store $\mathcal{I}_\text{appx}$ in a binary search tree; by doing this, we can answer the size, membership, and reporting queries for $\mathcal{Q}_\text{appx}$ in the required query times.
This step takes $\widetilde{O}(\delta)$ time, i.e., $\widetilde{O}(\min\{r/\varepsilon,n\})$ time, since $|\mathcal{I}_\text{appx}| \leq \delta$ in this case.
If the output-sensitive algorithm fails, we compute the sets $P$ and $P'$.
This can be done in $\widetilde{O}(r)$ time by using the support data structure $\mathcal{B}_2$.
After this, we compute $\mathcal{I}^*$, which again takes $\widetilde{O}(r)$ time by using $\mathcal{B}_2$; specifically, we consider each $i \in P$ and use $\mathcal{B}_2$ to find an interval in $\mathcal{I}$ that contains $J_i$.
We have $\mathcal{I}_\text{appx} = \mathcal{I}^* \sqcup (\bigsqcup_{i \in P'} \mathcal{I}_i^*)$.
To support the size query for $\mathcal{I}_\text{appx}$ in $O(1)$ time, we need to compute $|\mathcal{I}_\text{appx}| = |\mathcal{I}^*| + \sum_{i \in P'} |\mathcal{I}_i^*|$.
This can be done in $O(r)$ time, because we can query $\mathcal{D}_\text{old}^{(i)}$ to obtain $|\mathcal{I}_i^*|$ in $O(1)$ time.
In order to support the membership and reporting queries, we store $\mathcal{I}^*$ in a binary search trees $\mathcal{T}$.
Also, we store the set $P'$, and for each $i \in P'$ we store a pointer pointing to the data structure $\mathcal{D}_\text{old}^{(i)}$.
Consider a membership query $I \in \mathcal{I}$.
Using the binary search tree $\mathcal{T}$, we can obtain the multiplicity of $I$ in $\mathcal{I}^*$ in $O(\log |\mathcal{I}^*|)$ time.
Then we use $\mathcal{B}_1$ to find in $O(\log r)$ time the (at most) two $\mathcal{I}_i$'s that contain $I$ (say $I \in \mathcal{I}_i$ and $I \in \mathcal{I}_j$).
By querying $\mathcal{D}_\text{old}^{(i)}$ and $\mathcal{D}_\text{old}^{(j)}$, we know the multiplicities of $I$ in $\mathcal{I}_i^*$ and $\mathcal{I}_j^*$, which takes $O(\log |\mathcal{I}_i^*| + \log |\mathcal{I}_j^*|)$.
The sum of these multiplicities is just the multiplicity of $I$ in $\mathcal{I}_{appx}$.
The time for answering the query is $O(\log |\mathcal{I}^*| + \log |\mathcal{I}_i^*| + \log |\mathcal{I}_j^*| + \log r)$.
Note that $\log |\mathcal{I}^*| + \log |\mathcal{I}_i^*| + \log |\mathcal{I}_j^*| = O(\log |\mathcal{I}_\text{appx}|)$.
Furthermore, because $|\mathcal{I}_\text{appx}| > \delta$ (as the output-sensitive algorithm fails), we have $|\mathcal{I}_\text{appx}| = \Omega(r)$.
Thus, the time cost for a membership query is $O(\log |\mathcal{I}_\text{appx}|)$.
Finally, consider the reporting query.
We first use the binary search tree $\mathcal{T}$ to report the elements in $\mathcal{I}^*$, using $O(|\mathcal{I}^*|)$ time.
Then for each $i \in P'$, we query the data structure $\mathcal{D}_\text{old}^{(i)}$ to report the elements in $\mathcal{I}_i^*$.
The total time cost is $O(|\mathcal{I}^*|+\sum_{i \in P'} |\mathcal{I}_i^*| + |P'|)$.
Since $|P'| \leq r$ and $|\mathcal{I}_\text{appx}| = \Omega(r)$ as argued before, the time for answering the reporting query is $O(|\mathcal{I}_\text{appx}|)$.
The above work for storing $\mathcal{I}_\text{appx}$ takes $\widetilde{O}(r)$ time, since $|\mathcal{I}^*| = O(r)$ and $|P'| = O(r)$.
To summarize, maintaining the solution takes $\widetilde{O}(\min\{r/\varepsilon,n\} + r)$ time.

After processing $f(n_0,\varepsilon)$ operations, we need to reconstruct the entire data structure $\mathcal{D}_\text{new}$.
The reconstruction is the same as the initial construction, except that $n_0$ is replaced with $n_1$, the size of $(S,\mathcal{I})$ at the time of reconstruction.
Thus, the reconstruction takes $\widetilde{O}(n_1)$ time.
We amortize the time cost over all the $f(n_0,\varepsilon)$ operations in the period.
Since $n_1 \leq n_0+f(n_0,\varepsilon)$, the amortized time for reconstruction is $\widetilde{O}(n_0/f(n_0,\varepsilon))$, i.e., $\widetilde{O}(r)$.

Combining the time for updating the $\mathcal{D}_\text{old}$ data structures, the time for maintaining the solution, and the time for reconstruction, we see that the amortized update time of $\mathcal{D}_\text{new}$ is $\widetilde{O}(f^\alpha(n_0,\eps)/\varepsilon^{1-\alpha} + \min\{r/\varepsilon,n\} + r)$ during the first period (since $\mathcal{D}_\text{new}$ is reconstructed periodically, it suffices to analyze the update time in the first period).
By property \textbf{(1)} of $f$, we have $n = \Theta(n_0)$ at any time in the period, i.e., the size of $(S,\mathcal{I})$ is $\Theta(n_0)$ at any time in the period.
By property \textbf{(2)} of $f$, we further have $f(n,\eps) = \Theta(f(n_0,\eps))$ at any time in the period.
It follows that the amortized update time of $\mathcal{D}_\text{new}$ is $\widetilde{O}(f^\alpha(n,\eps)/\varepsilon^{1-\alpha} + \min\{n/(f(n,\eps) \cdot \eps),n\} + n/f(n,\eps))$ during the period.
To minimize the time complexity while guaranteeing the two conditions of $f$, we set $f(n,\eps) = \min\{n^{1-\alpha'}/\eps^{\alpha'},n/2\}$ where $\alpha'$ is as defined in Theorem~\ref{thm-boot}, i.e., $\alpha' = \alpha/(1+\alpha)$.
The following lemma shows that our choice of $f$ makes the time bound $\widetilde{O}(n^{\alpha'}/\eps^{1-\alpha'})$.
\begin{lemma}
When $f(n,\eps) = \min\{n^{1-\alpha'}/\eps^{\alpha'},n/2\}$, we have
\begin{equation*}
    \frac{f^\alpha(n,\eps)}{\eps^{1-\alpha}} + \min\left\{\frac{n}{f(n,\eps) \cdot \eps},n\right\} + \frac{n}{f(n,\eps)} = O\left(\frac{n^{\alpha'}}{\eps^{1-\alpha'}}\right).
\end{equation*}
\end{lemma}
\begin{proof}
If $f(n,\eps) = n^{1-\alpha'}/\eps^{\alpha'}$, then one can easily verify the equation in the lemma via a direct computation (by bounding each of the three terms on the left-hand side).
It suffices to verify the equation for the case $f(n,\eps) = n/2$.
In this case, we have $n^{1-\alpha'}/\eps^{\alpha'} \geq n/2$, implying that $n = O(1/\eps)$.
It follows that $n^\alpha/\eps^{1-\alpha} = O(n^{\alpha'}/\eps^{1-\alpha'})$.
So the first term in the left-hand side is bounded by $O(n^{\alpha'}/\eps^{1-\alpha'})$.
The second term is $O(\min\{1/\eps,n\})$, which is bounded by $O(n^{\alpha'}/\eps^{1-\alpha'})$.
The third term is clearly $O(1)$.
This proves the equation in the lemma.
\end{proof}


\section{Implementation details and detailed time analysis for the dynamic quadrant set cover data structure} \label{appx-anaquad}
We present the implementation details of our dynamic quadrant set cover data structure $\mathcal{D}_\text{new}$ in Section~\ref{sec-DQSC} as well as a detailed time analysis.
Since we are interested in the asymptotic bounds, we may assume that the approximation factor $\eps$ is sufficiently small, say $\eps < 1$.
Assume the function $f$ we choose satisfies two properties: \textbf{(1)} $\sqrt{m}/2 \leq f(m,\eps) \leq m/2$ for any $m$, and \textbf{(2)} $f(\Theta(m),\eps) = \Theta(f(m,\eps))$.
Note that property \textbf{(1)} implies that $r = \lceil n_0/f(n_0,\eps) \rceil = O(f(n_0,\eps))$ and $r^2 = O(n_0)$.

First, we discuss how to construct $\mathcal{D}_\text{new}$.
Constructing the data structure $\mathcal{A}$ takes $\widetilde{O}(n_0)$ time, as it is basic.
The grid can be built in $\widetilde{O}(n_0)$ time by sorting the points in $S$ and the vertices of the quadrants in $\mathcal{Q}$.
Once the grid is computed, we build a (static) point location data structure $\mathcal{B}_1$ which can report in $O(\log r)$ time, for a given point $a \in \mathbb{R}^2$, the grid cell that contains $a$.
Since the grid has $r^2$ cells, $\mathcal{B}_1$ can be built in $\widetilde{O}(r^2)$ time.
With $\mathcal{B}_1$ in hand, we can determine in $\widetilde{O}(1)$ time, for each point $a \in S$ (resp., each quadrant $Q \in \mathcal{Q}$), the cell contains $a$ (resp., the vertex of $Q$).
By doing this for all points in $S$ and all quadrants in $\mathcal{Q}$, we obtain all $S_{i,j}$ and the non-special quadrants in all $\mathcal{Q}_{i,j}$ in $\widetilde{O}(n_0 + r^2)$ time, i.e., $\widetilde{O}(n_0)$ time.
In order to compute the special quadrants, we need the following support data structure $\mathcal{B}_2$.
\begin{lemma} \label{lem-supspecquad}
One can store $\mathcal{Q}$ in a basic data structure $\mathcal{B}_2$ such that given an orthogonal rectangle $\Box$, the leftmost (resp., rightmost) quadrant in $\mathcal{Q}$ that right (resp., left) intersects $\Box$ and the topmost (resp., bottommost) quadrant in $\mathcal{Q}$ that bottom (resp., top) intersects $\Box$ can be reported in $\widetilde{O}(1)$ time (if they exist).
\end{lemma}
\begin{proof}
It suffices to show how to compute the leftmost quadrant in $\mathcal{Q}$ that right intersects a given orthogonal rectangle $\Box$.
Note that only southeast and northeast quadrants can left intersects a rectangle, and without loss of generality, it suffices to see how to find the leftmost southeast quadrant in $\mathcal{Q}$ that right intersects $\Box$.
Let $\mathcal{Q}^\text{SE} \subseteq \mathcal{Q}$ consist of the southeast quadrants.
We store the vertices of the quadrants in $\mathcal{Q}^\text{SE}$ in a basic 3-sided range-minimum data structure (Lemma~\ref{lem-3sided}) with $\widetilde{O}(1)$ query time, by setting the weight of each vertex to be its $x$-coordinate.
Let $\Box = [x_1,x_2] \times [y_1,y_2]$ be a given orthogonal rectangle.
Observe that a quadrant in $\mathcal{Q}^\text{SE}$ right intersects $\Box$ iff its vertex lies in the 3-sided rectangle $(x_1,x_2] \times [y_2,\infty)$.
Thus, the leftmost quadrant in $\mathcal{Q}^\text{SE}$ that right intersects $\Box$ corresponds to the lightest vertex (i.e., the vertex with the smallest weight) in $(x_1,x_2] \times [y_2,\infty)$, which can be reported by the 3-sided range-minimum data structure in $\widetilde{O}(1)$ time.
\end{proof}

Using the support data structure $\mathcal{B}_2$, we can find the special quadrants in all $\mathcal{Q}_{i,j}$ in $\widetilde{O}(r^2)$ time.
After this, we can build the data structures $\mathcal{D}_\text{old}^{(i,j)}$'s.
Constructing each $\mathcal{D}_\text{old}^{(i,j)}$ takes $\widetilde{O}(|S_{i,j}| + |\mathcal{Q}_{i,j}|)$ time.
By Lemma~\ref{lem-bound}, the time for constructing all $\mathcal{D}_\text{old}^{(i,j)}$'s is $\widetilde{O}(n_0+r^2)$.
Therefore, the entire construction time of $\mathcal{D}_\text{new}$ is $\widetilde{O}(n_0+r^2)$, i.e., $\widetilde{O}(n_0)$.

The support data structures $\mathcal{B}_1$ and $\mathcal{B}_2$ will be used later in the implementation of the update procedure of $\mathcal{D}_\text{new}$.
Thus, $\mathcal{B}_2$ will be updated after each operation (while $\mathcal{B}_1$ is static).
Besides $\mathcal{B}_1$ and $\mathcal{B}_2$, we need another support data structure $\mathcal{B}_3$ defined as follows.
\begin{lemma} \label{lem-supquadcover}
One can store $\mathcal{Q}$ in a basic data structure $\mathcal{B}_3$ such that given an orthogonal rectangle $\Box$, a quadrant $Q \in \mathcal{Q}$ that contains $\Box$ can be reported in $\widetilde{O}(1)$ time (if it exists).
\end{lemma}
\begin{proof}
It suffices to consider the southeast quadrants.
Let $\mathcal{Q}^\text{SE} \subseteq \mathcal{Q}$ consist of the southeast quadrants.
We store $\mathcal{Q}^\text{SE}$ in a binary search tree $\mathcal{T}$ where the key of a quadrant is the $x$-coordinate of its vertex.
We augment each node $\mathbf{u} \in \mathcal{T}$ with an additional field which stores the topmost quadrant in the subtree rooted at $\mathbf{u}$.
Given an orthogonal rectangle $\Box = [x_1,x_2] \times [y_1,y_2]$, we first look for the topmost quadrant $Q$ whose key is smaller than or equal to $x_1$.
With the augmented fields, $Q$ can be found in $\widetilde{O}(1)$ time using $\mathcal{T}$.
If $Q$ contains $\Box$, then we report $Q$, otherwise no quadrant in $\mathcal{Q}^\text{SE}$ contains $\Box$ (because a southeast quadrant contains $\Box$ iff the $x$-coordinate of its vertex is smaller than or equal to $x_1$ and the $y$-coordinate of its vertex is greater than or equal to $y_2$).
Clearly, $\mathcal{T}$ can be built in $\widetilde{O}(n_0)$ time and dynamized with $\widetilde{O}(1)$ update time, and thus $\mathcal{T}$ is basic.
\end{proof}

Next, we consider how to implement the update procedure of $\mathcal{D}_\text{new}$.
We first discuss the update of the $\mathcal{D}_\text{old}$ data structures.
If the operation is an insertion or deletion on $S$, then we use $\mathcal{B}_1$ to find the cell $\Box_{i,j}$ that contains the inserted/deleted point and update the data structure $\mathcal{D}_\text{old}^{(i,j)}$.
The situation is more complicated when the operation happens on $\mathcal{Q}$.
Suppose the operation inserts a quadrant $Q$ to $\mathcal{Q}$.
Without loss of generality, assume $Q$ is a southeast quadrant.
Using $\mathcal{B}_1$, we can find the cell $\Box_{i,j}$ that contains the vertex of $Q$.
We then update the data structure $\mathcal{D}_\text{old}^{(i,j)}$ by inserting $Q$ to $\mathcal{Q}_{i,j}$.
Besides, the insertion of $Q$ may also change $\mathcal{Q}_{i,k}$ for $k \in \{j+1,\dots,r\}$ and $\mathcal{Q}_{k,j}$ for $k \in \{i+1,\dots,r\}$.
Fix an index $k \in \{j+1,\dots,r\}$.
The quadrant $Q$ bottom intersects $\Box_{i,k}$.
We use the support data structure $\mathcal{B}_2$ to find the topmost quadrant $Q'$ in $\mathcal{Q}$ (before the insertion of $Q$) that bottom intersects $\Box_{i,k}$.
Then $Q'$ is a special quadrant in $\mathcal{Q}_{i,k}$.
If the $y$-coordinate of the vertex of $Q$ is greater than the $y$-coordinate of the vertex of $Q'$, then we delete $Q'$ from $\mathcal{Q}_{i,k}$ and insert $Q$ to $\mathcal{Q}_{i,k}$, and we update $\mathcal{D}_\text{old}^{(i,k)}$ twice for these two operations.
The data structures $\mathcal{D}_\text{old}^{(k,j)}$'s can be updated similarly.
Updating each $\mathcal{D}_\text{old}^{(i,k)}$ takes $\widetilde{O}(m_{i,k}^\alpha/\varepsilon^{1-\alpha})$ amortized time, where $m_{i,k}$ is the size of the current $(S_{i,k},\mathcal{Q}_{i,k})$, and updating each $\mathcal{D}_\text{old}^{(k,j)}$ takes $\widetilde{O}(m_{k,j}^\alpha/\varepsilon^{1-\alpha})$ amortized time.
Therefore, the total amortized time cost for updating these $\mathcal{D}_\text{old}$ data structures is bounded by $\widetilde{O}(\sum_{k=1}^r m_{i,k}^\alpha/\varepsilon^{1-\alpha} + \sum_{k=1}^r m_{k,j}^\alpha/\varepsilon^{1-\alpha})$.
By Lemma~\ref{lem-bound} and the fact $r = O(f(n_0,\varepsilon))$, we have $\sum_{k=1}^r m_{i,k} = O(f(n_0,\varepsilon))$ and $\sum_{k=1}^r m_{k,j} = O(f(n_0,\varepsilon))$ at any time of the first period.
Since $\alpha \leq 1$, by H\"older's inequality and Lemma~\ref{lem-bound},
\begin{equation*}
    \sum_{k=1}^r m_{i,k}^\alpha \leq \left(\frac{\sum_{k=1}^r m_{i,k}}{r}\right)^\alpha \cdot r = O(r^{1-\alpha} \cdot f^\alpha(n_0,\eps))
\end{equation*}
and similarly $\sum_{k=1}^r m_{k,j}^\alpha = O(r^{1-\alpha} \cdot f^\alpha(n_0,\eps))$.
It follows that updating the $\mathcal{D}_\text{old}$ data structures takes $\widetilde{O}(r^{1-\alpha} \cdot f^\alpha(n_0,\eps) / \eps^{1-\alpha})$ amortized time.
Updating the data structure $\mathcal{A}$ and the support data structures $\mathcal{B}_2$ and $\mathcal{B}_3$ can be done in $\widetilde{O}(1)$ time since they are basic.

Then we discuss the maintenance of the solution.
The time for simulating the output-sensitive algorithm is $\widetilde{O}(\mu \cdot \delta)$, i.e., $\widetilde{O}(\min\{r^2/\varepsilon,n\})$.
If the algorithm gives the solution $\mathcal{Q}_\text{appx}$, we compute $|\mathcal{Q}_\text{appx}|$ and store $\mathcal{Q}_\text{appx}$ in a binary search tree; by doing this, we can answer the size, membership, and reporting queries for $\mathcal{Q}_\text{appx}$ in the required query times.
This step takes $\widetilde{O}(\mu \cdot \delta)$ time, i.e., $\widetilde{O}(\min\{r^2/\varepsilon,n\})$ time, since $|\mathcal{Q}_\text{appx}| \leq \mu \cdot \delta$ in this case.
If the output-sensitive algorithm fails, we compute the sets $P$ and $P'$.
This can be done in $\widetilde{O}(r^2)$ time by using the support data structure $\mathcal{B}_3$.
After this, we compute $\mathcal{Q}^*$, which again takes $\widetilde{O}(r^2)$ time by using $\mathcal{B}_3$; specifically, we consider each $(i,j) \in P$ and use $\mathcal{B}_3$ to find a quadrant in $\mathcal{Q}$ that contains $\Box_{i,j}$.
We have $\mathcal{Q}_\text{appx} = \mathcal{Q}^* \sqcup (\bigsqcup_{(i,j) \in P'} \mathcal{Q}_{i,j}^*)$.
To support the size query for $\mathcal{Q}_\text{appx}$ in $O(1)$ time, we need to compute $|\mathcal{Q}_\text{appx}| = |\mathcal{Q}^*| + \sum_{(i,j) \in P'} |\mathcal{Q}_{i,j}^*|$.
This can be done in $O(r^2)$ time, because we can query $\mathcal{D}_\text{old}^{(i,j)}$ to obtain $|\mathcal{Q}_{i,j}^*|$ in $O(1)$ time.
To support the reporting query in $O(|\mathcal{Q}_\text{appx}|)$ time, we only need to store $\mathcal{Q}^*$ and $P'$, and store at each $(i,j) \in P'$ a pointer pointing to the data structure $\mathcal{D}_\text{old}^{(i,j)}$.
In this way, we can report the quadrants in $\mathcal{Q}^*$ and for each $(i,j) \in P'$, report the quadrants in $\mathcal{Q}_{i,j}^*$ in $O(|\mathcal{Q}_{i,j}^*|)$ time by querying $\mathcal{D}_\text{old}^{(i,j)}$.
Supporting the membership query in $O(\log |\mathcal{Q}_\text{appx}|)$ time is more difficult, since a quadrant may belong to many $\mathcal{Q}_{i,j}^*$'s.
To handle this issue, the idea is to collect all the special quadrants in the $\mathcal{Q}_{i,j}^*$'s.
Specifically, let $\mathcal{P}_{i,j}^* \subseteq \mathcal{Q}_{i,j}^*$ consist of the (at most) four special quadrants.
We can compute $\mathcal{P}_{i,j}^*$ for each $(i,j) \in P'$ in $\widetilde{O}(1)$ time by first finding the (at most) four special quadrants in $\mathcal{Q}_{i,j}$ using $\mathcal{B}_2$ and computing the multiplicity of each special quadrant in $\mathcal{Q}_{i,j}^*$ by querying $\mathcal{D}_\text{old}^{(i,j)}$.
We then store $\mathcal{Q}^* \sqcup (\bigsqcup_{(i,j) \in P'} \mathcal{P}_{i,j}^*)$ in a binary search tree $\mathcal{T}$.
Given a query quadrant $Q \in \mathcal{Q}$, we first use $\mathcal{T}$ to compute the multiplicity $p_1$ of $Q$ in $\mathcal{Q}^* \sqcup (\bigsqcup_{(i,j) \in P'} \mathcal{P}_{i,j}^*)$ in $O(\log |\mathcal{Q}_\text{appx}|)$ time.
Then we use $\mathcal{B}_1$ to find the cell $\Box_{i,j}$ that contains the vertex of $Q$ in $O(\log r)$ time and query $\mathcal{D}_\text{old}^{(i,j)}$ to find the multiplicity $p_2$ of $Q$ in $\mathcal{Q}_{i,j}^*$ in $O(\log |\mathcal{Q}_{i,j}^*|)$ time.
One can easily verify that $p_1+p_2$ is the multiplicity of $Q$ in $\mathcal{Q}_\text{appx}$.
The query takes $O(\log |\mathcal{Q}_\text{appx}| + \log r + \log |\mathcal{Q}_{i,j}^*|)$ time.
Note that $|\mathcal{Q}_{i,j}^*| \leq |\mathcal{Q}_\text{appx}|$ and $|\mathcal{Q}_\text{appx}| = \Omega(r)$, where the latter follow from the fact that $|\mathcal{Q}_\text{appx}| > \delta$ (as the output-sensitive algorithm fails).
Therefore, we can support the membership query in $O(\log |\mathcal{Q}_\text{appx}|)$ time.
The above work for storing $\mathcal{Q}_\text{appx}$ takes $\widetilde{O}(r^2)$ time, since $|\mathcal{Q}^*| = O(r^2)$ and $|P'| = O(r^2)$.
To summarize, maintaining the solution takes $\widetilde{O}(\min\{r^2/\varepsilon,n\} + r^2)$ time.

After processing $f(n_0,\varepsilon)$ operations, we need to reconstruct the entire data structure $\mathcal{D}_\text{new}$.
The reconstruction is the same as the initial construction, except that $n_0$ is replaced with $n_1$, the size of $(S,\mathcal{Q})$ at the time of reconstruction.
Thus, the reconstruction takes $\widetilde{O}(n_1)$ time.
We amortize the time cost over all the $f(n_0,\varepsilon)$ operations in the period.
Since $n_1 \leq n_0+f(n_0,\varepsilon)$, the amortized time for reconstruction is $\widetilde{O}(n_0/f(n_0,\varepsilon))$, i.e., $\widetilde{O}(r)$.

Combining the time for updating the $\mathcal{D}_\text{old}$ data structures, the time for maintaining the solution, and the time for reconstruction, we see that the amortized update time of $\mathcal{D}_\text{new}$ is $\widetilde{O}(r^{1-\alpha} \cdot f^\alpha(n_0,\eps) / \eps^{1-\alpha} + \min\{r^2/\varepsilon,n\} + r^2)$ during the first period (since $\mathcal{D}_\text{new}$ is reconstructed periodically, it suffices to analyze the update time in the first period).
By property \textbf{(1)} of $f$, we have $n = \Theta(n_0)$ at any time in the period, i.e., the size of $(S,\mathcal{Q})$ is $\Theta(n_0)$ at any time in the period.
By property \textbf{(2)} of $f$, we further have $f(n,\eps) = \Theta(f(n_0,\eps))$ at any time in the period.
It follows that the amortized update time of $\mathcal{D}_\text{new}$ is $\widetilde{O}(n^{1-\alpha}/ (f^{1-2\alpha}(n,\eps) \cdot \eps^{1-\alpha}) + \min\{n^2/(f^2(n,\eps) \cdot \eps),n\} + n^2/f^2(n,\eps))$.
To minimize the time complexity while guaranteeing the two conditions of $f$, we set $f(n,\eps) = \min\{n^{1-\alpha'/2}/(\sqrt{\eps})^{\alpha'},n/2\}$ where $\alpha'$ is as defined in Theorem~\ref{thm-bootquad}, i.e., $\alpha' = 2\alpha/(1+2\alpha)$.
Note that by doing this we have $f(n,\eps) \geq \sqrt{n}$ because of our assumption $\varepsilon < 1$.
The following lemma shows that our choice of $f$ makes the time bound $\widetilde{O}(n^{\alpha'}/\eps^{1-\alpha'})$.
\begin{lemma}
When $f(n,\eps) = \min\{n^{1-\alpha'/2}/(\sqrt{\eps})^{\alpha'},n/2\}$, we have
\begin{equation*}
    \frac{n^{1-\alpha}}{f^{1-2\alpha}(n,\eps) \cdot \eps^{1-\alpha}} + \min\left\{\frac{n^2}{f^2(n,\eps) \cdot \eps},n\right\} + \frac{n^2}{f^2(n,\eps)} = O\left(\frac{n^{\alpha'}}{\eps^{1-\alpha'}}\right).
\end{equation*}
\end{lemma}
\begin{proof}
If $f(n,\eps) = n^{1-\alpha'/2}/(\sqrt{\eps})^{\alpha'}$, then one can easily verify the equation in the lemma via a direct computation (by bounding each of the three terms on the left-hand side).
It suffices to verify the equation for the case $f(n,\eps) = n/2$.
In this case, we have $n^{1-\alpha'/2}/(\sqrt{\eps})^{\alpha'} \geq n/2$, implying that $n = O(1/\eps)$.
It follows that $n^\alpha/\eps^{1-\alpha} = O(n^{\alpha'}/\eps^{1-\alpha'})$.
So the first term in the left-hand side is bounded by $O(n^{\alpha'}/\eps^{1-\alpha'})$.
The second term is $O(\min\{1/\eps,n\})$, which is bounded by $O(n^{\alpha'}/\eps^{1-\alpha'})$.
The third term is clearly $O(1)$.
This proves the equation in the lemma.
\end{proof}

\section{Missing details in the output-sensitive quadrant set cover algorithm}

\subsection{Handling the no-solution case} \label{appx-nosol}
Let $U = \bigcup_{Q \in \mathcal{Q}} Q$.
Observe that no matter whether $(S,\mathcal{Q})$ has no set cover or not, the algorithm in Section~\ref{sec-osquad} gives an $O(1)$-approximate optimal set cover $\mathcal{Q}^*$ for $(S \cap U,\mathcal{Q})$.
Thus, in order to handle the no-solution case, we only need to check whether $\mathcal{Q}^*$ covers all points in $S$, after $\mathcal{Q}^*$ is computed.
Define $U^* = \bigcup_{Q \in \mathcal{Q}^*} Q$.
Note that $\mathcal{Q}^*$ is a set cover for $S$ iff $\mathbb{R}^2 \backslash U^*$ does not contain any point in $S$.
The area $\mathbb{R}^2 \backslash U^*$ is a rectilinear domain (not necessarily connected).
Since $U^*$ is the union of $O(|\mathcal{Q}^*|)$ quadrants, the complexity of $U^*$ is $O(|\mathcal{Q}^*|)$, so is the complexity of $\mathbb{R}^2 \backslash U^*$.
Furthermore, it is easy to compute $\mathbb{R}^2 \backslash U^*$ and decompose it into $O(|\mathcal{Q}^*|)$ rectangles in $\widetilde{O}(|\mathcal{Q}^*|)$ time, given $\mathcal{Q}^*$ in hand.
We then only need to test for each such rectangle $R$ whether $R$ contains any point in $S$ or not, which can be done via an orthogonal range-emptiness query on $S$; there are existing basic data structures that support orthogonal range-emptiness queries in $\widetilde{O}(1)$ time~\cite{mortensen2006fully}.

\subsection{Implementing the algorithm using basic data structures} \label{appx-quados}
In this section, we show that the output-sensitive quadrant set cover algorithm can be performed in $\widetilde{O}(\mathsf{opt})$ time by storing the instance $(S,\mathcal{Q})$ in some basic data structure.
As argued in Section~\ref{sec-osquad}, it suffices to show how to implement the following operations in $\widetilde{O}(1)$ time using basic data structures (please refer to Section~\ref{sec-osquad} for the notations).
\begin{itemize}
    \item Computing the point $\sigma$.
    \item Given a point $a \in \mathbb{R}^2$, testing whether $a \in U^\text{NE}$ and $a \in U^\text{NW}$.
    \item Given a point $a \in \mathbb{R}^2$, computing the quadrants $\varPhi_\rightarrow(a,\mathcal{Q}^\text{SW})$, $\varPhi_\rightarrow(a,\mathcal{Q}^\text{NW})$, $\varPhi_\uparrow(a,\mathcal{Q}^\text{SE})$, and $\varPhi_\uparrow(a,\mathcal{Q}^\text{NE})$.
    \item Given a number $\tilde{y}$, computing the point $\phi(\tilde{y})$.
\end{itemize}

\paragraph{Computing $\sigma$.}
Recall that $\gamma$ is the boundary of $U^\text{SE}$, which is a staircase curve from bottom-left to top-right.
The area $U^\text{SW}$ contains the bottom-left end of $\gamma$ (if $\gamma \cap U^\text{SW} \neq \emptyset$).
The point $\sigma$ is the ``endpoint'' of $\gamma \cap U^\text{SW}$, i.e., the point on $\gamma$ closest to the top-right end of $\gamma$ that is contained in $U^\text{SW}$.
We define the \textit{height} of $U^\text{SW}$ at $x \in \mathbb{R}$, denoted by $\mathsf{ht}_x(U^\text{SW})$, as the largest
number $y \in \mathbb{R}$ such that the point with coordinates $(x,y)$ is contained in $U^\text{SW}$; similarly, we can define the height of $U^\text{SE}$ at $x \in \mathbb{R}$, denoted by $\mathsf{ht}_x(U^\text{SE})$.
We say $U^\text{SW}$ is higher than $U^\text{SE}$ at $x \in \mathbb{R}$ if $\mathsf{ht}_x(U^\text{SW}) \geq \mathsf{ht}_x(U^\text{SE})$; otherwise, we say $U^\text{SW}$ is lower than $U^\text{SE}$ at $x$.
It is easy to see the following three facts about the point $\sigma$.
\begin{enumerate}
    \item The $x$-coordinate of $\sigma$, denoted by $x_\sigma$, is equal to $x(Q)$ for some $Q \in \mathcal{Q}^\text{SE} \cup \mathcal{Q}^\text{SW}$; recall that $x(Q)$ is the $x$-coordinate of the vertex of $Q$.
    \item $U^\text{SW}$ is higher than $U^\text{SE}$ at all $x < x_\sigma$ and $U^\text{SW}$ is lower than $U^\text{SE}$ at all $x > x_\sigma$.
    \item The coordinates of $\sigma$ is $(x_\sigma, \min\{\mathsf{ht}_{x_\sigma}(U^\text{SW}),\mathsf{ht}_{x_\sigma}(U^\text{SE})\})$.
\end{enumerate}
Based on these facts, we compute $\sigma$ as follows.
First, we need a basic data structure built on $\mathcal{Q}^\text{SW}$ (resp., $\mathcal{Q}^\text{SE}$) that can compute the height function in $\widetilde{O}(1)$ time.
\begin{lemma} \label{lem-htfunc}
One can store $\mathcal{Q}^\textnormal{SW}$ (resp., $\mathcal{Q}^\textnormal{SE}$) in a basic data structure which can report $\mathsf{ht}_x(U^\textnormal{SW})$ (resp., $\mathsf{ht}_x(U^\textnormal{SE})$) for a given number $x \in \mathbb{R}$ in $\widetilde{O}(1)$ time.
\end{lemma}
\begin{proof}
It suffices to consider $\mathcal{Q}^\text{SW}$.
We store $\mathcal{Q}^\text{SW}$ in a standard binary search tree $\mathcal{T}$, by using $x(Q)$ as the key of each quadrant $Q \in \mathcal{Q}^\text{SW}$.
At each node $\mathbf{u} \in \mathcal{T}$, we store a field $Y(\mathbf{u})$ which is the maximum of $y(Q)$ for all quadrants $Q$ in the subtree rooted at $\mathbf{u}$.
Clearly, $\mathcal{T}$ can be constructed in $\widetilde{O}(n)$ time where $n = |\mathcal{Q}^\text{SW}|$ and can be dynamized with $\widetilde{O}(1)$ update time to support insertions and deletions on $\mathcal{Q}^\text{SW}$, hence it is a basic data structure.
Next, we consider how to compute $\mathsf{ht}_x(U^\text{SW})$ for a given $x \in \mathbb{R}$ using $\mathcal{T}$.
One can easily see that $\mathsf{ht}_x(U^\text{SW}) = \max\{y(Q): Q \in \mathcal{Q}^\text{SW} \text{ with } x(Q) \geq x\}$.
In other words, $\mathsf{ht}_x(U^\text{SW})$ is the maximum of $y(Q)$ for all $Q \in \mathcal{Q}^\text{SW}$ corresponding to the nodes in $\mathcal{T}$ whose keys are at least $x$.
Therefore, using the field $Y(\mathbf{u})$, $\mathsf{ht}_x(U^\text{SW})$ can be computed in $O(\log n)$ time simply via a top-down walk in $\mathcal{T}$.
The walk begins at the root of $\mathcal{T}$, and we set $\mathsf{ht} = -\infty$ initially.
If the key of the current node is smaller than $x$, then we just go to its right child.
Otherwise, we update $\mathsf{ht}$ as $\mathsf{ht} \leftarrow \max\{\mathsf{ht},y(Q),Y(\mathbf{r})\}$ where $Q \in \mathcal{Q}^\text{SW}$ is the quadrant corresponding to the current node and $\mathbf{r}$ is the right child of the current node (if the current node has no right child, we update $\mathsf{ht}$ as $\mathsf{ht} \leftarrow \max\{\mathsf{ht},y(Q)\}$), and go to the left child of the current node.
When the walk ends at a leaf node, the number $\mathsf{ht}$ is just equal to $\mathsf{ht}_x(U^\text{SW})$.
The walk takes $O(\log n)$ time, hence $\mathsf{ht}_x(U^\text{SW})$ can be computed in $\widetilde{O}(1)$ time using $\mathcal{T}$.
This completes the proof of the lemma.
\end{proof}

With the above lemma in hand, we may now assume that the height functions $\mathsf{ht}_x(U^\text{SW})$ and $\mathsf{ht}_x(U^\text{SE})$ can be computed in $\widetilde{O}(1)$ time for any $x \in \mathbb{R}$.
In particular, we can test in $\widetilde{O}(1)$ time whether $U^\text{SW}$ is higher or lower than $U^\text{SE}$ at any $x \in \mathbb{R}$.
We then build a binary search tree $\mathcal{T}$ on $\mathcal{Q}^\text{SW} \cup \mathcal{Q}^\text{SE}$, by using $x(Q)$ as the key of each quadrant $Q \in \mathcal{Q}^\text{SW} \cup \mathcal{Q}^\text{SE}$.
Clearly, $\mathcal{T}$ can be constructed in $\widetilde{O}(n)$ time where $n = |\mathcal{Q}^\text{SW} \cup \mathcal{Q}^\text{SE}|$ and can be dynamized with $\widetilde{O}(1)$ update time, hence it is a basic data structure.
We observe that, using $\mathcal{T}$, we can determine in $\widetilde{O}(1)$ time for any given number $p \in \mathbb{R}$ which one of the following three is true: \textbf{(1)} $p < x_\sigma$, \textbf{(2)} $p = x_\sigma$, \textbf{(3)} $p > x_\sigma$.
To see this, consider a given number $p \in \mathbb{R}$.
We first search in $\mathcal{T}$ to see whether $p = x(Q)$ for some $Q \in \mathcal{Q}^\text{SW} \cup \mathcal{Q}^\text{SE}$.
If not, we know $p \neq x_\sigma$, because $x_\sigma = x(Q)$ for some $Q \in \mathcal{Q}^\text{SW} \cup \mathcal{Q}^\text{SE}$ (as we observed before).
In this case, we can decide whether $p < x_\sigma$ or $p > x_\sigma$ by simply testing whether $U^\text{SW}$ is higher or lower than $U^\text{SE}$ at $p$.
Specifically, if $U^\text{SW}$ is higher than $U^\text{SE}$ at $p$, then $p < x_\sigma$, otherwise $p > x_\sigma$, because $U^\text{SW}$ is higher (resp., lower) than $U^\text{SE}$ at all $x < x_\sigma$ (resp., $x > x_\sigma$) as we observed before.
The remaining case is that $p = x(Q)$ for some $Q \in \mathcal{Q}^\text{SW} \cup \mathcal{Q}^\text{SE}$.
In this case, we also test whether $U^\text{SW}$ is higher or lower than $U^\text{SE}$ at $p$; by doing this, we can decide whether $p \leq x_\sigma$ or $p \geq x_\sigma$.
Suppose $p \leq x_\sigma$.
To see whether $p < x_\sigma$ or $p = x_\sigma$, we search in $\mathcal{T}$ to find the smallest key $p'$ that is larger than $p$.
Let $\tilde{p} \in \mathbb{R}$ be any number such that $p < \tilde{p} < p'$.
If $U^\text{SW}$ is higher than $U^\text{SE}$ at $\tilde{p}$, then we have $p < \tilde{p} \leq x_\sigma$.
If $U^\text{SW}$ is lower than $U^\text{SE}$ at $\tilde{p}$, then we know $p \leq x_\sigma \leq \tilde{p}$.
Note that any number in $(p,\tilde{p}]$ is not equal to $x(Q)$ for any $Q \in \mathcal{Q}^\text{SW} \cup \mathcal{Q}^\text{SE}$, due to the choice of $p'$ and the inequality $p < \tilde{p} < p'$.
Therefore, $x_\sigma \notin (p,\tilde{p}]$, which implies $p = x_\sigma$.
To summarize, we can determine in $\widetilde{O}(1)$ time for any $p \in \mathbb{R}$ which one of the following three is true: \textbf{(1)} $p < x_\sigma$, \textbf{(2)} $p = x_\sigma$, \textbf{(3)} $p > x_\sigma$.

This allows us to compute $x_\sigma$ using a binary search manner.
Specifically, we do a top-down walk from the root of $\mathcal{T}$.
If the key of the current node is equal to $x_\sigma$, we are done.
Otherwise, if the key of the current node is smaller (resp., larger) than $x_\sigma$, we go to its right (resp., left) child, because the keys of the nodes in the left (resp., right) subtree are all smaller (resp., larger) than $x_\sigma$.
During the walk, we can definitely find a node whose key is $x_\sigma$, since $x_\sigma = x(Q)$ for some $Q \in \mathcal{Q}^\text{SW} \cup \mathcal{Q}^\text{SE}$, i.e., $x_\sigma$ is the key of some node in $\mathcal{T}$.
In this way, we can compute $x_\sigma$ in $\widetilde{O}(1)$ time.
As we observed before, the coordinates of $\sigma$ is $(x_\sigma,\min\{\mathsf{ht}_{x_\sigma}(U^\text{SW}),\mathsf{ht}_{x_\sigma}(U^\text{SE})\})$.
Thus, once we know $x_\sigma$, it suffices to compute $\mathsf{ht}_{x_\sigma}(U^\text{SW})$ and $\mathsf{ht}_{x_\sigma}(U^\text{SE})$, which takes $\widetilde{O}(1)$ time by Lemma~\ref{lem-htfunc}.
We conclude that computing $\sigma$ can be done in $\widetilde{O}(1)$ time (by properly store $\mathcal{Q}^\text{SW}$ and $\mathcal{Q}^\text{SE}$ in basic data structures).

\paragraph{Testing whether $a \in U^\text{NE}$ and $a \in U^\text{NW}$.}
It suffices to consider how to test whether $a \in U^\text{NE}$.
We store $\mathcal{Q}^\text{NE}$ in a binary search tree $\mathcal{T}$, by using $y(Q)$ as the key of each quadrant $Q \in \mathcal{Q}^\text{SW}$.
We augment each node $\mathbf{u} \in \mathcal{T}$ with a field which stores the leftmost quadrant in the subtree rooted at $\mathbf{u}$.
Clearly, $\mathcal{T}$ can be constructed in $\widetilde{O}(n)$ time where $n = |\mathcal{Q}^\text{NE}|$ and can be dynamized with $\widetilde{O}(1)$ update time, hence it is basic.
Given a point $a \in \mathbb{R}^2$, we first look for the leftmost quadrant $Q$ in $\mathcal{T}$ whose key is smaller than or equal to the $y$-coordinate of $a$.
With the augmented fields, $Q$ can be found in $\widetilde{O}(1)$ time.
If $a \in Q$, then we know $a \in U^\text{NE}$.
Otherwise, we claim that $a \notin U^\text{NE}$.
Indeed, a quadrant in $\mathcal{T}$ contains $a$ only if its key is smaller than or equal to the $y$-coordinate of $a$.
Since $Q$ is the leftmost one among such quadrants and $a \notin Q$, we know that $a$ is not contained in any quadrant in $\mathcal{T}$, i.e., $a \notin U^\text{NE}$.
We conclude that testing whether $a \in U^\text{NE}$ and $a \in U^\text{NW}$ for a given point $a \in \mathbb{R}^2$ can be done in $\widetilde{O}(1)$ time (by properly store $\mathcal{Q}^\text{NE}$ and $\mathcal{Q}^\text{NW}$ in basic data structures).

\paragraph{Computing $\varPhi$.}
Recall that for a point $a \in \mathbb{R}^2$ and a collection $\mathcal{P}$ of quadrants, $\varPhi_\rightarrow(a,\mathcal{P})$ and $\varPhi_\uparrow(a,\mathcal{P})$ denote the rightmost and topmost quadrants in $\mathcal{P}$ that contain $a$, respectively.
We want to compute $\varPhi_\rightarrow(a,\mathcal{Q}^\text{SW})$, $\varPhi_\rightarrow(a,\mathcal{Q}^\text{NW})$, $\varPhi_\uparrow(a,\mathcal{Q}^\text{SE})$, and $\varPhi_\uparrow(a,\mathcal{Q}^\text{NE})$ in $\widetilde{O}(1)$ time for a given point $a \in \mathbb{R}^2$, using basic data structures.
Here we only consider how to compute $\varPhi_\rightarrow(a,\mathcal{Q}^\text{SW})$, the other three can be computed in the same way.
We store $\mathcal{Q}^\text{SW}$ in a binary search tree $\mathcal{T}$, by using $y(Q)$ as the key of each quadrant $Q \in \mathcal{Q}^\text{SW}$.
At each node $\mathbf{u} \in \mathcal{T}$, we store a field that is the rightmost quadrant in the subtree rooted at $\mathbf{u}$.
Clearly, $\mathcal{T}$ can be constructed in $\widetilde{O}(n)$ time where $n = |\mathcal{Q}^\text{SW}|$ and can be dynamized with $\widetilde{O}(1)$ update time, hence it is a basic data structure.
Given a point $a \in \mathbb{R}^2$, we first look for the rightmost quadrant $Q$ in $\mathcal{T}$ whose key is greater than or equal to the $y$-coordinate $y_a$ of $a$.
With the augmented fields, $Q$ can be found in $\widetilde{O}(1)$ time.
If $Q$ contains $a$, then $Q$ is the rightmost quadrant in $\mathcal{T}$ (i.e., in $\mathcal{Q}^\text{SW}$) that contains $a$, because any quadrant $Q'$ that contains $a$ must satisfy $y(Q') \geq y_a$.
Otherwise, no quadrant in $\mathcal{Q}^\text{SW}$ contains $a$.
We conclude that computing $\varPhi_\rightarrow(a,\mathcal{Q}^\text{SW})$, $\varPhi_\rightarrow(a,\mathcal{Q}^\text{NW})$, $\varPhi_\uparrow(a,\mathcal{Q}^\text{SE})$, and $\varPhi_\uparrow(a,\mathcal{Q}^\text{NE})$ for a given point $a \in \mathbb{R}^2$ can be done in $\widetilde{O}(1)$ time (by properly store $\mathcal{Q}^\text{SW}$, $\mathcal{Q}^\text{NW}$, $\mathcal{Q}^\text{SE}$, $\mathcal{Q}^\text{NE}$ in basic data structures).

\paragraph{Computing $\phi(\tilde{y})$.}
Recall that for a number $\tilde{y} \in \mathbb{R}$, $\phi(\tilde{y})$ is the leftmost point in $S \cap U^\text{SE}$ whose $y$-coordinate is greater than $\tilde{y}$.
We want to store $S$ and $\mathcal{Q}^\text{SE}$ in some basic data structure such that $\phi(\tilde{y})$ can be computed in $\widetilde{O}(1)$ time for any given $\tilde{y} \in \mathbb{R}$.
For simplicity of exposition, let us make a general-position assumption: the points in $S$ and the vertices of the quadrants in $\mathcal{Q}^\text{SE}$ have distinct $y$-coordinates.
The first thing we need is a data structure that supports the so-called \textit{3-sided range-minimum} query.
A 3-sided range-minimum query on a set of \textit{weighted} points in $\mathbb{R}^2$ ask for the \textit{lightest} point (i.e., the point with the smallest weight) contained in a given 3-sided query rectangle $R = [x_0,\infty) \times [y_1,y_2]$.
\begin{lemma} \label{lem-3sided}
There exists a basic data structure that supports 3-sided range-minimum queries in $\widetilde{O}(1)$ time.
\end{lemma}
\begin{proof}
The standard range trees can answer \textit{static} 3-sided range-minimum queries in $\widetilde{O}(1)$ time, and can be constructed in $\widetilde{O}(n)$ time, where $n$ is the size of the dataset.
Since range-minimum queries are decomposable, the approach of~\cite{bentley1978decomposable} can be applied to dynamize the static data structure with $\widetilde{O}(1)$ update time, by paying an extra logarithmic factor in the query time.
This gives us the basic data structure that supports 3-sided range-minimum queries in $\widetilde{O}(1)$ time.
\end{proof}

We store $S$ in the basic 3-sided range-minimum data structure $\mathcal{A}$ of the above lemma, by setting the weight of each point in $S$ to be its $x$-coordinate.
Besides $\mathcal{A}$, we need a (1D) range tree $\mathcal{T}$ built on $S \cup V(\mathcal{Q}^\text{SE})$ for $y$-coordinates, where $V(\mathcal{Q}^\text{SE})$ is the set of the vertices of the quadrants in $\mathcal{Q}^\text{SE}$.
By the definition of a range tree, the points in $S \cup V(\mathcal{Q}^\text{SE})$ are one-to-one corresponding to the leaves of $\mathcal{T}$, where the left-right order of the leaves corresponds to the small-large order of the $y$-coordinates of the points.
The \textit{canonical subset} of each node $\mathbf{u} \in \mathcal{T}$ refers to the subset of $S \cup V(\mathcal{Q}^\text{SE})$ consisting of the points stored in the subtree rooted at $\mathbf{u}$.
For a node $\mathbf{u} \in \mathcal{T}$, we denote by $S(\mathbf{u})$ the set of the points in $S$ that are contained in the canonical subset of $\mathbf{u}$, and denote by $\mathcal{Q}(\mathbf{u})$ the collection of the quadrants in $\mathcal{Q}^\text{SE}$ whose vertices are contained in the canonical subset of $\mathbf{u}$.
Also, we write $U(\mathbf{u}) = \bigcup_{Q \in \mathcal{Q}(\mathbf{u})} Q$
At each node $\mathbf{u} \in \mathcal{T}$, we store the following three fields.
\begin{itemize}
    \item $y^-(\mathbf{u})$: the $y$-coordinate of the bottommost point in the canonical subset of $\mathbf{u}$.
    \item $y^+(\mathbf{u})$: the $y$-coordinate of the topmost point in the canonical subset of $\mathbf{u}$.
    \item $L(\mathbf{u})$: the leftmost quadrant in $\mathcal{Q}(\mathbf{u})$.
    \item $a(\mathbf{u})$: the leftmost point in $S(\mathbf{u})$ that is contained in $U(\mathbf{u})$.
\end{itemize}
Let $\mathbf{u} \in \mathcal{T}$ be a node and $\mathbf{l},\mathbf{r}$ be its left and right children, respectively.
It is clear that $y^-(\mathbf{u})$, $y^+(\mathbf{u})$, $L(\mathbf{u})$ can be computed in $O(1)$ time knowing the $y^-(\cdot)$, $y^+(\cdot)$, and $L(\cdot)$ fields of $\mathbf{l}$ and $\mathbf{r}$.
We claim that $a(\mathbf{u})$ can be computed in $\widetilde{O}(1)$ time based on the information stored at $\mathbf{u}$, $\mathbf{l}$, $\mathbf{r}$, and the 3-sided range-minimum data structure $\mathcal{A}$.
By definition, $a(\mathbf{u})$ is the leftmost point in $S(\mathbf{u})$ that is contained in $U(\mathbf{u})$.
Since $S(\mathbf{u}) = S(\mathbf{l}) \cup S(\mathbf{r})$, it suffices to compute the leftmost point in $S(\mathbf{l})$ contained in $U(\mathbf{u})$ and the leftmost point in $S(\mathbf{r})$ contained in $U(\mathbf{u})$.
Note that any point in $S(\mathbf{r})$ is not contained in $U(\mathbf{l})$, since any quadrant in $\mathcal{Q}(\mathbf{l})$ is ``below'' any point in $S(\mathbf{r})$.
It follows that the leftmost point in $S(\mathbf{r})$ that is contained in $U(\mathbf{u})$ is just $a(\mathbf{r})$.
To compute the leftmost point in $S(\mathbf{l})$ that is contained in $U(\mathbf{u})$, we only need to compute the leftmost point in $S(\mathbf{l})$ contained in $U(\mathbf{l})$ and the leftmost point in $S(\mathbf{l})$ contained in $U(\mathbf{r})$, because $U(\mathbf{u}) = U(\mathbf{l}) \cup U(\mathbf{r})$.
The leftmost point in $S(\mathbf{l})$ contained in $U(\mathbf{l})$ is just $a(\mathbf{l})$.
In order to compute the leftmost point in $S(\mathbf{l})$ contained in $U(\mathbf{r})$, we observe that a point in $S(\mathbf{l})$ is contained in $U(\mathbf{r})$ iff it is contained in the quadrant $L(\mathbf{r})$.
Thus, the leftmost point in $S(\mathbf{l})$ contained in $U(\mathbf{r})$ is just the leftmost point in $S(\mathbf{l})$ contained in $L(\mathbf{r})$.
Note that $S(\mathbf{l})$ is exactly the set of the points in $S$ that lie in the strip $P = \mathbb{R} \times [y^-(\mathbf{l}),y^+(\mathbf{l})]$, i.e., $S(\mathbf{l}) = S \cap P$.
Hence, we can query the data structure $\mathcal{A}$ with the 3-sided rectangle $P \cap L(\mathbf{r})$ to obtain the leftmost point in $S(\mathbf{l})$ contained in $L(\mathbf{r})$, i.e., the leftmost point in $S(\mathbf{l})$ contained in $U(\mathbf{r})$, which takes $\widetilde{O}(1)$ time.
We conclude that, by using the 3-sided range-minimum data structure $\mathcal{A}$, all the fields of a node $\mathbf{u}$ can be computed in $\widetilde{O}(1)$ time based on the information stored at $\mathbf{u}$ and their children.
Therefore, the range tree $\mathcal{T}$ with the augmented fields can be dynamized with $\widetilde{O}(1)$ update time using the standard technique for dynamizing augmented trees~\cite{cormen2009introduction}.
The construction time of $\mathcal{T}$ is clearly $\widetilde{O}(n)$ where $n = |S| + |\mathcal{Q}^\text{SE}|$, thus $\mathcal{T}$ is a basic data structure.

Next, we consider how to use $\mathcal{T}$ to compute $\phi(\tilde{y})$ in $\widetilde{O}(1)$ time for a given $\tilde{y} \in \mathbb{R}$.
We first find the $t = O(\log n)$ canonical nodes $\mathbf{u}_1,\dots,\mathbf{u}_t \in \mathcal{T}$ corresponding to the range $[\widetilde{y},\infty)$.
Suppose $\mathbf{u}_1,\dots,\mathbf{u}_t$ are sorted from left to right in $\mathcal{T}$.
By the property of canonical nodes, the canonical subsets of $\mathbf{u}_1,\dots,\mathbf{u}_t$ are disjoint and their union is the subset of $S \cup V(\mathcal{Q}^\text{SE})$ consisting of the points whose $y$-coordinates are in the range $[\widetilde{y},\infty)$.
The point $\phi(\tilde{y})$ we are looking for is just the leftmost point in $\bigcup_{i=1}^t S(\mathbf{u}_i)$ that is contained in $U^\text{SE}$.
Note that $\phi(\tilde{y})$ is not contained in any southeast quadrant $Q$ with $y(Q) < \tilde{y}$.
Thus, $\phi(\tilde{y})$ is the leftmost point in $\bigcup_{i=1}^t S(\mathbf{u}_i)$ that is contained in $\bigcup_{i=1}^t U(\mathbf{u}_i)$.
To compute $\phi(\tilde{y})$, it suffices to know the leftmost point in $S(\mathbf{u}_i)$ that is contained in $U(\mathbf{u}_j)$ for all $i,j \in \{1,\dots,t\}$.
If $i > j$, then any point in $S(\mathbf{u}_i)$ is not contained in $U(\mathbf{u}_j)$.
If $i = j$, then the leftmost point in $S(\mathbf{u}_i)$ contained in $U(\mathbf{u}_j)$ is just $a(\mathbf{u}_i) = a(\mathbf{u}_j)$.
If $i < j$, then a point in $S(\mathbf{u}_i)$ is contained in $U(\mathbf{u}_j)$ iff it is contained in the quadrant $L(\mathbf{u}_j)$.
Note that the points in $S(\mathbf{u}_i)$ are exactly the set of the points in $S$ that lie in the strip $P_i = \mathbb{R} \times [y^-(\mathbf{u}_i),y^+(\mathbf{u}_i)]$, i.e., $S(\mathbf{u}_i) = S \cap P_i$.
Therefore, the leftmost point in $S(\mathbf{u}_i)$ that is contained in $L(\mathbf{u}_j)$ can be computed in $\widetilde{O}(1)$ time by querying the data structure $\mathcal{A}$ with the 3-sided rectangle $P_i \cap L(\mathbf{u}_j)$.
To summarize, the leftmost point in $S(\mathbf{u}_i)$ that is contained in $U(\mathbf{u}_j)$ can be computed in $\widetilde{O}(1)$ time for any $i,j \in \{1,\dots,t\}$.
Since $t = O(\log n)$, computing the leftmost point in $S(\mathbf{u}_i)$ that is contained in $U(\mathbf{u}_j)$ for all $i,j \in \{1,\dots,t\}$ takes $\widetilde{O}(1)$ time.
Once we have these points, the leftmost one among them is just the leftmost point in $\bigcup_{i=1}^t S(\mathbf{u}_i)$ that is contained in $\bigcup_{i=1}^t U(\mathbf{u}_i)$, i.e., the point $\phi(\tilde{y})$.
We conclude that computing $\phi(\tilde{y})$ can be done in $\widetilde{O}(1)$ time (by properly storing $S$ and $\mathcal{Q}^\text{SE}$ in some basic data structure).



\end{document}